\documentclass[onecolumn, 12pt, journal]{IEEEtran}  % Comment this line out
\IEEEoverridecommandlockouts                              % This command is only
\usepackage[utf8]{inputenc}
\usepackage{amsmath,bm}
\usepackage{amsfonts}
\usepackage{amssymb}
\usepackage{graphicx}
\usepackage{amsfonts}
\usepackage{hyperref}
\usepackage{caption}
\usepackage{subcaption}
\usepackage{multirow}
\usepackage{algorithm}
\usepackage[noend]{algpseudocode}

\usepackage{mathtools}
\usepackage{xcolor}
\usepackage{cite}
\usepackage[normalem]{ulem}

\newtheorem{Theorem}{Theorem}
\newtheorem{Remark}{Remark}
\newtheorem{Assumption}{Assumption}
\usepackage{tikz}
	\usetikzlibrary{shadows,shapes,arrows}
	\tikzstyle{frame} = [draw, -latex]
	\tikzstyle{line} = [draw]
	\tikzstyle{line2} = [draw, dashdotted]
	\tikzstyle{line3} = [draw, dashed]
	\tikzstyle{line3UD} = [draw, dashed]
	\tikzstyle{place} = [circle, draw=black, fill=white, thick, inner sep=2pt, minimum size=1mm]
	\tikzstyle{place2} = [circle, draw=black, fill=black, thick, inner sep=2pt, minimum size=1mm]
	\tikzstyle{placeRed} = [circle, draw=red, fill=red, thick, inner sep=2pt, minimum size=1mm]
	\tikzstyle{vertex} = [circle, draw=black, fill=black, thick, inner sep=2pt, minimum size=1mm]
\usepackage{tkz-euclide}
\usetikzlibrary{shapes.geometric}

\makeatletter
\def\algbackskip{\hskip-\ALG@thistlm}
\makeatother

\usepackage[switch,columnwise]{lineno}
\usepackage{setspace}
\usepackage{enumerate}
\allowdisplaybreaks

\title{\LARGE \bf Distributed Stochastic Model Predictive Control for an Urban Traffic Network}
\author{Viet Hoang Pham$^{1}$ and Hyo-Sung Ahn$^{1}$ 
\thanks{\small $^{1}$School of Mechanical Engineering, Gwangju Institute of Science and Technology, Gwangju, Korea. E-mails: {vietph@gist.ac.kr}; {hyosung@gist.ac.kr.}}
}

\begin{document}
\doublespacing
%%%%%%%%%%%%%%%%%%%%%%%%%%%%%%%%%%%%%%%%%%%%%%%%%%%%%%%%%%%%%%%%%%%%%%%%
%\linenumbers
%\pagewiselinenumbers
%\linenumbers
%\linenumbersep 9pt\relax
\maketitle 
\thispagestyle{empty}
\pagestyle{empty}

%%%%%%%%%%%%%%%%%%%%%%%%%%%%%%%%%%%%%%%%%%%%%%%%%%%%%%%%%%%%%%%%%%%%%%%%%%%%%%%%%%%%%%%%%%
\begin{abstract}
In this paper, we design a stochastic model predictive control (MPC)-based traffic signal control method for urban networks when the uncertainties of the traffic model parameters (including the exogenous traffic flows and the turning ratios of downstream traffic flows) are taken into account.
Considering that the traffic model parameters are random variables with known expectations and variances, the traffic signal control and coordination problem is formulated as a quadratic program with linear and second-order cone constraints.
In order to reduce computational complexity, we suggest a way to decompose the optimization problem corresponding to the whole network into multiple subproblems.
By applying an Alternating Direction Method of Multipliers (ADMM) scheme, the optimal stochastic traffic signal splits are found in a distributed manner.
The effectiveness of the designed control method is validated via some simulations using VISSIM and MATLAB.
\end{abstract}
%%%%%%%%%%%%%%%%%%%%%%%%%%%%%%%%%%%%%%%%%%%%%%%%%%%%%%%%%%%%%%%%%%%%%%%%%%%%%%%%%%%%%%%%%%

%%%%%%%%%%%%%%%%%%%%%%%%%%%%%%%%%%%%%%%%%%%%%%%%%%%%%%%%%%%%%%%%%%%%%%%%%%%%%%%%%%%%%%%%%%
\section{Introduction}
%%%%%%%%%%%%%%%%%%%%%%%%%%%%%%%%%%%%%%%%%%%%%%%%%%%%%%%%%%%%%%%%%%%%%%%%%%%%%%%%%%%%%%%%%%
Controlling traffic signals in junctions is always a crucial request in order to coordinate traffic flows, which enhances the safety and smoothness for the movement of vehicles in urban networks. Specially, the traffic demand has increased dramatically over past decades while the road infrastructures were rarely extended. Various control strategies \cite{MarkosPapageorgiou2003, LeiChen2016} have been developed to utilize radically the existing traffic resource to avoid traffic congestion.
By using on-line traffic data detected by road sensors, many coordinated traffic-responsive methods have been developed in order to obtain the optimal traffic signal settings. Some famous ones can be named as SCOOT \cite{PBHunt1982}, SCATS \cite{PRLowrie1982}, PRODYN \cite{JLFarges1983}, OPAC \cite{NHGartner1983} and RHODES \cite{PMirchandani1998}. However, these methods are not suitable to a network consisting of multiple junctions since they suffer from exponential complexity with either exhaustive search or dynamic programming optimization.
To cope with this issue, Linear-Quadratic Regulator (LQR)-based \cite{ChristinaDiakaki2002, KonstantinosAmpountolas2009} and Model Predictive Control (MPC)-based \cite{KonstantinosAmpountolas2009, KonstantinosAmpountolas2010, ShuLin2011, LucasBarcelosdeOliveira2010, EduardoCamponogara2011, AndyHFChow2020, RRNegenborn2008, SteliosTimotheou2015, ZhaoZhou2017, BaoLinYe2016, PietroGrandinetti2018} methods transformed the traffic signal control problems into easily solvable forms such as constrained quadratic or linear programs.
In these approaches, traffic models based on vehicles conservation law, such as store-and-forward model \cite{ChristinaDiakaki2002, KonstantinosAmpountolas2009, KonstantinosAmpountolas2010}, S-model \cite{ShuLin2011, ZhaoZhou2017} and cell-transmission model \cite{SteliosTimotheou2015, PietroGrandinetti2018}, are widely used to predict future traffic behaviors since they are simple but accurate enough to describe the traffic evolution over time and space.

In MPC-based approach, by using up-to-date traffic measurements and estimations, a constrained optimization problem corresponding to the traffic signal control problem over a finite prediction horizon of the considered urban network is formulated in every control time step. The objective function to be minimized is a combination of some network-related performance indexes (such as vehicle distribution, total time spent, etc) while the constraints, which are derived from the capacities of road links and junctions, are to guarantee the smooth operation of the whole network. By solving the traffic signal control problem, an optimal traffic signal timing plan over some time steps ahead is obtained. But only the elements corresponding to the current traffic signal cycle are implemented for junctions. Then the horizon is shifted by one sample to repeat the overall procedure with new updated traffic data. This repetition enhances the reliability of the computed control decisions.
Compare to LQR methods, MPC-based methods are considered to have a better reaction to variation of traffic dynamics, for example, the re-routing of vehicles \cite{Heydecker2005, AndyHFChow2020}.
In LQR approach, a feedback control law based on traffic dynamics of the considered network is derived to solve an infinite-horizontal and unconstrained optimization problem. However, some limitation constraints such as the capacity of the roads are not considered. Moreover, it is difficult to implement an LQR method in decentralized or distributed manners since the feedback law requires the information of the whole network.

Due to the promising performance and the applicability in real-time implementation, MPC-based traffic signal control approach has attracted a great attention \cite{BaoLinYe2019}.
Depending on the underlying optimization methods to solve optimal control problems, MPC traffic signal control methods can be grouped into centralized \cite{KonstantinosAmpountolas2009, KonstantinosAmpountolas2010, ShuLin2011}, decentralized \cite{LucasBarcelosdeOliveira2010, EduardoCamponogara2011, AndyHFChow2020} and distributed \cite{RRNegenborn2008, SteliosTimotheou2015, ZhaoZhou2017, BaoLinYe2016, PietroGrandinetti2018} ones.
In centralized framework, only one central controller is used to solve traffic signal control problems. However, the computational load may become too huge to complete in short time.
Thus, decentralized and distributed control schemes are proposed in order to reduce the computational complexity and execution time, which makes the application of MPC controllers feasible for a large-scale urban network.
To attain this objective, the whole network is decomposed into multiple smaller subnetworks and each of them is controlled by one controller, i.e., an agent in multiagent system. As a result, the traffic signal control problem corresponding to the overall network are reformulated as the union of multiple smaller subproblems, each corresponds to one subnetwork.
Multiple agents cooperate to find the optimal solution when each agent uses its own information and shares/exchanges information with its neighbors.
The main difference between decentralized and distributed methods is the role of interactions among subnetworks in the reformulated control problems.
In decentralized methods \cite{LucasBarcelosdeOliveira2010, EduardoCamponogara2011, AndyHFChow2020}, these interactions are considered as the known inputs and outputs of subnetworks in each iteration. They are computed by the local estimated solutions of subproblems in the previous iteration. Since the decentralized control problems are not equivalent to the centralized ones, only suboptimal solutions can be achieved in these methods.
Meanwhile, distributed methods aim to design frameworks, in which agents work cooperatively to find the optimal solution of the centralized control problems. By considering the interaction among subnetworks as the coupled constraints among subproblems \cite{RRNegenborn2008, SteliosTimotheou2015, ZhaoZhou2017, BaoLinYe2016, PietroGrandinetti2018}, the distributed control problems could be considered equivalent to the centralized ones.

In \cite{RRNegenborn2008}, Negenborn et. al. use auxiliary problem principle to design a multi-agent model predictive control framework for transportation networks. This method is also applied in two-level hierarchical framework for MPC traffic signal control of urban networks. Although the execution time is reduced significantly in distributed setup, this amount is still long for real-time application because of the slow convergence.
Dual-decomposition based methods are also frequently applied in distributed MPC traffic signal control \cite{BaoLinYe2016, PietroGrandinetti2018}. However, these methods require that the control problems are strongly convex optimization problems.
Meanwhile, Alternating Direction Methods of Multipliers (ADMM) schemes can guarantee the convergence to the optimal solution for general convex optimization problems \cite{StephenBoyd2011, BingshengHe2015, XinxinLi2015, WeiDeng2017}.
In traffic control literature, different ADMM schemes have been applied to design distributed control methods \cite{JackReilly2015, SteliosTimotheou2015}.
Following Asynchronous Distributed ADMM in \cite{ErminWei2013}, the authors in \cite{JackReilly2015} design a distributed optimization method and apply it to decentralized freeway control. However, the designed method of this paper requires that only one coupled constraint is considered in each iteration. This asynchronous scheme makes a distributed implementation difficult and slow down the convergence rate.
Algorithm 1 in \cite{SteliosTimotheou2015} is based on Jacobi ADMM scheme \cite{WeiDeng2017}. Although this algorithm allows agents to work in parallel, the convergence to the optimal solution is not proved theoretically.

Although MPC traffic signal control methods are widely studied, most of them only focus on nominal control problems under the assumption that all traffic model parameters, including exogenous in/out-flows and turning ratios of traffic flows, are predetermined precisely.
There are few works considering the uncertainties in the traffic model parameters.
The authors in \cite{TamasTettamanti2014} develop a constrained mini-max approach to achieve robust optimal signal splits for urban networks when the fluctuations of traffic demands are assumed to be bounded.
In \cite{VietHoangPham2020,VietHoangPham2020TITS}, we assume the uncertainties in exogenous traffic flows and turning ratios are small and added more constraints to guarantee the smooth operation of roads.
In traffic control literature, there are some further approaches that investigate the uncertainties in traffic demands and turning ratios. For isolated junctions, \cite{YafengYin2008} proposes three models to minimize the average delays with variant traffic demands, and \cite{LihuiZhang2010} includes probability information of day-to-day demand variations.
Liu et al. \cite{HaoLiu2022} apply distributionally robust optimization concept \cite{Calafiore2006} and Lax-Hopf solution to propose a traffic control model considering the uncertainties in turning ratios of the downstream traffic flows. In \cite{ZCSu2021}, an adaptive traffic controller combining approximate dynamic programming and model-based optimization is presented for urban networks with stochastic demands.

In this paper, we propose a stochastic MPC traffic signal control method for urban networks when the uncertainties in the traffic model parameters (i.e., exogenous in/out-flows and the turning ratios of the downstream traffic flows corresponding to some future cycles) are taken into account.
Under the assumption that the traffic model parameters are random variables with given expected values and variances, we formulate a stochastic version of the nominal MPC traffic signal control problem.
The optimal traffic signal splits are sought to minimize the expectation of the nominal cost function while guaranteeing the probability that all constraints corresponding to the smooth operations of road links and junctions are satisfied.
By borrowing the result of distributionally robust chance constraint \cite{Calafiore2006}, we reformulate the control problem as a quadratic program with linear and second-order cone constraints.
As the traffic network can be considered as the union of multiple subnetworks, this spatially structural property is utilized to design a distributed method for reducing the executing time in finding the optimal solution.
By using a simple operator splitting technique for inequality constraints and coupled constraints, the distributed control problem has a separated form which is applicable to ADMM methods and is equivalent to the stochastic MPC traffic signal control problem of the whole network.
Then we use the proximal ADMM scheme \cite{BingshengHe2015, XinxinLi2015, WeiDeng2017} to design a distributed method for agents using only their local information.
Moreover, specially formulaic properties of the linear and second-order cone constraints are exploited to reduce the complexity in local computational works. Agents are not required to solve local minimization problems. Instead, they use only simple arithmetic calculations in their update procedures.

Consequently, the main contributions of this paper are as follows:
\begin{itemize}
\item We employ the distributionally robust chance constraint concept in MPC framework to formulate a stochastic MPC traffic signal control problem for an urban network. It is a convex optimization problem which can be solved effectively by different optimization methods. This makes it possible to have a real-time implementation in both centralized and distributed manners. 
\item We design a fully distributed method for solving the stochastic MPC traffic signal control problem. The convergence to the optimal solution is guaranteed in both theoretical proof and simulations. In each iteration of the designed method, agents can work in parallel and use only simple arithmetic calculations. Moreover, the communication is required among neighboring agents.
\end{itemize}

The remainder of this paper is organized as follows.
Section II provides some mathematical backgrounds, which will be used in the analysis in the later sections.
In Section III, we use the store-and-forward model to formulate the nominal MPC traffic signal control problem.
Then its stochastic version is reformulated in Section IV.
In Section V, we apply ADMM to solve the stochastic MPC traffic signal control problem in a distributed manner.
We verify the effectiveness of our proposed method via simulations in Section VI.
Section VII concludes this paper.
%%%%%%%%%%%%%%%%%%%%%%%%%%%%%%%%%%%%%%%%%%%%%%%%%%%%%%%%%%%%%%%%%%%%%%%%%%%%%%%%%%%%%%%%%%

%%%%%%%%%%%%%%%%%%%%%%%%%%%%%%%%%%%%%%%%%%%%%%%%%%%%%%%%%%%%%%%%%%%%%%%%%%%%%%%%%%%%%%%%%%
\subsection*{Notations}
%%%%%%%%%%%%%%%%%%%%%%%%%%%%%%%%%%%%%%%%%%%%%%%%%%%%%%%%%%%%%%%%%%%%%%%%%%%%%%%%%%%%%%%%%%
We use $\mathbb{R}$, $\mathbb{R}_{-}$, $\mathbb{R}^n$ and $\mathbb{R}^{m \times n}$ to denote the set of real numbers, the set of nonpositive real numbers, the set of $n$-dimension real vectors and the set of $m \times n$ matrices, respectively.
For a given set $\mathcal{A}$, $|\mathcal{A}|$ represents the cardinality of this set. 
When the set $\mathcal{A}$ has finite number of vectors, i.e., $\mathcal{A} = \{\textbf{a}_1, \textbf{a}_2, \dots, \textbf{a}_n\}$, we use $col \mathcal{A}$ and $row \mathcal{A}$ to define the column vector and row vector as
\[row \mathcal{A} = row \{\textbf{a}_1, \textbf{a}_2, \dots, \textbf{a}_n\} = [\textbf{a}_1^T, \textbf{a}_2^T, \dots, \textbf{a}_n^T],\]
\[col \mathcal{A} = col \{\textbf{a}_1, \textbf{a}_2, \dots, \textbf{a}_n\} = \left(row \mathcal{A}\right)^T.\]
Let $E\left[X\right]$ and $Var\left[X\right]$ be expected value and variance of a random variable $X$, respectively.
The correlation coefficient between two random variables $X$ and $Y$ is denoted by $CoRel\left[X, Y\right]$.
For a random vector $\textbf{X}$, we use $E\left[\textbf{X}\right]$ and $\boldsymbol{\Sigma}[\textbf{X}]$ to denote its expected value vector and covariance matrix.
%%%%%%%%%%%%%%%%%%%%%%%%%%%%%%%%%%%%%%%%%%%%%%%%%%%%%%%%%%%%%%%%%%%%%%%%%%%%%%%%%%%%%%%%%%

%%%%%%%%%%%%%%%%%%%%%%%%%%%%%%%%%%%%%%%%%%%%%%%%%%%%%%%%%%%%%%%%%%%%%%%%%%%%%%%%%%%%%%%%%%
\section{Preliminaries}
%%%%%%%%%%%%%%%%%%%%%%%%%%%%%%%%%%%%%%%%%%%%%%%%%%%%%%%%%%%%%%%%%%%%%%%%%%%%%%%%%%%%%%%%%%
\subsubsection{Expected value and variance of random variables \cite{RoyDYates2004}}\label{subEV}
%%%%%%%%%%%%%%%%%%%%%%%%%%%%%%%%%%%%%%%%%%%%%%%%%%%%%%%%%%%%%%%%%%%%%%%%%%%%%%%%%%%%%%%%%%
%In this subsection, we provide some basic equations to compute expected values and variances of random variables .
For any random variable $X$, we have
\begin{equation}\label{eq_square_RV}
E[X^2] = \left(E[X]\right)^2 + Var[X]
\end{equation}
A random vector $\textbf{X} = \left[X_1, X_2, \dots, X_n\right]^T$ is a collection of random variables $X_i, i = 1, \dots, n$.
The expected value vector of $\textbf{X}$ is $E[\textbf{X}] = \left[E[X_1], E[X_2], \dots, E[X_n]\right]^T$ and its covariance matrix $\boldsymbol{\Sigma}[\textbf{X}] = \left[\boldsymbol{\Sigma}[\textbf{X}]\right]_{ij}$ is an $n \times n$ matrix whose $ij$-element is defined by
\[\left[\boldsymbol{\Sigma}[\textbf{X}]\right]_{ij} = \left\{\begin{matrix*}[l]
Var[X_i], & i = j,\\
CoRel[X_i, X_j] \sqrt{Var[X_i]Var[X_j]}, & i \neq j,
\end{matrix*}\right.\]

Let $a_i, i = 1, \dots, n$ and $b$ be real coefficients; then the linear combination $Y = \textbf{a}^T \textbf{X} + b = \sum_{j = 1}^{n} a_jX_j + b$ is a random variable with $\textbf{a} = \left[a_1, \dots, a_n\right]^T$.
The expected value of the random variable $Y$ is
\begin{equation}\label{eq_mean_linearcombination}
E[Y] = \textbf{a}^T E\left[\textbf{X}\right] + b = \sum_{j = 1}^{n} a_j E[X_j] + b,
\end{equation}
and its variance is
\begin{align}
Var[Y] = \textbf{a}^T \boldsymbol{\Sigma}[\textbf{X}] \textbf{a} = \sum\limits_{j = 1}^{n} a_j^2 Var[X_j] + \sum\limits_{j = 1}^{n} \sum\limits_{k \neq j} CoRel[X_j, X_k]a_ja_k \sqrt{Var[X_j]Var[X_k]}.\label{eq_variance_linearcombination}
\end{align}
%%%%%%%%%%%%%%%%%%%%%%%%%%%%%%%%%%%%%%%%%%%%%%%%%%%%%%%%%%%%%%%%%%%%%%%%%%%%%%%%%%%%%%%%%%
\subsubsection{Proximal ADMM}\label{subADMM}
%%%%%%%%%%%%%%%%%%%%%%%%%%%%%%%%%%%%%%%%%%%%%%%%%%%%%%%%%%%%%%%%%%%%%%%%%%%%%%%%%%%%%%%%%%
Consider the following convex optimization problem
\begin{equation}\label{eq_ADMMproblem}
\min\limits_{\textbf{x} \in \mathcal{X}, \textbf{y} \in \mathcal{Y}} \textrm{ }\Psi_x(\textbf{x}) + \Psi_y(\textbf{y}) \textrm{ s.t. } \textbf{A}\textbf{x} + \textbf{B}\textbf{y} = \textbf{c}.
\end{equation}
where $\Psi_x(\textbf{x}), \Psi_y(\textbf{y})$ are convex functions, $\mathcal{X}, \mathcal{Y}$ are convex sets, $\textbf{A}, \textbf{B}$ and $\textbf{c}$ are known matrices and vector.
The Lagrangian function of the problem \eqref{eq_ADMMproblem} is given by \[\mathcal{L} = \Psi_x(\textbf{x}) + \Psi_y(\textbf{y}) + \boldsymbol{\lambda}^T\left(\textbf{A}\textbf{x} + \textbf{B}\textbf{y} - \textbf{c}\right)\] where $\boldsymbol{\lambda}$ is the dual variable associated with the equality constraint.
The ADMM scheme for solving \eqref{eq_ADMMproblem} is given by the iteration update \eqref{eq_ADMM} in which $\rho > 0$ is a penalty parameter and $\textbf{G}$ is a symmetric and positive semidefinite matrix.
\begin{subequations}\label{eq_ADMM}
\begin{align}
\textbf{x}(l+1) = \arg\min\limits_{\textbf{x} \in \mathcal{X}} \Bigl\{& \Psi_x(\textbf{x}) + \frac{\rho}{2}\left|\left|\textbf{A}\textbf{x} + \textbf{B}\textbf{y}(l) - \textbf{c} - \frac{1}{\rho}\boldsymbol{\lambda}(l)\right|\right|^2 + \frac{1}{2}\left|\left|\textbf{x} - \textbf{x}(l)\right|\right|_{\textbf{G}}^2 \Bigr\},\\
\textbf{y}(l+1) = \arg\min\limits_{\textbf{y} \in \mathcal{Y}} \Bigl\{& \Psi_y(\textbf{y}) + \frac{\rho}{2}\left|\left|\textbf{A}\textbf{x}(l+1) + \textbf{B}\textbf{y} - \textbf{c} - \frac{1}{\rho}\boldsymbol{\lambda}(l)\right|\right|^2 \Bigr\},\\
\boldsymbol{\lambda}(l+1) = \boldsymbol{\lambda}(l) - \rho \Bigl(& \textbf{A}\textbf{x}(l+1) + \textbf{B}\textbf{y}(l+1) - \textbf{c} \Bigr),
\end{align}
\end{subequations}
In \cite{BingshengHe2015}, B. He et. al. proved the worst-case $O\left(\frac{1}{l}\right)$ convergence rate for the ADMM scheme \eqref{eq_ADMM}. By using the analysis in \cite{WeiDeng2017}, this convergence rate can be refined to $o\left(\frac{1}{l}\right)$.
The following theorem summarizes the results in \cite{BingshengHe2015, XinxinLi2015} for the convergence of ADMM scheme \eqref{eq_ADMM}.
\begin{Theorem}\label{th_ADMM}
Let $\boldsymbol{\Theta} = blkdiag\left\{ \textbf{G}, \rho\textbf{D}_y^T\textbf{D}_y, \frac{1}{\rho}\textbf{I} \right\}$ and $\textbf{v} = [\textbf{x}^T, \textbf{y}^T, \boldsymbol{\lambda}^T]^T$. Then we have
\begin{enumerate}
\item $||\textbf{v}(l) - \textbf{v}^{opt}||_{\boldsymbol{\Theta}}^2 \le \frac{1}{(l+1)}||\textbf{v}(0) - \textbf{v}^{opt}||_{\boldsymbol{\Theta}}^2$ where $\textbf{v}^{opt}$ corresponds to one optimal solution of \eqref{eq_ADMMproblem}.
\item $||\textbf{v}(l) - \textbf{v}(l+1)||_{\boldsymbol{\Theta}}^2 = o\left(\frac{1}{l}\right)$
\item $|| \textbf{A}\textbf{x}(l) + \textbf{B}\textbf{y}(l) - \textbf{c}|| = o\left(\frac{1}{\sqrt{l}}\right)$
\item $|\psi_x(\textbf{x}(l)) + \psi_y(\textbf{y}(l)) - \psi_x(\textbf{x}^{opt}) + \psi_y(\textbf{y}^{opt})| = o\left(\frac{1}{\sqrt{l}}\right)$
\end{enumerate}
\end{Theorem}
The item 1), 2) and 3) \& 4) of Theorem \ref{th_ADMM} are corresponding to Theorem 6.1 in \cite{BingshengHe2015}, Theorem 5.2 and 5.3 in \cite{WeiDeng2017}, respectively.
%%%%%%%%%%%%%%%%%%%%%%%%%%%%%%%%%%%%%%%%%%%%%%%%%%%%%%%%%%%%%%%%%%%%%%%%%%%%%%%%%%%%%%%%%%
\subsubsection{Projection on a Bounded Set and a Second-Order Cone}\label{subProjBC}
%%%%%%%%%%%%%%%%%%%%%%%%%%%%%%%%%%%%%%%%%%%%%%%%%%%%%%%%%%%%%%%%%%%%%%%%%%%%%%%%%%%%%%%%%%
A bounded set $\mathcal{B}$ is described by the set of scalars as $\mathcal{B} = \left\{x \in \mathbb{R}: x_{lb} \le x \le x_{ub}\right\}$ where $x_{lb} \in \mathbb{R}$ and $x_{ub} \in \mathbb{R}$ are respectively called the lower and the upper bounds.
A unit second-order cone of dimension $n+1$ is defined by $\mathcal{C}^{n+1} = \left\{\textbf{x} = \left[\begin{matrix} \textbf{x}_1 \\ x_2 \end{matrix} \right]: \textbf{x}_1 \in \mathbb{R}^n, x_2 \in \mathbb{R}, \left|\left|\textbf{x}_1\right|\right| \le x_2\right\}$.
We know that both $\mathcal{B}$ and $\mathcal{C}^{n+1}$ are convex sets in a Hilbert space.
The following theorem characterizes a property for the projection onto these sets, which will be used in the proof of Theorem \ref{th_updatelaw} (APPENDIX-B).
\begin{Theorem}[Theorem 2.3. in \cite{DavidKinderlehrer1980}]\label{th_proj}
Let $\Omega$ be a closed convex set in a Hilbert space.
Then $\textbf{z} = Prj_{\Omega}(\textbf{x})$, the projection of $\textbf{x}$ onto $\Omega$, if and only if $\textbf{z} \in \Omega: (\textbf{z} - \textbf{x})^T(\textbf{y} - \textbf{z}) \ge 0, \forall \textbf{y} \in \Omega$.
\end{Theorem}

It is straightforward to verify that the projection of any point $y \in \mathbb{R}$ onto the bounded set $\mathcal{B}$ is
\begin{equation}\label{eq_ProjBoundedSet}
Prj_{\mathcal{B}}(y) = \max\left\{x_{lb}, \min\left\{y, x_{ub}\right\}\right\}
\end{equation}
For any point $\textbf{y} = \left[\begin{matrix} \textbf{y}_1^T & y_2 \end{matrix} \right]^T \in \mathbb{R}^{n+1}$ where $\textbf{y}_1 \in \mathbb{R}^n$ and $y_2 \in \mathbb{R}$, its projection onto $\mathcal{C}^{n+1}$, the standard unit second-order cone of dimension $n + 1$, is computed by the equation \eqref{eq_ProjConeSet}. Detailed analysis can be found in \cite{LingchenKong2009}.
\begin{equation}\label{eq_ProjConeSet}
Prj_{\mathcal{C}^{n+1}}(\textbf{y}) = \omega_1(\textbf{y}) \left[\begin{matrix} -\boldsymbol{\theta}(\textbf{y}) \\ 1 \end{matrix} \right] + \omega_2(\textbf{y}) \left[\begin{matrix} \boldsymbol{\theta}(\textbf{y}) \\ 1 \end{matrix} \right]
\end{equation}
where $\omega_1(\textbf{y}) = \frac{1}{2}\max\left\{0, y_2 - \left|\left|\textbf{y}_1\right|\right|_2\right\}$, $\omega_2(\textbf{y}) = \frac{1}{2}\max\left\{0, y_2 + \left|\left|\textbf{y}_1\right|\right|_2\right\}$ and $\boldsymbol{\theta}(\textbf{y}) =\frac{\textbf{y}_1}{\left|\left|\textbf{y}_1\right|\right|_2}$ if $\textbf{y}_1 \neq \textbf{0}$ and is any vector of length 1 in $\mathbb{R}^n$ if $\textbf{y}_1 = \textbf{0}$.

Moreover, we consider a vector $\textbf{x} = \left[\begin{matrix} \textbf{x}_1^T & \textbf{x}_2^T & \cdots & \textbf{x}_n^T \end{matrix} \right]^T$ and the set $\Omega = \Omega_1 \times \Omega_2 \times \cdots \times \Omega_n$ where the dimensions of the vector $\textbf{x}_i$ and the set $\Omega_i$ are equal for all $i = 1, \dots, n$.
The projection of $\textbf{x}$ onto the set $\Omega$ is the vector
\begin{equation}\label{eq_ProjOntoSets}
Prj_{\Omega}(\textbf{x}) = col\left\{ Prj_{\Omega_1}(\textbf{x}_1), Prj_{\Omega_2}(\textbf{x}_2), \cdots, Prj_{\Omega_n}(\textbf{x}_n) \right\}
\end{equation}
%%%%%%%%%%%%%%%%%%%%%%%%%%%%%%%%%%%%%%%%%%%%%%%%%%%%%%%%%%%%%%%%%%%%%%%%%%%%%%%%%%%%%%%%%%
\subsubsection{Distributionally robust chance constraint}\label{subRC}
%%%%%%%%%%%%%%%%%%%%%%%%%%%%%%%%%%%%%%%%%%%%%%%%%%%%%%%%%%%%%%%%%%%%%%%%%%%%%%%%%%%%%%%%%%
Consider the affine function $\hat{\textbf{a}}^T\textbf{x} + \hat{b}$ where $\textbf{x} = \left[x_1, x_2, \dots, x_n\right]^T$ is a variable vector, $\hat{\textbf{a}} = \left[\hat{a}_1, \dots, \hat{a}_n\right]^T$ and $\hat{b}$ are uncertain parameters with given probability distributions.
Assume $\hat{\textbf{f}} = [\hat{\textbf{a}}^T, \hat{b}]^T$ is a random vector in an ambiguity set $\mathcal{F}$ with known expectation $E[\hat{\textbf{f}}] = [E\left[\hat{a}_1\right], \dots, E\left[\hat{a}_n\right], E[\hat{b}]]^T$ and covariance matrix $\boldsymbol{\Sigma}[\hat{\textbf{f}}]$.
As stated in \cite{Calafiore2006}, the constraint $\min_{\hat{\textbf{f}} \in \mathcal{F}} \textrm{P}\left( \hat{\textbf{a}}^T\textbf{x} + \hat{b} \le 0\right) \ge 1 - \epsilon^0$ can be converted into the following inequality
\begin{equation}\label{eq_chacneConstraintsIneq}
\sqrt{\frac{1 - \epsilon^0}{\epsilon^0}\left[\begin{matrix}\textbf{x}^T & 1\end{matrix}\right] \boldsymbol{\Sigma}[\hat{\textbf{f}}] \left[\begin{matrix}\textbf{x}\\ 1\end{matrix}\right]} + \sum_{j = 1}^{n} E\left[\hat{a}_j\right]x_j + E\left[\hat{b}\right] \le 0.
\end{equation}
Because $\boldsymbol{\Sigma}[\hat{\textbf{f}}]$ is a symmetric, positive semidefinite matrix, there always exists the matrix $\textbf{G}_{\hat{\textbf{f}}}$ such that $\boldsymbol{\Sigma}[\hat{\textbf{f}}] = \textbf{G}_{\hat{\textbf{f}}}^T\textbf{G}_{\hat{\textbf{f}}}$.
Then the inequality \eqref{eq_chacneConstraintsIneq} is equivalent to the following condition
\begin{equation}\label{eq_soc_constraint}
\left[\begin{matrix}
\sqrt{\frac{1 - \epsilon^0}{\epsilon^0}} \textbf{G}_{\hat{\textbf{f}}} \left[\begin{matrix}\textbf{x}^T & 1\end{matrix}\right]^T\\
- \sum_{j = 1}^{n} E\left[\hat{a}_j\right]x_j - E\left[\hat{b}\right]
\end{matrix}\right] \in \mathcal{C}^{n + 1}
\end{equation}
where $\mathcal{C}^{n + 1}$ is the unit second-order cone of dimension $n+1$.
%%%%%%%%%%%%%%%%%%%%%%%%%%%%%%%%%%%%%%%%%%%%%%%%%%%%%%%%%%%%%%%%%%%%%%%%%%%%%%%%%%%%%%%%%%

%%%%%%%%%%%%%%%%%%%%%%%%%%%%%%%%%%%%%%%%%%%%%%%%%%%%%%%%%%%%%%%%%%%%%%%%%%%%%%%%%%%%%%%%%%
\section{Model-based predicted traffic signal control}
%%%%%%%%%%%%%%%%%%%%%%%%%%%%%%%%%%%%%%%%%%%%%%%%%%%%%%%%%%%%%%%%%%%%%%%%%%%%%%%%%%%%%%%%%%
In this paper, we employ store-and-forward modeling approach, which is initially proposed by Gazis et al. \cite{DenosCGazis1963} and widely used in urban traffic control research \cite{ChristinaDiakaki2002, KonstantinosAmpountolas2009, KonstantinosAmpountolas2010, BaoLinYe2016}, to formulate MPC traffic signal control problems. Our modified traffic model is based on the analysis in \cite{ChristinaDiakaki2002} and the modifications aim to have a simple graphical illustration for a general urban network.
%%%%%%%%%%%%%%%%%%%%%%%%%%%%%%%%%%%%%%%%%%%%%%%%%%%%%%%%%%%%%%%%%%%%%%%%%%%%%%%%%%%%%%%%%%
\subsection{Graph representation for urban traffic network}
%%%%%%%%%%%%%%%%%%%%%%%%%%%%%%%%%%%%%%%%%%%%%%%%%%%%%%%%%%%%%%%%%%%%%%%%%%%%%%%%%%%%%%%%%%
An urban network is considered as a set of junctions connected by roads. Each road consists of one or more lanes and a junction can be signalized or not. Here, we define a junction as a common crossing area of different roads. The traffic is divided into streams where each of them is defined as a flow of vehicles from one road to another road crossing the corresponding junction. Two streams at one junction are called compatible if they can safely cross the junction simultaneously. In a signalized junction, some sets of compatible streams are given right of way consequently in order to avoid conflict. Such a set is called a signal phase. We define a road link as the set of streams in one road that have the right of way in the same signal phase. In the view of traffic infrastructure, one road link corresponds to one or some specific lanes in the same road since vehicles tend to choose the suitable lanes in the road depending on their turning purpose (left/right/straight) at the junction. Sometimes, this selection is due to traffic rule, e.g., only vehicles in some special lanes are allowed to turn left. 
Thus, an urban network can be represented by a directed graph $\mathcal{T} = (\mathcal{J}, \mathcal{L})$ where $\mathcal{J} = \{J_1, J_2, \dots\}$ is the set of junctions and $\mathcal{L} = \{1, 2, \cdots\}$ is the set of road links.

We use $\sigma(z)$ and $\tau(z)$ to denote its upstream junction and downstream junction for a road link $z \in \mathcal{L}$, respectively.
It means vehicles move from $\sigma(z)$ to $\tau(z)$ in the road link $z$.
Let $\mathcal{O}$ be the set of all external source nodes and destination nodes which are outside of the considered urban network; but they are connected to the network.
So $\sigma(z), \tau(z) \in \mathcal{J} \cup \mathcal{O}$ for all road links $z \in \mathcal{L}$.
Consider two road links $z, w \in \mathcal{L}$, if $\sigma(z) = \tau(w)$ then the road link $z$ is a downstream neighbor of the road link $w$ and the road link $w$ is an upstream neighbor of the road link $z$.
The set $\mathcal{N}_z^-$ of the downstream neighbors and the set $\mathcal{N}_z^+$ of the upstream neighbors corresponding to the road link $z \in \mathcal{L}$ are defined as
\[\mathcal{N}_z^- = \{w \in \mathcal{L}: \sigma(w) = \tau(z)\}, \textrm{ } \mathcal{N}_z^+ = \{w \in \mathcal{L}: \tau(w) = \sigma(z)\}.\]
For a node $J_v \in \mathcal{J} \cup \mathcal{O}$, we use $\mathcal{L}_{J_v}^{in}$ and $\mathcal{L}_{J_v}^{out}$ to denote the sets of its incoming road links and its outgoing road links.
\[\mathcal{L}_{J_v}^{in} = \{z \in \mathcal{L}: \tau(z) = {J_v}\}, \textrm{ } \mathcal{L}_{J_v}^{out} = \{z \in \mathcal{L}: \sigma(z) = {J_v}\}\]
It is easy to see that $\mathcal{N}_z^- \subset \mathcal{L}_{J_v}^{out}$ if $z \in \mathcal{L}_{J_v}^{in}$ and $\mathcal{N}_z^+ \subset \mathcal{L}_{J_v}^{in}$ if $z \in \mathcal{L}_{J_v}^{out}$.
We also define the set of source and destination road links for the urban network by
\[\mathcal{L}^{in} = \{z \in \mathcal{L}: \sigma(z) \in \mathcal{O}\}, \textrm{ } \mathcal{L}^{out} = \{z \in \mathcal{L}: \tau(z) \in \mathcal{O}\}.\]

Consider the left figure in Fig. \ref{fig_exampleUTN} as an illustrative example.
The original network consists of $4$ junctions.
%%%%%%%%%%%%%%%%%%%%%%%%%%%%%%%%%%%%%%%%%%%%%%%%%%%%%%%%%%%%%%%%%%%%%%%%%%%%%%%%%%%%%%%%%%
\begin{figure}[htb]
\centering
\raisebox{0.075\height}{\includegraphics[width=0.28\textwidth]{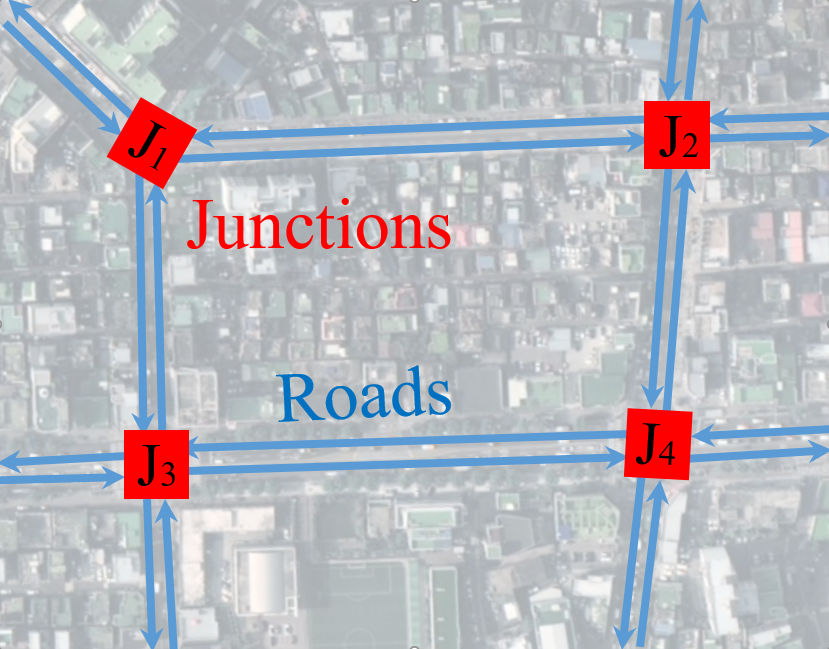}}
%a) Source image from Google Earth Map.\\
\scalebox{0.3}{\begin{tikzpicture}[
textnode/.style={rectangle, fill=white, opacity=0.9, minimum size=1mm, text=blue,text opacity=1,font=\fontsize{25}{25}\selectfont},
internalnode/.style={rectangle, draw=black, fill=white,  very thick, minimum size=30mm, text=blue,text opacity=1,font=\fontsize{45}{45}\selectfont},
boundarynode1/.style={rectangle, draw=white, fill=white,  very thick, minimum width=1mm, minimum height=1mm, text=blue,text opacity=1,font=\fontsize{25}{25}\selectfont},
boundarynode2/.style={rectangle, draw=white, fill=white,  very thick, minimum width=1mm, minimum height=1mm, text=blue,text opacity=1,font=\fontsize{25}{25}\selectfont},
ellipsenodeIn1/.style={ellipse, draw=black, fill=black, minimum height=4mm, minimum width= 2mm},
ellipsenodeIn2/.style={ellipse, draw=black, fill=black, minimum width=4mm, minimum height= 2mm},
ellipsenodeOut1/.style={ellipse, draw=black, fill=white, minimum height=4mm, minimum width= 2mm},
ellipsenodeOut2/.style={ellipse, draw=black, fill=white, minimum width=4mm, minimum height= 2mm},
framenode/.style={rectangle, fill=white, opacity=0.0, minimum size=150mm, text=blue},
]
%Nodes
\node at (2.25,9.5) [framenode] (nFrame) {};
\node at (0,12.5) [internalnode] (n1) {$J_{1}$};
\node at (6.5,12.5) [internalnode] (n2) {$J_{2}$};
\node at (0,6) [internalnode] (n3) {$J_{3}$};
\node at (6.5,6) [internalnode] (n4) {$J_{4}$};
\node at (-4,12.5) [boundarynode2] (nB1) {};
\node at (-4.1,13.3) [ellipsenodeOut1] (nBi1) {};
\node at (-4.1,11.7) [ellipsenodeIn1] (nBo1) {};
\node at (6.5,16.5) [boundarynode1] (nB2) {};
\node at (5.7,16.6) [ellipsenodeIn2] (nBi2) {};
\node at (7.3,16.6) [ellipsenodeOut2] (nBo2) {};
\node at (10.5,12.5) [boundarynode2] (nB3) {};
\node at (10.6,13.3) [ellipsenodeIn1] (nBi3) {};
\node at (10.6,11.7) [ellipsenodeOut1] (nBo3) {};
\node at (-4,6) [boundarynode2] (nB4) {};
\node at (-4.1,6.8) [ellipsenodeOut1] (nBi4) {};
\node at (-4.1,5.6) [ellipsenodeIn1] (nBo4_1) {};
\node at (-4.1,4.8) [ellipsenodeIn1] (nBo4_2) {};
\node at (10.5,6) [boundarynode2] (nB5) {};
\node at (10.6,7.2) [ellipsenodeIn1] (nBi5_1) {};
\node at (10.6,6.4) [ellipsenodeIn1] (nBi5_2) {};
\node at (10.6,5.2) [ellipsenodeOut1] (nBo5) {};
\node at (0,2.0) [boundarynode1] (nB6) {};
\node at (-0.8,1.9) [ellipsenodeOut2] (nBo6) {};
\node at (0.4,1.9) [ellipsenodeIn2] (nBi6_1) {};
\node at (1.2,1.9) [ellipsenodeIn2] (nBi6_2) {};
\node at (6.5,2.0) [boundarynode1] (nB7) {};
\node at (6.9,1.9) [ellipsenodeIn2] (nBi7_1) {};
\node at (7.7,1.9) [ellipsenodeIn2] (nBi7_2) {};
\node at (5.7,1.9) [ellipsenodeOut2] (nBo7) {};
%Arrows
\draw[->,{line width=3pt},black!40, transform canvas={xshift=-8mm}] (nB2.south)--(nB2.south|-n2.north);
\draw[<-,{line width=3pt},black!40, transform canvas={xshift=8mm}] (nB2.south)--(nB2.south|-n2.north);

\draw[<-,{line width=3pt},black!40, transform canvas={yshift=8mm}] (nB1.east)->(n1.west);
\draw[->,{line width=3pt},black!40, transform canvas={yshift=-8mm}] (nB1.east)->(n1.west);
\draw[<-,{line width=3pt},black!40, transform canvas={yshift=4mm}] (n1.east)->(n2.west);
\draw[<-,{line width=3pt},black!40, transform canvas={yshift=12mm}] (n1.east)->(n2.west);
\draw[->,{line width=3pt},black!40, transform canvas={yshift=-8mm}] (n1.east)->(n2.west);
\draw[<-,{line width=3pt},black!40, transform canvas={yshift=8mm}] (n2.east)->(nB3.west);
\draw[->,{line width=3pt},black!40, transform canvas={yshift=-8mm}] (n2.east)->(nB3.west);

\draw[->,{line width=3pt},black!40, transform canvas={xshift=-12mm}] (n1.south)--(n1.south|-n3.north);
\draw[->,{line width=3pt},black!40, transform canvas={xshift=-4mm}] (n1.south)--(n1.south|-n3.north);
\draw[<-,{line width=3pt},black!40, transform canvas={xshift=8mm}] (n1.south)--(n1.south|-n3.north);
\draw[->,{line width=3pt},black!40, transform canvas={xshift=-12mm}] (n2.south)--(n2.south|-n4.north);
\draw[->,{line width=3pt},black!40, transform canvas={xshift=-4mm}] (n2.south)--(n2.south|-n4.north);
\draw[<-,{line width=3pt},black!40, transform canvas={xshift=8mm}] (n2.south)--(n2.south|-n4.north);

\draw[<-,{line width=3pt},black!40, transform canvas={yshift=8mm}] (nB4.east)->(n3.west);
\draw[->,{line width=3pt},black!40, transform canvas={yshift=-4mm}] (nB4.east)->(n3.west);
\draw[->,{line width=3pt},black!40, transform canvas={yshift=-12mm}] (nB4.east)->(n3.west);
\draw[<-,{line width=3pt},black!40, transform canvas={yshift=12mm}] (n3.east)->(n4.west);
\draw[<-,{line width=3pt},black!40, transform canvas={yshift=4mm}] (n3.east)->(n4.west);
\draw[->,{line width=3pt},black!40, transform canvas={yshift=-4mm}] (n3.east)->(n4.west);
\draw[->,{line width=3pt},black!40, transform canvas={yshift=-12mm}] (n3.east)->(n4.west);
\draw[<-,{line width=3pt},black!40, transform canvas={yshift=12mm}] (n4.east)->(nB5.west);
\draw[<-,{line width=3pt},black!40, transform canvas={yshift=4mm}] (n4.east)->(nB5.west);
\draw[->,{line width=3pt},black!40, transform canvas={yshift=-8mm}] (n4.east)->(nB5.west);

\draw[->,{line width=3pt},black!40, transform canvas={xshift=-8mm}] (n3.south)--(n3.south|-nB6.north);
\draw[<-,{line width=3pt},black!40, transform canvas={xshift=4mm}] (n3.south)--(n3.south|-nB6.north);
\draw[<-,{line width=3pt},black!40, transform canvas={xshift=12mm}] (n3.south)--(n3.south|-nB6.north);
\draw[->,{line width=3pt},black!40, transform canvas={xshift=-8mm}] (n4.south)--(n4.south|-nB7.north);
\draw[<-,{line width=3pt},black!40, transform canvas={xshift=4mm}] (n4.south)--(n4.south|-nB7.north);
\draw[<-,{line width=3pt},black!40, transform canvas={xshift=12mm}] (n4.south)--(n4.south|-nB7.north);
%Links
\node at (5.7,15.35) [textnode] (nL1) {$1$};
\node at (7.3,15.35) [textnode] (nL2) {$2$};

\node at (-2.6,13.25) [textnode] (nL3) {$3$};
\node at (-2.6,11.75) [textnode] (nL4) {$4$};

\node at (3,13.7) [textnode] (nL5) {$5$};
\node at (4,12.9) [textnode] (nL6) {$6$};
\node at (3,11.7) [textnode] (nL7) {$7$};

\node at (9,13.25) [textnode] (nL8) {$8$};
\node at (9,11.75) [textnode] (nL9) {$9$};

\node at (-1.15,10) [textnode] (nL10) {$10$};
\node at (-0.3,8.75) [textnode] (nL11) {$11$};
\node at (0.75,10) [textnode] (nL12) {$12$};

\node at (5.25,10) [textnode] (nL13) {$13$};
\node at (6.15,8.75) [textnode] (nL14) {$14$};
\node at (7.25,10) [textnode] (nL15) {$15$};

\node at (-2.5,6.85) [textnode] (nL6) {$16$};
\node at (-2.8,5.7) [textnode] (nL17) {$17$};
\node at (-2.8,4.8) [textnode] (nL18) {$18$};

\node at (2.8,7.25) [textnode] (nL19) {$19$};
\node at (3.95,6.4) [textnode] (nL20) {$20$};
\node at (2.7,5.6) [textnode] (nL21) {$21$};
\node at (3.7,4.75) [textnode] (nL22) {$22$};

\node at (9.2,7.25) [textnode] (nL23) {$23$};
\node at (9.2,6.35) [textnode] (nL24) {$24$};
\node at (9.1,5.15) [textnode] (nL25) {$25$};

\node at (-0.75,3.3) [textnode] (nL26) {$26$};
\node at (0.26,3.2) [textnode] (nL27) {$27$};
\node at (1.25,3.2) [textnode] (nL28) {$28$};

\node at (5.6,3.3) [textnode] (nL29) {$29$};
\node at (6.75,3.2) [textnode] (nL30) {$30$};
\node at (7.7,3.2) [textnode] (nL31) {$31$};
\end{tikzpicture}}
%b) Graph representation.
\caption{An example of urban network consisting of $4$ junctions. The left figure is source image from Google Earth Map and the right figure is the graph representation.} \label{fig_exampleUTN}
\end{figure}
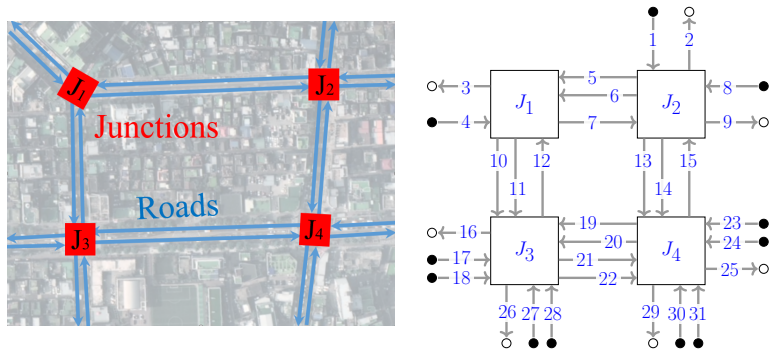
%%%%%%%%%%%%%%%%%%%%%%%%%%%%%%%%%%%%%%%%%%%%%%%%%%%%%%%%%%%%%%%%%%%%%%%%%%%%%%%%%%%%%%%%%%
\begin{table}
\begin{center}
\begin{minipage}{.6\linewidth}
\captionof{table}{Sequences of traffic signal phases.}\label{tbl_sequencePhases}
\scalebox{0.8}
{\begin{tabular}{c|c|c|c|c|c}
\hline
Type & Junction & $1^{st}$ phase & $2^{nd}$ phase & $3^{rd}$ phase & $4^{th}$ phase\\
\hline
Type 1 & \rule{0pt}{15pt} $J_1$ & \includegraphics[width=0.035\textwidth, angle=90]{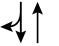} & \includegraphics[width=0.035\textwidth]{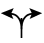}  & \includegraphics[width=0.035\textwidth, angle=90]{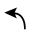} & none \\
\hline
Type 2 & \rule{0pt}{15pt} $J_2$ & \includegraphics[width=0.035\textwidth]{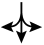} & \includegraphics[width=0.035\textwidth, angle=90]{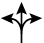} & \includegraphics[width=0.035\textwidth]{4_phasetype4.png} & \includegraphics[width=0.035\textwidth, angle=90]{4_phasetype3.png}\\
\hline
Type 3 & \rule{0pt}{15pt} $J_3, J_4$ & \includegraphics[width=0.035\textwidth]{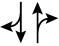} & \includegraphics[width=0.035\textwidth]{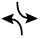}  & \includegraphics[width=0.035\textwidth, angle=90]{4_phasetype1.png} & \includegraphics[width=0.035\textwidth, angle=90]{4_phasetype2.png} \\
\hline
\end{tabular}}
\end{minipage}
\begin{minipage}{.35\linewidth}
\captionof{table}{Active road links.}\label{tbl_LinkPhases}
\scalebox{0.68}
{\begin{tabular}{c|c|c|c|c}
\hline
Junction & $1^{st}$ phase & $2^{nd}$ phase & $3^{rd}$ phase & $4^{th}$ phase\\
\hline
\rule{0pt}{15pt} $J_1$ & $4, 5$ & $12$ & $6$ & none \\
\hline
\rule{0pt}{15pt} $J_2$ & $1$ & $8$ & $15$ & $7$\\
\hline
\rule{0pt}{15pt} $J_3$ & $10, 28$ & $11, 27$  & $18, 19$ & $17, 20$ \\
\hline
\rule{0pt}{15pt} $J_4$ & $13, 31$ & $14, 30$ & $22, 23$ & $21, 24$ \\
\hline
\end{tabular}}
\end{minipage}
\end{center}
\end{table}
%%%%%%%%%%%%%%%%%%%%%%%%%%%%%%%%%%%%%%%%%%%%%%%%%%%%%%%%%%%%%%%%%%%%%%%%%%%%%%%%%%%%%%%%%%
Assume that the sequences of traffic signal phases corresponding to internal junctions are given in TABLE. \ref{tbl_sequencePhases}.
Then we have the right figure as the corresponding graph representation of this urban network.
In this figure, the gray arrows illustrate road links, the squares represent internal junctions and the black/white ellipses represent external source/destination junctions in $\mathcal{O}$.
For this network, we have $\mathcal{J} = \{J_{1}, J_{2}, J_{3}, J_{4}\}$ and $\mathcal{L} = \{1, 2, \dots, 31\}$.
TABLE. \ref{tbl_LinkPhases} indicates the road links corresponding to the signal phases of the internal junctions.
The road links $6, 11, 14, 17, 20, 21, 24, 27$ and $30$ are reserved for only turning left of their corresponding roads.
The set of source road links is $\mathcal{L}^{in} = \left\{1, 4, 8, 17, 18, 23, 24, 27, 28, 30, 31\right\}$ and the set of destination road links is $\mathcal{L}^{out} = \left\{2, 3, 9, 16, 25, 26, 29\right\}$.
Consider the junction $J_3$, the sets of its incoming road links and its outgoing road links are $\mathcal{L}_{J_3}^{in} = \{10, 11, 17, 18, 19, 20, 27, 28\}$ and $\mathcal{L}_{J_3}^{out} = \{12, 16, 21, 22, 26\}$, respectively.
For the road links $6, 12, 19$ and $20$, we have $\sigma(6) = J_2, \tau(6) = J_1$, $\sigma(12) = J_3, \tau(12) = J_1$, $\sigma(19) = J_4, \tau(19) = J_3$ and $\sigma(20) = J_4, \tau(20) = J_3$.
The set of their upstream neighbors and downstream neighbors are given as $\mathcal{N}_{6}^+ = \{1, 8, 15\}$, $\mathcal{N}_{6}^- = \{10, 11\}$, $\mathcal{N}_{12}^+ = \{17, 19, 28\}$, $\mathcal{N}_{12}^- = \{3, 7\}$, $\mathcal{N}_{19}^+ = \{13, 23, 30\}$, $\mathcal{N}_{19}^- = \{12, 16\}$ and $\mathcal{N}_{20}^+ = \{13, 23, 30\}$, $\mathcal{N}_{20}^- = \{26\}$.
%%%%%%%%%%%%%%%%%%%%%%%%%%%%%%%%%%%%%%%%%%%%%%%%%%%%%%%%%%%%%%%%%%%%%%%%%%%%%%%%%%%%%%%%%%
\subsection{Prediction model}
%%%%%%%%%%%%%%%%%%%%%%%%%%%%%%%%%%%%%%%%%%%%%%%%%%%%%%%%%%%%%%%%%%%%%%%%%%%%%%%%%%%%%%%%%%
In store-and-forward modeling approach, the numbers of vehicles in the road links are considered as the traffic states and the downstream traffic flows are control inputs.
Let $t$ be the current time and $T$ is the control time interval.
\begin{equation}\label{eq_predictedTrafficState}
n_{z}(t+k+1|t) = n_{z}(t) + \sum_{l = 0}^{k} \Bigl( \underbrace{d_{z}(t+l|t) - s_{z}(t+l|t)}_{:= e_z(t+l|t)}  + \sum_{w \in \mathcal{N}_z^+}r_{wz}(t+l|t)q_{w}(t+l|t) - q_{z}(t+l|t) \Bigr)
\end{equation}
Consider the road link $z\in \mathcal{L}$ connecting two junctions $J_u$ and $J_v$ (see Fig. 2). Assume that the number $n_{z}(t)$ of vehicles contained in the road link $z$ at time $tT$ is known. Its corresponding predicted traffic states $n_{z}(t + k + 1|t)$ at time $(t + k + 1)T$, for $ k = 0, 1, \dots$, are given by the conservation law equation \eqref{eq_predictedTrafficState}. In this equation, the predicted downstream traffic flow $q_{z}(t+l|t)$ is defined as the number of vehicles leaving the road link ${z}$ through its downstream junction $\tau(z)$, $d_{z}(t+l|t)$ is the estimated number of vehicles entering to the road link $z$ but not from other roads in the considered urban network, $s_{z}(t+l|t)$ is the estimated number of vehicles exiting the urban network from the road link $z$ but not through the downstream node $\tau(z)$ and the estimated turning ratio $r_{wz}(t+l|t)$ corresponds to the percentage of the downstream traffic flow from the road link $w$ to its downstream neighbor $z \in \mathcal{N}_w^-$; all these traffic flows and turning rates are defined in the time interval $[(t+l)T, (t+l+1)T]$.
The downstream traffic flows corresponding to the upstream neighbors of a road link is called its upstream traffic flows.
The exogenous traffic flows of the road link $z$ are from/toward a source/destination node (only if $\sigma(z) \in \mathcal{O}$/$\tau(z) \in \mathcal{O}$) or locations in the road link $z$ (e.g., building, parking lot, etc.).
For easy notation, we define $e_z(t + l|t) = d_z(t + l|t) - s_z(t + k|t)$ as the difference of exogenous flows corresponding to the road link $z \in \mathcal{L}$.
For the turning rates of every road link $z$, where $\tau(z) \in \mathcal{J}$, it is clear that $r_{zw}(t+k|t) \ge 0$ and $\sum_{w \in \mathcal{N}_z^-}r_{zw}(t+k|t) = 1$ for all $k \ge 0$.
%%%%%%%%%%%%%%%%%%%%%%%%%%%%%%%%%%%%%%%%%%%%%%%%%%%%%%%%%%%%%%%%%%%%%%%%%%%%%%%%%%%%%%%%%%
\begin{figure}[htb]
\centering
\scalebox{0.86}{\begin{tikzpicture}
\draw[-] (1.25,1.25) -- (0.5,1.25); \draw[-] (1.25,1.25) -- (1.25,0.5);
\draw[-] (1.25,2.75) -- (0.5,2.75); \draw[-] (1.25,2.75) -- (1.25,3.5);
\draw[-] (2.75,1.25) -- (2.75,0.5); \draw[-] (2.75,1.25) -- (5.25,1.25); \draw[-] (5.25,1.25) -- (5.25,0.5);
\draw[-] (2.75,2.75) -- (2.75,3.5); \draw[-] (2.75,2.75) -- (5.25,2.75); \draw[-] (5.25,2.75) -- (5.25,3.5);
\draw[-] (6.75,1.25) -- (6.75,0.5); \draw[-] (6.75,1.25) -- (7.5,1.25);
\draw[-] (6.75,2.75) -- (6.75,3.5); \draw[-] (6.75,2.75) -- (7.5,2.75);
\draw[dashed] (0.5,2) -- (1.5,2); \draw[dashed] (2.5,2) -- (5.5,2); \draw[dashed] (6.5,2) -- (7.5,2);
\draw[dashed] (2,0.5) -- (2,1.5); \draw[dashed] (2,2.5) -- (2,3.5);
\draw[dashed] (6,0.5) -- (6,1.5); \draw[dashed] (6,2.5) -- (6,3.5);
\node at (2,2.1) [] (node_i) {\large $J_u$};
\node at (6,2.1) [] (node_j) {\large $J_v$};
\node at (4,1.75) [] (link_z) {\large link $z$};
\node at (1.65,3.25) [] (link_w1) {\large $w_1$};
\node at (0.65,1.65) [] (link_w2) {\large $w_2$};
\node at (2.35,0.65) [] (link_w3) {\large $w_3$};
\node at (6.35,3.25) [] (link_w4) {\large $w_4$};
\node at (7.35,1.65) [] (link_w5) {\large $w_5$};
\node at (5.65,0.65) [] (link_w6) {\large $w_6$};
\draw[-,red] (1.75,3) -- (1.75,1.8); \draw[->,red] (1.75,1.8) -- (3,1.8);
\draw[->,red] (1,1.6) -- (3,1.6);
\draw[-,red] (2.25,0.95) -- (2.25,1.4); \draw[->,red] (2.25,1.4) -- (3,1.4);
\draw[->,blue] (5,1.6) -- (7,1.6);
\draw[->,blue] (5.75,1.6) -- (5.75,0.95);
\draw[->,blue] (6.25,1.6) -- (6.25,3);
\draw[->,green] (3.45,0.75) -- (3.45,1.45);
\draw[->,green] (4.65,1.45) -- (4.65,0.75);
\node at (3.2,1) [] (dz) {\large $d_z$};
\node at (4.85,1) [] (sz) {\large $s_z$};
\end{tikzpicture}}
\caption{Traffic flows of the road link $z \in \mathcal{L}$ connecting two junctions $J_u, J_v$, i.e., $\sigma(z) = J_u$, $\tau(z) = J_v$, and $\mathcal{N}_z^+ = \{w_1, w_2, w_3\}$ and $\mathcal{N}_z^- = \{w_4, w_5, w_6\}$. In this figure, the blue arrows correspond to downstream traffic flows, the green arrows are exogenous traffic flows of the road link $z$ and the red arrows are the downstream traffic flows of the upstream neighbors entering the road link $z$.}\label{fig_model}
\end{figure}
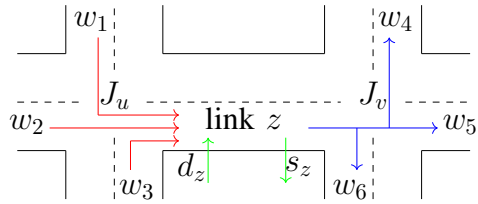
%%%%%%%%%%%%%%%%%%%%%%%%%%%%%%%%%%%%%%%%%%%%%%%%%%%%%%%%%%%%%%%%%%%%%%%%%%%%%%%%%%%%%%%%%%

As the definition, the downstream traffic flow $q_{z}(t+k|t)$ can be controlled by regulating traffic signals in the junction $\tau(z) \in \mathcal{J}$, while the exogenous traffic flows $d_{z}(t+k|t)$, $s_{z}(t+k|t)$ and the turning rates $r_{wz}(t+k|t), \forall w \in \mathcal{N}_z^+,$ can be measured/estimated but not controlled, $\forall k \ge 0$.
In store-and-forward modeling approach \cite{DenosCGazis1963, ChristinaDiakaki2002}, the downstream traffic flow of the road link $z \in \mathcal{L}$ is determined as $q_{z}(t+k|t) = S_{z} g_z(t+k|t)$ where $S_{z}$ is the saturation flow of the road link $z$ and $g_z(t+k|t)$ is the green time length assigned to the road link $z$, where $\tau(z) \in \mathcal{J}$, in the time interval $[(t+k)T,(t+k+1)T]$.
This implicitly assumes that the green time intervals are always fully utilized, i.e., there are enough vehicles awaiting in the road link $z \in \mathcal{L}$ in its green time duration.
To guarantee this requirement, in this paper, we assume that the time interval $T$ is chosen to be small enough and add the following constraint. 
\begin{equation}\label{eq_predictedTrafficFlow}
q_z(t+k|t) \le n_z(t+k|t) + e_z(t+k|t),
\end{equation}
i.e., the downstream traffic flow departing from the road link $z \in \mathcal{L}$ in the time interval $[(t+k)T, (t+k+1)T]$ is less than the number of the vehicles in this road link at the time $(t+k)T$ plus the difference of its exogenous traffic flows in the time interval $[(t+k)T, (t+k+1)T]$.
In addition, the upstream traffic flows entering into the road link $z$ from its upstream neighbors are necessary to satisfy the following constraint for avoiding the traffic congestion
\begin{equation}\label{eq_upstream_limit}
\sum\limits_{w \in \mathcal{N}_z^+}r_{wz}(t+k|t)q_{w}(t+k|t) \le \overline{n_z} - n_z(t+k|t) - e_z(t+k|t)
\end{equation}
where $\overline{n_z}$ is the maximum admissible number of vehicles in the road link $z \in \mathcal{L}$.
%In this paper, we assume that the time interval $T$ is chosen to be small enough to avoid the waste of green time.
Two constraints \eqref{eq_predictedTrafficFlow} and \eqref{eq_upstream_limit} are well studied in cell transmission model (CTM).

The signal timing plan of a signalized junction $J_v \in \mathcal{J}$ is determined by a sequence of traffic signal phases and their assigned signal splits.
%One traffic signal phase in a junction consists of non-conflict directions of vehicle movements from its incoming road links to its outgoing road links.
We use $\mathcal{P}_{J_v}$ to denote the set of traffic signal phases of the junction $J_v \in \mathcal{J}$.
Denote by $g_{p}(t+k|t)$ the split which is the green time assigned to the phase $p \in \mathcal{P}_{J_v}$ in the time interval $[(t+k)T,(t+k+1)T]$.
In order to enhance synchronization of network operation, let all junctions in $\mathcal{J}$ have a common cycle which is equal to control time interval $T$.
The sequence of traffic signal phases in every junction is assumed to be predetermined suitably with its own structure.
The lost time $L_{J_v}$ is defined as the sum of time lengths for clearing vehicles between consecutive traffic signal phases in a junction $J_v \in \mathcal{J}$.
Then the following constraint is required to be satisfied
\begin{equation}\label{eq_signal_limit1}
\sum_{p \in \mathcal{P}_{J_v}} g_{p}(t+k|t) \le T - L_{J_v} := C_{J_v}.
\end{equation}
In addition, the traffic signal split of one phase is usually restricted in a given limited range as
\begin{equation}\label{eq_signal_limit2}
0 \le g_p(t+k|t) \le \overline{g_p} \le T, \forall p \in \mathcal{P}_{J_v}, \forall J_v \in \mathcal{J}.
\end{equation}
Let $\mathcal{P}_{z} \subsetneq \mathcal{P}_{J_v}$, where $\tau(z) = J_v \in \mathcal{J}$, be the set of traffic signal phases which give the right of way to the road link $z$.
Then we have
\begin{equation}\label{eq_signal_roadlink}
0 \le q_z(t+k|t) \le S_z\sum_{p \in \mathcal{P}_{z}} g_{p}(t+k|t), \forall z: \tau(z) \in \mathcal{J}.
\end{equation}
Since there may be more than one road links which can be activated in one traffic signal phase, we distinguish the green time $g_z(t+k|t) = \frac{1}{S_z}q_z(t+k|t)$ assigned to the road link $z: \tau(z) \in \mathcal{J}$ from the splits of phases in $\mathcal{P}_z$.
Then even if two road links $w, z \in \mathcal{L}_{J_v}^{in}$ have the same set of the assigned traffic signal phases, i.e., $\mathcal{P}_{z} \equiv \mathcal{P}_{w}$, their assigned green time lengths can be different. This setup enhances the flexibility of the traffic model. 
For a destination road link $z$ where $\tau(z) \in \mathcal{O}$, we assume that there is an upper limitation for its downstream traffic flow in each time interval.
\begin{equation}\label{eq_outstream_roadlink}
0 \le q_z(t+k|t) \le q_z^{max}(t+k|t), \forall z: \tau(z) \in \mathcal{O},
\end{equation}
where $q_z^{max}(t+k|t)$ is given for all $k \ge 0$.
%%%%%%%%%%%%%%%%%%%%%%%%%%%%%%%%%%%%%%%%%%%%%%%%%%%%%%%%%%%%%%%%%%%%%%%%%%%%%%%%%%%%%%%%%%
\subsection{Model predictive traffic signal control problem}
%%%%%%%%%%%%%%%%%%%%%%%%%%%%%%%%%%%%%%%%%%%%%%%%%%%%%%%%%%%%%%%%%%%%%%%%%%%%%%%%%%%%%%%%%%
In MPC-based traffic signal control approach for an urban network, the key concept is to employ the current traffic states and estimated traffic model parameters in determining an optimal coordination of traffic signals among junctions. The following finite horizon optimal control problem \eqref{eq_problemTP} is formulated and solved in every control time step $t$.
\begin{equation}\label{eq_problemTP}
\begin{matrix*}[l]
\min\limits_{n_z, q_{z}, g_{p}} & \Phi(t) = \sum\limits_{k = 0}^{K - 1} \sum\limits_{z \in \mathcal{L}} \Phi_z(t + k|t)\\
\textrm{s.t. } & \eqref{eq_predictedTrafficState}, (\ref{eq_predictedTrafficFlow}-\ref{eq_outstream_roadlink}), \forall z \in \mathcal{L}, J_v \in \mathcal{J}, 0 \le k \le K - 1.
\end{matrix*}
\end{equation}
where $\Phi_z(t + k|t)$ reflects some traffic control performance indexes of the road link $z \in \mathcal{L}$ in the $(t+k)$-th cycle and the number $K$ is called the predictive horizon time.
As mentioned in the previous subsection, the equality constraint \eqref{eq_predictedTrafficState} is used to predict future traffic states of the road links and the inequalities constraints \eqref{eq_predictedTrafficFlow}-\eqref{eq_outstream_roadlink} are required to guarantee a smooth operation of the traffic network.
The constraint \eqref{eq_upstream_limit}, which is considered as safety constraint of the road link, is to guarantee every road link avoids traffic congestion. The inequalities \eqref{eq_signal_limit1} and \eqref{eq_signal_limit2} correspond to hard constraints on the splits of traffic signal phases, which imply the compatibility of traffic streams crossing a signalized junction.
Meanwhile, the constraints \eqref{eq_signal_roadlink} and \eqref{eq_outstream_roadlink} represent the upper bounds of the downstream traffic flows of road links and the constraint \eqref{eq_predictedTrafficFlow} is to avoid the waste of the green time.
In \eqref{eq_problemTP}, $n_z(t+k+1|t), q_z(t+k+1|t), \forall z \in \mathcal{L}$ and $g_p(t+k|t), \forall p \in \mathcal{P}_{J_v}, J_v \in \mathcal{J}, \forall k = 0, \dots, K-1$, are considered as variables.

By solving the optimization problem \eqref{eq_problemTP}, we obtain a signal timing plan in the next $K$ cycles; which is supposed to minimize the risk of traffic congestion and to satisfy the constraints on the capacities of all road links in $\mathcal{L}$ and all junctions in $\mathcal{J}$.
The bigger $K$ may provide a better control performance, but it requires more computation load.
It is usual to choose $K > 1$ in order to avoid myopic control scheme.
However, only obtained traffic signal splits corresponding to the current time step, i.e., $g_z(t+0|t) = \frac{1}{S_z}q_z(t+0|t)$ for all $z \in \mathcal{L}$ and $g_p(t+0|t)$ for all $p \in \mathcal{P}_{J_v}, J_v \in \mathcal{J}$, are implemented.
Then the overall process is repeated again in next control time step with new updated traffic states and estimated traffic model parameters.
The repetition of MPC-based approach enhances the reliability of the computed control decisions.

To exploit the capacity of existing transportation infrastructure, the following three performance indexes are widely considered to determine the optimal traffic signal splits for near future cycles in MPC traffic signal control approach:
\begin{align}
\Phi^{(1)}(t) &= \sum_{k = 0}^{K - 1} \sum_{z \in \mathcal{L}} \frac{1}{\overline{n_z}}\Bigl(n_z(t + k +1|t)\Bigr)^2 \label{eq_cost1}\\
\Phi^{(2)}(t) &= \sum_{k = 0}^{K - 1} \sum_{z \in \mathcal{L}} Tn_z(t + k + 1|t) \label{eq_cost2}\\
\Phi^{(3)}(t) &= \sum_{k = 0}^{K - 1} \sum_{z \in \mathcal{L}} \left(n_z(t + k|t) - q_z(t+k|t) \right) \label{eq_cost3}
\end{align}
$\Phi^{(1)}(t)$ is applied to balance the relative occupation of the road links; the objective of $\Phi^{(2)}(t)$ is to reduce the total time of vehicles spent in the urban network; and the quantity $n_z(t + k|t) - q_z(t+k|t)$ in $\Phi^{(3)}(t)$ corresponds to the number of inactive vehicles of the road link $z \in \mathcal{L}$ in $[(t+k)T, (t+k+1)T]$.
In this paper, the cost function of the overall urban network is considered to have the following quadratic form
\begin{align}
\Phi(t) = \sum\limits_{k = 0}^{K - 1} \sum\limits_{z \in \mathcal{L}} & \biggl\{ \alpha_{z,k} (n_z(t + k + 1|t))^2 + \beta_{z,k} n_z(t + k + 1|t) - \gamma_{z,k} q_z(t+k|t) \biggr\}\label{eq_costfunction}
\end{align}
where $\alpha_{z,k}, \beta_{z,k}$ and $\gamma_{z,k}$ are given positive constants.
The equation \eqref{eq_costfunction} is a generalized cost function corresponding to the weighted sum of some performance indexes.
%%%%%%%%%%%%%%%%%%%%%%%%%%%%%%%%%%%%%%%%%%%%%%%%%%%%%%%%%%%%%%%%%%%%%%%%%%%%%%%%%%%%%%%%%%

%%%%%%%%%%%%%%%%%%%%%%%%%%%%%%%%%%%%%%%%%%%%%%%%%%%%%%%%%%%%%%%%%%%%%%%%%%%%%%%%%%%%%%%%%%
\section{Reformulation for Stochastic MPC approach}
%%%%%%%%%%%%%%%%%%%%%%%%%%%%%%%%%%%%%%%%%%%%%%%%%%%%%%%%%%%%%%%%%%%%%%%%%%%%%%%%%%%%%%%%%%
\subsection{Assumption on uncertainties of traffic model parameters}
%%%%%%%%%%%%%%%%%%%%%%%%%%%%%%%%%%%%%%%%%%%%%%%%%%%%%%%%%%%%%%%%%%%%%%%%%%%%%%%%%%%%%%%%%%
To formulate the MPC traffic signal control problem \eqref{eq_problemTP}, it is required to know the current traffic states and the traffic model parameters, i.e., the future exogenous in/out-flows and the turning ratios of the downstream traffic flows in next $K$ cycles of the road links.
Then the quality of the optimal traffic signal splits which are obtained by solving \eqref{eq_problemTP} may depend on the accuracy of the traffic model parameters.
They are usually assumed to be known exactly in the nominal MPC traffic signal control strategies.
However, this assumption is hardly satisfied in practical situations even though recent advancement of sensing and information technology increases the accuracy of measurements significantly.
The main reason is due to the fluctuation in the historically collected data and unpredictable events in future (weather conditions, accidents, and drivers' decisions).
When the uncertainties in the estimation of traffic model parameters cannot be avoided, the difference between the predicted states and the real ones increases as time horizon grows.
As a result, the obtained solution corresponding to the nominal case may not be optimal or even suboptimal.
This motivates us to study a stochastic optimal traffic signal control problem under the following assumption.
\begin{Assumption}\label{aspt_parameters}
The difference of the exogenous inflow and outflow, and the turning ratios of the traffic downstream flow of every road link $z \in \mathcal{L}$ at future time $t + k$, i.e., $e_z(t + k|t)$ and $r_{zw}(t + k|t)$ where $w \in \mathcal{N}_z^-$, are random variables with known statistical information.
\end{Assumption}

In Assumption \ref{aspt_parameters}, we assume that the expected values and variances of the exogenous in/out-flows and the turning ratios can be calculated with high confidence based on historically collected data though the future traffic model parameters cannot be predetermined precisely.
Let $\mathcal{RP}(t)$ be the set of all uncertain traffic model parameters for the predictive horizon at time $t$, i.e., $\mathcal{RP}(t) = \bigcup_{z \in \mathcal{L}} \mathcal{RP}_z(t)$ where
\begin{equation}\label{eq_aspt}
\mathcal{RP}_z(t) = \left\{e_z(t + k|t), r_{zw}(t + k|t): w \in \mathcal{N}_z^-, 0 \le k \le K - 1 \right\}.
\end{equation}
Then $E[X]$ and $Var[X]$ are known for every $X \in \mathcal{RP}(t)$.
In addition, the covariance matrix $\boldsymbol{\Sigma}[\textbf{X}]$ is given for any vector $\textbf{X} = [X_1, \dots, X_n]^T$, where $X_i \in \mathcal{RP}(t)$.
%%%%%%%%%%%%%%%%%%%%%%%%%%%%%%%%%%%%%%%%%%%%%%%%%%%%%%%%%%%%%%%%%%%%%%%%%%%%%%%%%%%%%%%%%%
\subsection{Problem reformulate}
%%%%%%%%%%%%%%%%%%%%%%%%%%%%%%%%%%%%%%%%%%%%%%%%%%%%%%%%%%%%%%%%%%%%%%%%%%%%%%%%%%%%%%%%%%
Since there exist uncertainties in the estimation of traffic model parameters, we consider the following stochastic program \eqref{eq_SMPCProblem} instead of the nominal MPC traffic signal control problem \eqref{eq_problemTP}.
\begin{subequations}\label{eq_SMPCProblem}
\begin{align}
\min &\textrm{ } E\left[ \Phi(t) \right]\\
\textrm{s.t. }& 
\eqref{eq_signal_limit1}, \eqref{eq_signal_limit2}, \forall J_v \in \mathcal{J}, 0 \le k \le K-1,\\
&\eqref{eq_signal_roadlink}, \eqref{eq_outstream_roadlink}, \forall z \in \mathcal{L}, 0 \le k \le K-1,\\
&P\left[\eqref{eq_predictedTrafficFlow} \textrm{ and } \eqref{eq_upstream_limit} \right] \ge 1 - \epsilon_t, \forall z \in \mathcal{L}, 1 \le k \le K,
\end{align}
\end{subequations}
where $P\left[ \cdot \right]$ is the probability of the event in $\left[ \cdot \right]$ and $\epsilon_t$ is a given small number.
It is easy to see that the problem \eqref{eq_SMPCProblem} is a stochastic version of the constrained optimization problem \eqref{eq_problemTP}.
Indeed, the cost function (\ref{eq_SMPCProblem}a) is the expectation of the nominal cost function in \eqref{eq_problemTP}; the hard constraints (\ref{eq_SMPCProblem}b) and (\ref{eq_SMPCProblem}c) need to be satisfied since they are independent of the uncertain traffic model parameters; and the constraint (\ref{eq_SMPCProblem}d) represents the requirement of the inequalities \eqref{eq_predictedTrafficFlow} and \eqref{eq_upstream_limit} for all road links in $\mathcal{L}$ need to be satisfied with the probability of higher than $1- \epsilon_t$.
The equality constraint \eqref{eq_predictedTrafficState} will be embedded in the computation of the expectation cost (\ref{eq_SMPCProblem}a) and the soft constraint (\ref{eq_SMPCProblem}d), as will be shown later.
The remainder of this subsection focuses on reformulating the stochastic program \eqref{eq_SMPCProblem} in an explicit form under Assumption \ref{aspt_parameters}.
%%%%%%%%%%%%%%%%%%%%%%%%%%%%%%%%%%%%%%%%%%%%%%%%%%%%%%%%%%%%%%%%%%%%%%%%%%%%%%%%%%%%%%%%%%
\subsubsection{Traffic states and cost function}
%%%%%%%%%%%%%%%%%%%%%%%%%%%%%%%%%%%%%%%%%%%%%%%%%%%%%%%%%%%%%%%%%%%%%%%%%%%%%%%%%%%%%%%%%%
We first observe that the predicted traffic state $n_z(t+k+1|t)$ in \eqref{eq_predictedTrafficState} is a random variable. It is a linear combination of the turning ratios corresponding to its upstream traffic flows and its exogenous traffic flows when considering specific traffic signal splits, $q_w(t + k|t)$ for all $w \in \mathcal{N}_z^+, k = 0, \dots K-1$, as coefficients.
Consequently, the cost function $\Phi(t)$ in \eqref{eq_costfunction} is also a random variable.
By using basic equations given in Appendix \ref{subEV}, the explicit formulations of the expectations and variances of traffic states and the cost function can be found as follows.

We define the random vector $\textbf{X}_z(t+k|t) = col\left\{r_{wz}(t+k|t): w \in \mathcal{N}_z^+\right\}$ and the stacked vector $\textbf{f}_z(t+k|t) = col\left\{q_w(t+k|t): w \in \mathcal{N}_z^+\right\}$ corresponding to the upstream traffic flows of the road link $z$ in time interval $[(t+k)T, (t+k+1)T]$.
Then the predicted traffic state $n_z(t + k + 1|t)$ is described by the following linear combination 
\[\sum_{l = 0}^{k}\textbf{f}_z(t+l|t)^T\textbf{X}_z(t+l|t) + \sum_{l = 0}^{k}e_z(t+l|t) + n_z(t) - \sum_{l = 0}^{k} q_z(t + l|t)\]
According to \eqref{eq_mean_linearcombination} and \eqref{eq_variance_linearcombination}, its expected value is $E\left[n_z(t + k + 1|t)\right] = n_z(t) - \sum_{l = 0}^{k} q_z(t+l|t) + \sum_{l = 0}^{k}\textbf{f}_z(t+l|t)^T E\left[\textbf{X}_z(t+l|t)\right] + \sum_{l = 0}^{k}E\left[e_z(t+l|t)\right]$ and its variance $Var\left[n_z(t + k + 1|t)\right]$ is given by
\[\textbf{f}_z(t:t+k|t)^T\boldsymbol{\Sigma}[\textbf{X}_z^{(1)}(t:t+k|t)]\textbf{f}_z(t:t+k|t).\]
where $\textbf{X}_z^{(1)}(t:t+k|t) = col\left\{\textbf{X}_z(t|t), \dots, \textbf{X}_z(t+k|t),\right.$ $\left.e_z(t|t) + \dots + e_z(t+k|t)\right\}$ and $\textbf{f}_z(t:t+k|t) = col\left\{\textbf{f}_z(t|t), \dots, \textbf{f}_z(t+k|t), 1\right\}$.
For simple notations, we define $\hat{n}_{z,k} = E\left[n_z(t + k + 1|t)\right]$, $\hat{q}_{z,k} = q_z(t + k|t)$ and $\hat{\textbf{f}}_{z,k} = \textbf{f}_z(t:t+k|t)$ for all road link $z \in \mathcal{L}$.
Then
\begin{equation}\label{eq_SMPC_TVPredicted}
\hat{n}_{z,k} = \hat{n}_{z,k-1} + \sum_{w \in \mathcal{N}_z^+} r_{wz,k}\hat{q}_{w,k} - \hat{q}_{z,k} + e_{z,k}
\end{equation}
where $r_{wz,k} = E[r_{wz}(t+k|t)]$ and $e_{z,k} = E[e_{z}(t + k|t)]$.
Under Assumption \ref{aspt_parameters}, the expected values $E[r_{wz}(t+k|t)], \forall w \in \mathcal{N}_z^+, E[e_{z}(t + k|t)]$ and the variance matrix $\boldsymbol{\Sigma}[\textbf{X}_z^{(1)}(t+k|t)]$ are known for all $z \in \mathcal{L}, k = 0, \dots, K-1$. Let the matrix $\textbf{G}_{z,k}^{(1)}$ be defined by $(\textbf{G}_{z,k}^{(1)})^T\textbf{G}_{z,k}^{(1)} = \boldsymbol{\Sigma}[\textbf{X}_z^{(1)}(t:t+k|t)]$. Then we have
\begin{equation}\label{eq_var_temp1}
Var[n_z(t + k + 1|t)] = ||\textbf{G}_{z,k}^{(1)}\hat{\textbf{f}}_{z,k}||^2.
\end{equation}
As a special case, for every source road link $z \in \mathcal{L}^{in}$, we have  $\mathcal{N}_z^+ = \emptyset$ and $Var[n_z(t + k + 1|t)] = ||\hat{\textbf{f}}_{z,k}||^2 = Var[e_{z}(t|t) + \dots + e_z(t + k|t)]$ is a scalar. In addition, we have $||\hat{\textbf{f}}_{z,0}||^2 = Var[e_{z}(t|t)]$ for every road link $z \in \mathcal{L}$.
By similar analysis, we have
\begin{equation}\label{eq_var_temp2}
Var[n_z(t+k|t) + e_z(t+k|t)] = ||\textbf{G}_{z,k}^{(2)}\hat{\textbf{f}}_{z,k-1}||^2.
\end{equation}
where $(\textbf{G}_{z,k}^{(2)})^T\textbf{G}_{z,k}^{(2)} = \boldsymbol{\Sigma}[\textbf{X}_z^{(2)}(t:t+k|t)]$ and $\textbf{X}_z^{(2)}(t:t+k|t) = col\bigl\{\textbf{X}_z(t|t), \dots, \textbf{X}_z(t+k-1|t), e_z(t|t) + \dots + e_z(t+k|t)\bigr\}$.
Moreover, we conveniently define $Var[n_z(t|t)] = \hat{f}_{z,-1} = 0, \forall z \in \mathcal{L}$.
From \eqref{eq_square_RV}, we have $E\left[\Bigl(n_z(t + k + 1|t)\Bigr)^2\right] = \hat{n}_{z,k}^2 + ||\hat{\textbf{f}}_{z,k}||_2^2$.
Then the expectation $E[\Phi(t)]$ is equal to
\begin{equation}\label{eq_expectation_cost}
\sum\limits_{k = 0}^{K-1} \sum\limits_{z \in \mathcal{L}} \left(\alpha_{z,k}\hat{n}_{z,k}^2 + \beta_{z,k}\hat{n}_{z,k} - \gamma_{z,k}\hat{q}_{z,k} + \alpha_{z,k}||\textbf{G}_{z,k}^{(1)}\hat{\textbf{f}}_{z,k}||^2\right)
\end{equation}
%%%%%%%%%%%%%%%%%%%%%%%%%%%%%%%%%%%%%%%%%%%%%%%%%%%%%%%%%%%%%%%%%%%%%%%%%%%%%%%%%%%%%%%%%%
\subsubsection{Constraints}
%%%%%%%%%%%%%%%%%%%%%%%%%%%%%%%%%%%%%%%%%%%%%%%%%%%%%%%%%%%%%%%%%%%%%%%%%%%%%%%%%%%%%%%%%%
The inequalities \eqref{eq_predictedTrafficFlow} can be rewritten as
\begin{equation}\label{eq_temp1}
-n_z(t+k|t) - e_z(t+k|t) + q_z(t+k|t) \le 0, \forall z \in \mathcal{L}
\end{equation}
From the vehicle conservation law \eqref{eq_predictedTrafficState}, we have $\sum_{w \in \mathcal{N}_z^+}r_{wz}(t+k|t)q_w(t+k|t) + n_z(t+k|t) + e_z(t+k|t) = n_z(t+k+1|t) + q_z(t+k|t)$ and the constraint \eqref{eq_upstream_limit} is equivalent to \eqref{eq_temp2} as follows.
\begin{subequations}\label{eq_temp2}
\begin{gather}
n_z(t+k|t) + e_z(t+k|t) \le \overline{n_z}, \forall z \in \mathcal{L}^{in},\\
n_z(t+k+1|t) + q_z(t+k|t) - \overline{n_z} \le 0, \forall z \in \mathcal{L} \backslash \mathcal{L}^{in}.
\end{gather}
\end{subequations}
By using the concept of distributionally robust chance constraint\cite{Calafiore2006}, which is summarized in Appendix \ref{subRC}, the constraint (\ref{eq_SMPCProblem}d) is equivalent to the following constraints 
\begin{subequations}\label{eq_safetyconstraint_withprob}
\begin{align}
\sqrt{\frac{1 - \epsilon_t}{\epsilon_t} Var[n_z(t + k|t) + e_z(t+k|t)]} &\le \hat{n}_{z,k-1} - \hat{q}_{z,k} + e_{z,k}, \forall z \in \mathcal{L}, k \ge 0\\
\sqrt{\frac{1 - \epsilon_t}{\epsilon_t} Var[n_z(t + k|t) + e_z(t+k|t)]} &\le \overline{n_z} - \hat{n}_{z,k} - e_{z,k}, \forall z \in \mathcal{L}^{in}, k \ge 0\\
\sqrt{\frac{1 - \epsilon_t}{\epsilon_t} Var[n_z(t + k + 1|t)]} &\le \overline{n_z} - \hat{n}_{z,k} - \hat{q}_{z,k}, \forall z \in \mathcal{L} \backslash \mathcal{L}^{in}, k \ge 0.
\end{align}
\end{subequations}
Since the variances $Var[n_z(t|t)] = 0, \forall z \in \mathcal{L},$ the following constraint is derived from (\ref{eq_safetyconstraint_withprob}a) for $k = 0$.
\begin{equation}\label{eq_soc_constraint01}
\hat{q}_{z,0} \le n_z(t) + e_{z,0} - ||\hat{\textbf{f}}_{z,0}||, \forall z \in \mathcal{L}.
\end{equation}
For every source road link $z \in \mathcal{L}^{in}$, we have $Var[n_z(t+k+1|t)], \forall k = 0, \dots, K-1,$ are known scalars. The inequalities (\ref{eq_safetyconstraint_withprob}a) for $k > 0$ and (\ref{eq_safetyconstraint_withprob}b) become
\begin{gather}
-\hat{n}_{z,k-1} + \hat{q}_{z,k} \le e_{z,k} - \sqrt{(1-\epsilon_t)/\epsilon_t}||\hat{\textbf{f}}_{z,k}||, \forall k \ge 1,\label{eq_soc_constraint02}\\
\hat{n}_{z,k} \le \overline{n_z} - e_{z,k+1} - \sqrt{(1-\epsilon_t)/\epsilon_t}||\hat{\textbf{f}}_{z,k+1}||, \forall k \le K.\label{eq_soc_constraint03}
\end{gather}
In the case of $z \in \mathcal{L} \backslash \mathcal{L}^{in}$, the constraints (\ref{eq_safetyconstraint_withprob}a) for $k > 0$ and (\ref{eq_safetyconstraint_withprob}c) can be transformed into second-order cone (SOC) constraints \eqref{eq_soc_constraint1} and \eqref{eq_soc_constraint2}, respectively.
\begin{align}
\left[\sqrt{\frac{1 - \epsilon_t}{\epsilon_t}} \left(\textbf{G}_{z,k+1}^{(2)}\hat{\textbf{f}}_{z,k}\right)^T, \hat{n}_{z,k} - \hat{q}_{z,k+1} + e_{z,k+1}\right]^T \in \mathcal{C}^{k|\mathcal{N}_z^+| + 5}, \forall 0 \le 0 < K-1, \label{eq_soc_constraint1}\\
\left[\sqrt{\frac{1 - \epsilon_t}{\epsilon_t}} \left(\textbf{G}_{z,k}^{(1)}\hat{\textbf{f}}_{z,k}\right)^T,  \overline{n_z} - \hat{n}_{z,k} - \hat{q}_{z,k}\right]^T \in \mathcal{C}^{k|\mathcal{N}_z^+| + 2}, \forall 0 \le 0 \le K-1.\label{eq_soc_constraint2}
\end{align}
where $\mathcal{C}^{m}$ is the unit second-order cone of dimension $m$ (see more in Appendix \ref{subProjBC}).
We also define $\hat{g}_{p,k} = g_p(t + k|t)$ for the planned traffic signal splits of the phase $p \in \mathcal{P}$ in the time interval $[(t+k)T, (t+k+1)T]$. They need to satisfy all their hard constraints. That means $\forall k \ge 0, J_v \in \mathcal{J}$, we have
\begin{subequations}\label{eq_hardconstraint}
\begin{align}
& 0 \le \hat{g}_{p,k} \le \overline{g_p}, \forall p \in \mathcal{P}_{J_v},\\ 
& \sum_{p \in \mathcal{P}_{J_v}} \hat{g}_{p,k} \le C_{J_v}.
\end{align}
\end{subequations}
The inequalities \eqref{eq_signal_roadlink} and \eqref{eq_outstream_roadlink}, $\forall k \ge 0, z \in \mathcal{L}$, are rewritten as
\begin{subequations}\label{eq_downstream_constraint}
\begin{gather}
0 \le \hat{q}_{z,k} \le S_z\sum\limits_{p \in \mathcal{P}_w}\hat{g}_{p,k}, \textrm{ if } \tau(z) \in \mathcal{J},\\
0 \le \hat{q}_{z,k} \le \overline{q_{z,k}}, \textrm{ if } \tau(z) \in \mathcal{O}.
\end{gather}
\end{subequations}
where $\overline{q_{z,k}} = q_z^{max}(t+k|t)$ for all $z \in \mathcal{L}^{out}$ or $\tau(z) \in \mathcal{O}$.
%%%%%%%%%%%%%%%%%%%%%%%%%%%%%%%%%%%%%%%%%%%%%%%%%%%%%%%%%%%%%%%%%%%%%%%%%%%%%%%%%%%%%%%%%%
\subsubsection{Stochastic MPC traffic signal control problem}
%%%%%%%%%%%%%%%%%%%%%%%%%%%%%%%%%%%%%%%%%%%%%%%%%%%%%%%%%%%%%%%%%%%%%%%%%%%%%%%%%%%%%%%%%%
Summarizing the analysis in this subsection, we have the detailed form of the stochastic program \eqref{eq_SMPCProblem} as follows:
\begin{equation}\label{eq_SMPCProblem_detailedform}
\min\limits_{\Omega_t}\textrm{ }\hat{\Phi}_t = E[\Phi(t)] \textrm{ given in \eqref{eq_expectation_cost}}
\end{equation}
where the feasible set $\Omega_t$ is defined as follows:
\scalebox{0.81}{$\Omega_t = \left\{ \left\{\begin{matrix} \hat{g}_{p,k}, \hat{n}_{z,k}, \hat{q}_{z,k}:\\ \forall p \in \mathcal{P}_{J_v}, J_v \in \mathcal{J},\\ z \in \mathcal{L}, 0 \le k \le K-1 \end{matrix} \right\} 
\left|\begin{matrix*}[l] 
\textrm{linear inequalities } \eqref{eq_soc_constraint01}, \eqref{eq_hardconstraint}, \eqref{eq_downstream_constraint},\\
\textrm{linear equalities }  \eqref{eq_SMPC_TVPredicted},\\
\textrm{linear inequalities } \eqref{eq_soc_constraint02}, \eqref{eq_soc_constraint03} \textrm{ if } \sigma(z) \in \mathcal{O},\\
\textrm{SOC constraints } \eqref{eq_soc_constraint1}, \eqref{eq_soc_constraint2} \textrm{ if } \sigma(z) \in \mathcal{J},
\end{matrix*}\right. \right\}$}.
Notice that the constrained optimization problem \eqref{eq_SMPCProblem_detailedform} has the form of a quadratic program with linear and second-order cone constraints.
Since the stochastic MPC traffic signal control problem \eqref{eq_SMPCProblem_detailedform} is a convex optimization problem, it has at least one optimal solution if its feasible set is nonempty.
So, we make the following assumption:
\begin{Assumption} 
For all $t \ge 0$, we have $\Omega_t \neq \emptyset$.
\end{Assumption}
%%%%%%%%%%%%%%%%%%%%%%%%%%%%%%%%%%%%%%%%%%%%%%%%%%%%%%%%%%%%%%%%%%%%%%%%%%%%%%%%%%%%%%%%%%

%%%%%%%%%%%%%%%%%%%%%%%%%%%%%%%%%%%%%%%%%%%%%%%%%%%%%%%%%%%%%%%%%%%%%%%%%%%%%%%%%%%%%%%%%%
\section{Distributed Stochastic MPC traffic signal control}
%%%%%%%%%%%%%%%%%%%%%%%%%%%%%%%%%%%%%%%%%%%%%%%%%%%%%%%%%%%%%%%%%%%%%%%%%%%%%%%%%%%%%%%%%%
\subsection{Decomposition of traffic network and control problem}
%%%%%%%%%%%%%%%%%%%%%%%%%%%%%%%%%%%%%%%%%%%%%%%%%%%%%%%%%%%%%%%%%%%%%%%%%%%%%%%%%%%%%%%%%%
Since an urban traffic network consists of many road links and junctions, the load of gathering traffic data (current traffic states, estimated exogenous in/out-flows, and turning ratios) to formulate the stochastic MPC traffic signal control problem \eqref{eq_SMPCProblem_detailedform} and the computation burden to solve this constrained optimization problem could be huge for only one centralized controller, specially in real-time application.
Thanks to the spatial structure of the urban network, the high-dimensional optimization problem of the overall network can be decomposed into multiple subproblems.
The optimal stochastic traffic signal splits can be obtained by applying distributed optimization methods.
In the following parts of this subsection, we first propose a way to decompose the urban network and formulate a distributed version of the stochastic MPC traffic signal control problem \eqref{eq_SMPCProblem_detailedform} corresponding to this decomposition.
Then, in next subsection, we propose a distributed method to determine an optimal solution of the problem \eqref{eq_SMPCProblem_detailedform} based on the proximal ADMM scheme \eqref{eq_ADMM} (given in Appendix A2).
%%%%%%%%%%%%%%%%%%%%%%%%%%%%%%%%%%%%%%%%%%%%%%%%%%%%%%%%%%%%%%%%%%%%%%%%%%%%%%%%%%%%%%%%%%
\subsubsection{Network decomposition}
%%%%%%%%%%%%%%%%%%%%%%%%%%%%%%%%%%%%%%%%%%%%%%%%%%%%%%%%%%%%%%%%%%%%%%%%%%%%%%%%%%%%%%%%%%
It is easy to observe that the problem \eqref{eq_SMPCProblem_detailedform} has a separable objective function and its constraints set can be divided into many individual subsets combining with coupled constraints.
This is because the considered urban network can be divided into many interconnected subnetworks.
Assume that it consists of $N$ subnetworks, denoted by $\mathcal{S}_i, i = 1,\dots, N$, and each of them is controlled by one local controller.
In multiagent system perspective, these local controllers are called agents and they work in parallel to reduce the execution time.

Let $\mathcal{J}_{\mathcal{S}_i}$ be the set of internal signalized junctions of the subnetwork $\mathcal{S}_i$ and $\mathcal{O}_{\mathcal{S}_i} = \{\sigma(z): z \in \mathcal{L}^{in} \textrm{ where } \tau(z) \in \mathcal{J}_{\mathcal{S}_i}\} \cup \{\tau(z): z \in \mathcal{L}^{in} \textrm{ where } \sigma(z) \in \mathcal{J}_{\mathcal{S}_i}\}$ be the set of external nodes connecting to junctions in $\mathcal{J}_{\mathcal{S}_i}$.
The set $\mathcal{L}_{\mathcal{S}_i}$ of internal road links of the subnetwork $\mathcal{S}_i$ is defined by \[\mathcal{L}_{\mathcal{S}_i} = \{z \in \mathcal{L}: \sigma(z), \tau(z) \in \mathcal{J}_{\mathcal{S}_i} \cup \mathcal{O}_{\mathcal{S}_i}\}\] and the set of source road links of the subnetwork $\mathcal{S}_i$ is \[\mathcal{L}_{\mathcal{S}_i}^{in} = \left\{z: \sigma(z) \in \mathcal{O}_{\mathcal{S}_i} \textrm{ and } \tau(z) \in \mathcal{J}_{\mathcal{S}_i} \right\}.\]
Define $\mathcal{L}_{\mathcal{S}_i\mathcal{S}_j}$ as the set of all road links connecting from the subnetwork $\mathcal{S}_i$ to the subnetwork $\mathcal{S}_j$, that means \[\mathcal{L}_{\mathcal{S}_i\mathcal{S}_j} = \left\{z: \sigma(z) \in \mathcal{J}_{\mathcal{S}_i} \textrm{ and } \tau(z) \in \mathcal{J}_{\mathcal{S}_j}\right\}.\]
By abuse of notations, we use $\mathcal{S}_i$ to refer the agent and the subproblem corresponding to the subnetwork $\mathcal{S}_i$. 
If $\mathcal{L}_{\mathcal{S}_i\mathcal{S}_j} \neq \emptyset$, we call $\mathcal{S}_i$ and $\mathcal{S}_j$ are neighboring agents (subnetworks/subproblems).
The set of neighbors of an agent $\mathcal{S}_i$ is defined by \[\mathcal{N}_{\mathcal{S}_i} = \left\{\mathcal{S}_j: \mathcal{L}_{\mathcal{S}_i\mathcal{S}_j} \neq \emptyset \textrm{ and/or } \mathcal{L}_{\mathcal{S}_j\mathcal{S}_i} \neq \emptyset \right\}.\]

For better understanding of the aforementioned notations, we consider the urban network in Fig. \ref{fig_exampleUTN} when assuming it is divided into two subnetworks: $\mathcal{S}_1$ consists of two junctions $J_1, J_2$ and $\mathcal{S}_2$ consists of two junctions $J_3, J_4$.
That means $\mathcal{J}_{\mathcal{S}_1} = \{J_1, J_2\}$ and $\mathcal{J}_{\mathcal{S}_2} = \{J_3, J_4\}$.
The sets of road links corresponding to subnetworks are then given as: $\mathcal{L}_{\mathcal{S}_1} = \{1, 2, \dots, 9\}$, $\mathcal{L}_{\mathcal{S}_2} = \{16, 17, \dots, 31\}$, $\mathcal{L}_{\mathcal{S}_1\mathcal{S}_2} = \{10, 11, 13, 14\}$, $\mathcal{L}_{\mathcal{S}_2\mathcal{S}_1} = \{12, 15\}$ and $\mathcal{L}_{\mathcal{S}_1}^{in} = \{1, 4, 8\}$, $\mathcal{L}_{\mathcal{S}_2}^{in} = \{17, 18, 23, 24, 27, 28, 30, 31\}$.
In addition, we have $\mathcal{N}_{\mathcal{S}_1} = \left\{\mathcal{S}_2\right\}$, $\mathcal{N}_{\mathcal{S}_2} = \left\{\mathcal{S}_1\right\}$.
%%%%%%%%%%%%%%%%%%%%%%%%%%%%%%%%%%%%%%%%%%%%%%%%%%%%%%%%%%%%%%%%%%%%%%%%%%%%%%%%%%%%%%%%%%
\subsubsection{Problem splitting}
%%%%%%%%%%%%%%%%%%%%%%%%%%%%%%%%%%%%%%%%%%%%%%%%%%%%%%%%%%%%%%%%%%%%%%%%%%%%%%%%%%%%%%%%%%
To reformulate the problem \eqref{eq_SMPCProblem_detailedform} in a distributed manner, we now define the local control variables, local constraints for each agent (i.e., the local subproblem) and the coupled constraints among neighboring agents.
The local subproblem of the subnetwork $\mathcal{S}_i$ is obtained by substituting the sets $\mathcal{L}_{\mathcal{S}_i}$ and $\mathcal{J}_{\mathcal{S}_i}$ into the role of the sets $\mathcal{L}$ and $\mathcal{J}$ respectively in the problem \eqref{eq_SMPCProblem_detailedform}.
Naturally, the local control variables of agent $\mathcal{S}_i$ are the signal splits of its internal junctions (i.e., $\hat{g}_{p,k}$ for all $p \in \mathcal{P}_{J_v}$, $J_v \in \mathcal{J}_{\mathcal{S}_i}$, $k = 0, \dots, K-1$), the downstream traffic flows and the expected traffic states of its internal road links (i.e., $\hat{q}_{z,k}$ and $\hat{n}_{z,k}$ for all $z \in \mathcal{L}_{\mathcal{S}_i}$, $k = 0, \dots, K-1$).
For one road link $z \in \mathcal{L}_{\mathcal{S}_i\mathcal{S}_j}$, its upstream traffic flows depend on the traffic signals of one intersection in $\mathcal{S}_i$ while its downstream traffic flows depend on the one in $\mathcal{S}_j$.
For easy decomposition of the problem \eqref{eq_SMPCProblem_detailedform}, in this paper, we consider the expectations of its predicted traffic states as local variables in the subproblem $\mathcal{S}_i$ and consider its downstream traffic flows as local variables in the subproblem $\mathcal{S}_j$.
In addition, we make the following assumption:
\begin{Assumption}
Each agent $\mathcal{S}_i$ can measure and estimate the stochastic information of random variables in the set $\mathcal{RP}_{\mathcal{S}_i}(t) = \bigcup\limits_{J_v \in \mathcal{J}_{\mathcal{S}_i}}\Bigl( \bigcup\limits_{z \in \mathcal{L}_{J_v}^{in} \cup \mathcal{L}_{J_v}^{out}} \mathcal{RP}_z(t)\Bigr)$.
\end{Assumption}
\begin{Remark}
As shown in the first part of Subsection III.B, the expectation and the variance of traffic states corresponding to the road link $z$ depend on the stochastic information of the turning ratios of its upstream traffic flows and its exogenous traffic flows, i.e., $r_{wz}(t+k|t), \forall w \in \mathcal{N}_z^+$ and $e_z(t+k|t)$, $\forall k = 0, 1, \dots, K-1$. Then if an agent $\mathcal{S}_i$ can measure/estimate the expectations, variances and covariances for random variables in the union of the sets $\mathcal{RP}_z(t), \forall z \in \mathcal{L}_{J_v}^{in} \cup \mathcal{L}_{J_v}^{out}$, corresponding to an internal junction $J_v \in \mathcal{J}_{\mathcal{S}_i}$, this agent can formulate the detailed form of the expectation \eqref{eq_SMPC_TVPredicted} and variances \eqref{eq_var_temp1}, \eqref{eq_var_temp2} for all road links in $\mathcal{L}_{J_v}^{out}$.
\end{Remark}

Consider the expectation cost $\hat{\Phi}_t = E[\Phi(t)]$ in \eqref{eq_expectation_cost}, it can be decomposed into the sum of local cost functions as $\hat{\Phi}_t = \sum_{i = 1}^{N} \hat{\Phi}_{\mathcal{S}_i}(t)$ with \[\hat{\Phi}_{\mathcal{S}_i}(t) = \scalebox{0.85}{$\sum\limits_{z: \sigma(z) \in \mathcal{J}_{\mathcal{S}_i} \cup \mathcal{O}_{\mathcal{S}_i}} \hat{\Phi}_z^{n}(t) + \sum\limits_{z: \tau(z) \in \mathcal{J}_{\mathcal{S}_i} \cup \mathcal{O}_{\mathcal{S}_i}} \hat{\Phi}_z^{q}(t) + \sum\limits_{J_v \in \mathcal{J}_{\mathcal{S}_i}} \hat{\Phi}_{J_v}(t)$},\]
where $\hat{\Phi}_z^{n}(t) = \sum\limits_{k = 0}^{K-1}\left(\alpha_{z,k}\hat{n}_{z,k}^2 + \beta_{z,k}\hat{n}_{z,k}\right)$, $\hat{\Phi}_z^{q}(t) = \sum\limits_{k = 0}^{K-1}\left(-\gamma_{z,k}\hat{q}_{z,k}\right)$ and $\hat{\Phi}_{J_v}(t) = \sum\limits_{k = 0}^{K-1} \sum\limits_{z \in \mathcal{L}_{J_v}^{out}} \alpha_{z,k}||\hat{f}_{z,k}||^2$.
Note that $\hat{\Phi}_{J_v}$ is the function of $\hat{q}_{w,k}, \forall w \in \mathcal{L}_{J_v}^{in}, k = 0, \dots, K-1$.
Let $\hat{\textbf{n}}_{\mathcal{S}_i,k} = col\bigl\{\hat{n}_{z,k}: z \in \mathcal{L}_{\mathcal{S}_i} \cup \mathcal{L}_{\mathcal{S}_i}^{in} \cup \bigcup_{\mathcal{S}_j \in \mathcal{N}_{\mathcal{S}_i}} \mathcal{L}_{\mathcal{S}_i\mathcal{S}_j}\bigr\}$ and $\hat{\textbf{q}}_{\mathcal{S}_i,k} = col\bigl\{\hat{q}_{z,k}: z \in \bigcup_{J_v \in \mathcal{J}_{\mathcal{S}_i} \cup \mathcal{O}_{\mathcal{S}_i}} \mathcal{L}_{J_v}^{in} \bigr\}$, $\forall k = 0, \dots, K-1$, and consider these stacked vectors as local variables of agent $\mathcal{S}_i$. It is easy to verify that $\hat{\textbf{n}}_{\mathcal{S}_i,k} = col\left\{\hat{n}_{z,k}: \sigma(z) \in \mathcal{J}_{\mathcal{S}_i} \cup \mathcal{O}_{\mathcal{S}_i}\right\}$ and $\hat{\textbf{q}}_{\mathcal{S}_i,k} = col\left\{\hat{q}_{z,k}: \tau(z) \in \mathcal{J}_{\mathcal{S}_i} \cup \mathcal{O}_{\mathcal{S}_i}\right\}$
and the local cost function of the agent $\mathcal{S}_i$, i.e., $\hat{\Phi}_{\mathcal{S}_i}(t)$, is a function depending on $\hat{\textbf{n}}_{\mathcal{S}_i,k}, \hat{\textbf{q}}_{\mathcal{S}_i,k}, \forall k = 0, \dots, K-1$.

Consider the constraints \eqref{eq_SMPC_TVPredicted} for all $k = 0, \dots, K - 1$, it is easy to verify that the expectation values $\hat{n}_{z,k}$ corresponding to a road link $z \in \mathcal{L}_{\mathcal{S}_i}$ depend on local variables in $\hat{\textbf{n}}_{\mathcal{S}_i,k}, \hat{\textbf{q}}_{\mathcal{S}_i,k}$, $k = 0, \dots, K - 1$. However, the ones corresponding to a road link $z \in \mathcal{L}_{\mathcal{S}_i\mathcal{S}_j}$ are functions depending on the downstream traffic flows $\hat{q}_{z,k}, \forall k = 0, \dots, K - 1$, (which are local variables of the agent $\mathcal{S}_j$). To cope with this issue, we assume that the agent $\mathcal{S}_i$ stores copies of these variables. Let $\hat{q}_{z,k}^{copy} = \hat{q}_{z,k}, \forall z \in \mathcal{L}_{\mathcal{S}_i\mathcal{S}_j}, \forall \mathcal{S}_j \in \mathcal{N}_{\mathcal{S}_i}, k = 0, \dots, K-1$. The copy $\hat{q}_{z,k}^{copy}, \forall k = 0, \dots, K-1$, can substitute to the variable $\hat{q}_{z,k}$ in the equation \eqref{eq_SMPC_TVPredicted} and the inequalities \eqref{eq_soc_constraint1}, \eqref{eq_soc_constraint2} for the road link $z \in \mathcal{L}_{\mathcal{S}_i\mathcal{S}_j}$ to convert these constraints to local constraints of agent $\mathcal{S}_i$.
We also define $\hat{\textbf{g}}_{\mathcal{S}_i,k} = col\left\{\hat{g}_{p,k}: p \in \bigcup_{J_v \in \mathcal{J}_{\mathcal{S}_i}} \mathcal{P}_{J_v}\right\}$ as local variables of $\mathcal{S}_i$.
Thus, the vector of all local control variables in the subproblem $\mathcal{S}_i$ is defined as $\textbf{x}_{\mathcal{S}_i} = \left[\begin{matrix} \hat{\textbf{x}}_{\mathcal{S}_i}^T, row\left\{\hat{\textbf{x}}_{\mathcal{S}_i\mathcal{S}_j}: \mathcal{S}_j \in \mathcal{N}_{\mathcal{S}_i}\right\}\end{matrix}\right]^T$ where
\[\hat{\textbf{x}}_{\mathcal{S}_i} = \left[\begin{matrix} \hat{\textbf{n}}_{\mathcal{S}_i,0}^T, \hat{\textbf{q}}_{\mathcal{S}_i,0}^T, \hat{\textbf{g}}_{\mathcal{S}_i,0}^T, \cdots, \hat{\textbf{n}}_{\mathcal{S}_i,K-1}^T, \hat{\textbf{q}}_{\mathcal{S}_i,K-1}^T, \hat{\textbf{g}}_{\mathcal{S}_i,K-1}^T\end{matrix}\right]^T.\]
\[\hat{\textbf{x}}_{\mathcal{S}_i\mathcal{S}_j} = col\left\{\hat{q}_{z,k}^{copy}: z \in \mathcal{L}_{\mathcal{S}_i\mathcal{S}_j}, 0 \le k \le K-1\right\}, \forall \mathcal{S}_j \in \mathcal{N}_{\mathcal{S}_i}.\]
Let $\zeta_{\mathcal{S}_i}$ be the dimension of the local vector $\textbf{x}_{\mathcal{S}_i} \in \mathbb{R}^{\zeta_{\mathcal{S}_i}}$.
We consider the constraints $\hat{q}_{z,k}^{copy} = \hat{q}_{z,k}, \forall z \in \mathcal{L}_{\mathcal{S}_i\mathcal{S}_j}, \forall k = 0, \dots, K-1$, as the coupled constraints between two agents $\mathcal{S}_i$ and $\mathcal{S}_j$. To describe these constraints in a stacked form, we define two matrices $\textbf{P}_{\mathcal{S}_i\mathcal{S}_j}$ and $\textbf{Q}_{\mathcal{S}_i\mathcal{S}_j}$ such that $\textbf{P}_{\mathcal{S}_i\mathcal{S}_j}\textbf{x}_{\mathcal{S}_i} = \hat{\textbf{x}}_{\mathcal{S}_i\mathcal{S}_j}$ and $\textbf{Q}_{\mathcal{S}_i\mathcal{S}_j}\textbf{x}_{\mathcal{S}_i} = col\left\{\hat{q}_{z,k}: z \in \mathcal{L}_{\mathcal{S}_j\mathcal{S}_i}, 0 \le k \le K-1\right\}$. So, for all $i = 1, \dots, N$, we have
\begin{equation}\label{eq_coupledequalities_subproblem}
\textbf{P}_{\mathcal{S}_i\mathcal{S}_j}\textbf{x}_{\mathcal{S}_i} = \textbf{Q}_{\mathcal{S}_j\mathcal{S}_i}\textbf{x}_{\mathcal{S}_j}, \forall \mathcal{S}_j \in \mathcal{N}_{\mathcal{S}_i}.
\end{equation}
From the definitions of the local cost function and variables, we have $\hat{\Phi}_{\mathcal{S}_i} = \Phi_{\mathcal{S}_i}(\textbf{x}_{\mathcal{S}_i}) + constant$ where
\begin{equation}\label{eq_localcost_subproblem}
\Phi_{\mathcal{S}_i}(\textbf{x}_{\mathcal{S}_i}) =\frac{1}{2}\textbf{x}_{\mathcal{S}_i}^T\textbf{H}_{\mathcal{S}_i}\textbf{x}_{\mathcal{S}_i} + \textbf{h}_{\mathcal{S}_i}^T\textbf{x}_{\mathcal{S}_i}.
\end{equation}
In which $\textbf{H}_{\mathcal{S}_i} = \left[[\textbf{H}_{\mathcal{S}_i}]_{mn}\right] \in \mathbb{R}^{\zeta_{\mathcal{S}_i} \times \zeta_{\mathcal{S}_i}}$ and $\textbf{h}_{\mathcal{S}_i} = \left[[\textbf{h}_{\mathcal{S}_i}]_{m}\right] \in \mathbb{R}^{\zeta_{\mathcal{S}_i}}$ are defined as follows:
\[[\textbf{H}_{\mathcal{S}_i}]_{mn} = \left\{ \begin{matrix*}[l]
\alpha_{z,k}, & \textrm{if } m = n, [\textbf{x}_{\mathcal{S}_i}]_m \equiv \hat{n}_{z,k},\\
\sum\limits_{z \in \mathcal{N}_{w_1}^- \cap \mathcal{N}_{w_2}^-} \sum\limits_{l = k}^{K-1} \alpha_{z,l}[\textbf{G}_{z,l}^{(1)}]_{m_ln_l}, & \textrm{if } [\textbf{x}_{\mathcal{S}_i}]_m \equiv \hat{q}_{w_1,k}, [\textbf{x}_{\mathcal{S}_i}]_n \equiv \hat{q}_{w_2,k} \textrm{ and } [\textbf{f}_{z,l}]_{m_l} \equiv \hat{q}_{w_1,k}, [\textbf{f}_{z,l}]_{n_l} \equiv \hat{q}_{w_2,k},\\
0, & otherwise
\end{matrix*} \right.\]
\[[\textbf{h}_{\mathcal{S}_i}]_{m} = \left\{ \begin{matrix*}[l]
\beta_{z,k}, & \textrm{if } [\textbf{x}_{\mathcal{S}_i}]_m \equiv \hat{n}_{z,k},\\
-\gamma_{z,k}, & \textrm{if } [\textbf{x}_{\mathcal{S}_i}]_m \equiv \hat{q}_{z,k},\\
0, & otherwise
\end{matrix*} \right.\]
It can be verified that that $\textbf{H}_{\mathcal{S}_i}$ is a symmetric and positive semidefinite matrix, $\textbf{H}_{\mathcal{S}_i} \succeq 0$ and $\textbf{H}_{\mathcal{S}_i} = \textbf{H}_{\mathcal{S}_i}^T$, for all $i$.
Since the constant part in $\hat{\Phi}$ can be removed in optimization problem without any change in the optimal solution, $\Phi_{\mathcal{S}_i}(\textbf{x}_{\mathcal{S}_i})$ is now considered as a local cost function of the agent $\mathcal{S}_i$, $\forall i = 1, \dots, N$.
Due to the dependence on only the variables in $\textbf{x}_{\mathcal{S}_i}$, the following constraints are considered as the local constraints of the agent $\mathcal{S}_i$.

The linear inequalities \eqref{eq_soc_constraint01}, \eqref{eq_downstream_constraint} for all incoming road links of junctions in $\mathcal{J}_{\mathcal{S}_i} \cup \mathcal{O}_{\mathcal{S}_i}$, \eqref{eq_hardconstraint} for all traffic signal phases of junctions in $\mathcal{J}_{\mathcal{S}_i}$, and \eqref{eq_soc_constraint02}, \eqref{eq_soc_constraint03} for all source road links in $\mathcal{L}_{\mathcal{S}_i}^{in}$ can be described in the matrix form \[\textbf{D}_{0, \mathcal{S}_i}\textbf{x}_{\mathcal{S}_i} \le \textbf{d}_{0, \mathcal{S}_i}\] with suitable matrix $\textbf{D}_{0, \mathcal{S}_i} = \left[\left(\textbf{D}_{0, \mathcal{S}_i}^{\star}\right)^T, \left(\textbf{D}_{0, \mathcal{S}_i}^{(0)}\right)^T, \left(\textbf{D}_{0, \mathcal{S}_i}^{(1)}\right)^T, \dots, \left(\textbf{D}_{0, \mathcal{S}_i}^{(K-1)}\right)^T\right]^T \in \mathbb{R}^{\xi_{0,\mathcal{S}_i} \times \zeta_{\mathcal{S}_i}}$ and vector $\textbf{d}_{0, \mathcal{S}_i} = \left[\left(\textbf{d}_{0, \mathcal{S}_i}^{\star}\right)^T, \left(\textbf{d}_{0, \mathcal{S}_i}^{(0)}\right)^T, \left(\textbf{d}_{0, \mathcal{S}_i}^{(1)}\right)^T, \dots, \left(\textbf{d}_{0, \mathcal{S}_i}^{(K-1)}\right)^T\right]^T \in \mathbb{R}^{\xi_{0,\mathcal{S}_i}}$. In which, $\textbf{D}_{0, \mathcal{S}_i}^{\star} = -\textbf{I}$ having similar dimension as $\textbf{x}_{\mathcal{S}_i}$, $\textbf{d}_{0, \mathcal{S}_i}^{\star} = \textbf{0}$ having similar dimension as $\textbf{x}_{\mathcal{S}_i}$ and $\textbf{D}_{0, \mathcal{S}_i}^{(0)} \in \mathbb{R}^{(2\zeta_{1,\mathcal{S}_i} + 2\zeta_{4,\mathcal{S}_i} + \zeta_{5,\mathcal{S}_i} + \zeta_{6,\mathcal{S}_i}) \times \zeta_{\mathcal{S}_i}}$, $\textbf{D}_{0, \mathcal{S}_i}^{(k)} \in \mathbb{R}^{(\zeta_{1,\mathcal{S}_i} + 2\zeta_{2,\mathcal{S}_i} + \zeta_{4,\mathcal{S}_i} + \zeta_{5,\mathcal{S}_i} + \zeta_{6,\mathcal{S}_i}) \times \zeta_{\mathcal{S}_i}}$
where $\zeta_{1,\mathcal{S}_i} = |\mathcal{L}_{\mathcal{S}_i}|, \zeta_{2,\mathcal{S}_i} = |\mathcal{L}_{\mathcal{S}_i}^{in}|, \zeta_{3,\mathcal{S}_i} = |\{z: z \in \mathcal{L}_{\mathcal{S}_i\mathcal{S}_j}, \mathcal{S}_j \in \mathcal{N}_{\mathcal{S}_i}\}|, \zeta_{4,\mathcal{S}_i} = |\{z: z \in \mathcal{L}_{\mathcal{S}_j\mathcal{S}_i}, \mathcal{S}_j \in \mathcal{N}_{\mathcal{S}_i}\}|, \zeta_{5,\mathcal{S}_i} = |\{p:p \in \mathcal{P}_{J_v}, J_v \in \mathcal{J}_{\mathcal{S}_i}\}|$ and $\zeta_{6,\mathcal{S}_i} = |\mathcal{J}_{\mathcal{S}_i}|$.
Their elements are defined as follows.\\
\scalebox{0.9}{$\textbf{D}_{0, \mathcal{S}_i}^{(0)}]_{mn} = \left\{ \begin{matrix*}[l]
1, & \textrm{if } 0 < m \le \zeta_{1,\mathcal{S}_i} + \zeta_{4,\mathcal{S}_i}, n = m + \zeta_{1,\mathcal{S}_i} + \zeta_{3,\mathcal{S}_i},\\
1, & \textrm{if } \zeta_{1,\mathcal{S}_i} + \zeta_{4,\mathcal{S}_i} < m \le 2(\zeta_{1,\mathcal{S}_i} + \zeta_{4,\mathcal{S}_i}), n = m + \zeta_{3,\mathcal{S}_i} - \zeta_{4,\mathcal{S}_i},\\
-S_z, & \textrm{if } \zeta_{1,\mathcal{S}_i} + \zeta_{4,\mathcal{S}_i} < m \le 2(\zeta_{1,\mathcal{S}_i} + \zeta_{4,\mathcal{S}_i}), [\hat{\textbf{q}}_{\mathcal{S}_i,0}]_{m - \zeta_{1,\mathcal{S}_i} - \zeta_{4,\mathcal{S}_i}} \equiv \hat{q}_{z,0},[\textbf{x}_{\mathcal{S}_i}]_n \equiv \hat{g}_{p,0}, p \in \mathcal{P}_z,\\
1, & \textrm{if } 2(\zeta_{1,\mathcal{S}_i} + \zeta_{4,\mathcal{S}_i}) < m \le 2(\zeta_{1,\mathcal{S}_i} + \zeta_{4,\mathcal{S}_i}) + \zeta_{5,\mathcal{S}_i}, n = m + 2\zeta_{1,\mathcal{S}_i} + \zeta_{3,\mathcal{S}_i} + \zeta_{4,\mathcal{S}_i},\\
1, & \textrm{if } m = l + 2(\zeta_{1,\mathcal{S}_i} + \zeta_{4,\mathcal{S}_i}) + \zeta_{5,\mathcal{S}_i}, l > 0, [\textbf{x}_{\mathcal{S}_i}]_n \equiv \hat{g}_{p,0}, p \in \mathcal{P}_{J_v}, J_v \textrm{ is } l^{th}-element \textrm{ in } \mathcal{P}_{\mathcal{S}_i}\\
0, & otherwise.
\end{matrix*} \right.$}\\
\scalebox{0.9}{$[\textbf{d}_{0, \mathcal{S}_i}^{(0)}]_{m} = \left\{ \begin{matrix*}[l]
n_z(t) + e_{z,0} - ||\hat{\textbf{f}}_{z,0}||^2, & \textrm{if } 0 < m \le \zeta_{1,\mathcal{S}_i} + \zeta_{4,\mathcal{S}_i}, [\hat{\textbf{q}}_{\mathcal{S}_i,0}]_m \equiv \hat{q}_{z,0},\\
\overline{q_{z,0}}, & \textrm{if } m = n + \zeta_{1,\mathcal{S}_i} + \zeta_{4,\mathcal{S}_i}, 0 < n \le \zeta_{1,\mathcal{S}_i} + \zeta_{4,\mathcal{S}_i}, [\hat{\textbf{q}}_{\mathcal{S}_i,0}]_{n} \equiv \hat{q}_{z,0},\tau(z) \in \mathcal{O}_{\mathcal{S}_i},\\
\overline{g_p}, & \textrm{if } m = n + 2(\zeta_{1,\mathcal{S}_i} + \zeta_{4,\mathcal{S}_i}), 0 < n \le \zeta_{5,\mathcal{S}_i}, [\hat{\textbf{g}}_{\mathcal{S}_i,0}]_n \equiv \hat{g}_{p,0},\\
C_{J_v}, & \textrm{if } m = n + 2(\zeta_{1,\mathcal{S}_i} + \zeta_{4,\mathcal{S}_i}) + \zeta_{5,\mathcal{S}_i}, 0 < n \le \zeta_{6,\mathcal{S}_i}, J_v \textrm{ is } n^{th}-element \textrm{ in } \mathcal{P}_{\mathcal{S}_i}\\
0, & otherwise.
\end{matrix*} \right.$}\\
If $0 < k \le K-1$\\
\scalebox{0.8}{$[\textbf{D}_{0, \mathcal{S}_i}^{(k)}]_{mn} = \left\{ \begin{matrix*}[l]
-1, & \textrm{if } 0 < m \le \zeta_{2,\mathcal{S}_i}, n = m + (k-1)(2\zeta_{1,\mathcal{S}_i} + \zeta_{3,\mathcal{S}_i} + \zeta_{4,\mathcal{S}_i}),\\
1, & \textrm{if } 0 < m \le \zeta_{2,\mathcal{S}_i}, n = m + (k-1)(2\zeta_{1,\mathcal{S}_i} + \zeta_{3,\mathcal{S}_i} + \zeta_{4,\mathcal{S}_i}) + \zeta_{1,\mathcal{S}_i} + \zeta_{3,\mathcal{S}_i},\\
1, & \textrm{if } \zeta_{2,\mathcal{S}_i} < m \le 2\zeta_{2,\mathcal{S}_i}, n = m + k(2\zeta_{1,\mathcal{S}_i} + \zeta_{3,\mathcal{S}_i} + \zeta_{4,\mathcal{S}_i}),\\
1, & \textrm{if } 2\zeta_{2,\mathcal{S}_i} < m \le 2\zeta_{2,\mathcal{S}_i} + \zeta_{1,\mathcal{S}_i} + \zeta_{4,\mathcal{S}_i}, n = m + \zeta_{1,\mathcal{S}_i} + \zeta_{3,\mathcal{S}_i} - 2\zeta_{2,\mathcal{S}_i},\\
-S_z, & \textrm{if } 2\zeta_{2,\mathcal{S}_i} < m \le 2\zeta_{2,\mathcal{S}_i} + \zeta_{1,\mathcal{S}_i} + \zeta_{4,\mathcal{S}_i}, [\hat{\textbf{q}}_{\mathcal{S}_i,k}]_{m - 2\zeta_{2,\mathcal{S}_i}} \equiv \hat{q}_{z,k},[\textbf{x}_{\mathcal{S}_i}]_n \equiv \hat{g}_{p,k}, p \in \mathcal{P}_z,\\
1, & \textrm{if } 2\zeta_{2,\mathcal{S}_i} + \zeta_{1,\mathcal{S}_i} + \zeta_{4,\mathcal{S}_i} < m \le 2\zeta_{2,\mathcal{S}_i} + \zeta_{1,\mathcal{S}_i} + \zeta_{4,\mathcal{S}_i} + \zeta_{5,\mathcal{S}_i}, n = m + k(2\zeta_{1,\mathcal{S}_i} + \zeta_{3,\mathcal{S}_i} + \zeta_{4,\mathcal{S}_i}) + 2\zeta_{1,\mathcal{S}_i} + \zeta_{3,\mathcal{S}_i} + \zeta_{4,\mathcal{S}_i},\\
1, & \textrm{if } m = l + 2\zeta_{2,\mathcal{S}_i} + \zeta_{1,\mathcal{S}_i} + \zeta_{4,\mathcal{S}_i} + \zeta_{5,\mathcal{S}_i}, l > 0, [\textbf{x}_{\mathcal{S}_i}]_n \equiv \hat{g}_{p,k}, p \in \mathcal{P}_{J_v}, J_v \textrm{ is } l^{th}-element \textrm{ in } \mathcal{P}_{\mathcal{S}_i}\\
0, & otherwise.
\end{matrix*} \right.$}\\
\scalebox{0.8}{$[\textbf{d}_{0, \mathcal{S}_i}^{(k)}]_{m} = \left\{ \begin{matrix*}[l]
e_{z,k} - \sqrt{(1-\epsilon_t)/\epsilon_t}||\hat{\textbf{f}}_{z,k}||^2, & \textrm{if } 0 < m \le \zeta_{2,\mathcal{S}_i}, [\hat{\textbf{q}}_{z,k}]_m \equiv \hat{q}_{z,k}, z \in \mathcal{L}_{\mathcal{S}_i}^{in},\\
\overline{n_z} - e_{z,k+1} - \sqrt{(1-\epsilon_t)/\epsilon_t}||\hat{\textbf{f}}_{z,k+1}||^2, & \textrm{if } m = n + \zeta_{2,\mathcal{S}_i}, 0 < n \le \zeta_{2,\mathcal{S}_i}, [\hat{\textbf{q}}_{z,k}]_n \equiv \hat{q}_{z,k}, z \in \mathcal{L}_{\mathcal{S}_i}^{in},\\
\overline{q_{z,k}}, & \textrm{if } m = n + \zeta_{1,\mathcal{S}_i} + \zeta_{4,\mathcal{S}_i}, 0 < n \le \zeta_{1,\mathcal{S}_i} + \zeta_{4,\mathcal{S}_i}, [\hat{\textbf{q}}_{\mathcal{S}_i,k}]_{n} \equiv \hat{q}_{z,k},\tau(z) \in \mathcal{O}_{\mathcal{S}_i},\\
\overline{g_p}, & \textrm{if } m = n + 2(\zeta_{1,\mathcal{S}_i} + \zeta_{4,\mathcal{S}_i}), 0 < n \le \zeta_{5,\mathcal{S}_i}, [\hat{\textbf{g}}_{\mathcal{S}_i,k}]_n \equiv \hat{g}_{p,k},\\
C_{J_v}, & \textrm{if } m = n + 2(\zeta_{1,\mathcal{S}_i} + \zeta_{4,\mathcal{S}_i}) + \zeta_{5,\mathcal{S}_i}, 0 < n \le \zeta_{6,\mathcal{S}_i}, J_v \textrm{ is } n^{th}-element \textrm{ in } \mathcal{P}_{\mathcal{S}_i}\\
0, & otherwise.
\end{matrix*} \right.$}\\
The second-order cone constraint (\eqref{eq_soc_constraint1} and \eqref{eq_soc_constraint2}) with the definitions \eqref{eq_var_temp1} and \eqref{eq_var_temp2} have the following form \[\textbf{D}_{j, \mathcal{S}_i}\textbf{x}_{\mathcal{S}_i} - \textbf{d}_{j, \mathcal{S}_i} \in \mathcal{C}_{j,\mathcal{S}_i}\] where $\mathcal{C}_{j,\mathcal{S}_i}$ is a unit second-order cone with suitable dimension.
Denote $\xi_{\mathcal{S}_i}$ as the number of the second-order constraints in the subproblem $\mathcal{S}_i$.
Then all inequality constraints for the control variables of the subproblem $\mathcal{S}_i$ can be described by
\begin{equation}\label{eq_localinequalities_subproblem}
\textbf{D}_{\mathcal{S}_i}\textbf{x}_{\mathcal{S}_i} - \textbf{d}_{\mathcal{S}_i} \in \Omega_{\mathcal{S}_i}
\end{equation}
where $\textbf{D}_{\mathcal{S}_i} = [\textbf{D}_{0, \mathcal{S}_i}^T, \textbf{D}_{1, \mathcal{S}_i}^T, \dots, \textbf{D}_{\xi_{\mathcal{S}_i}, \mathcal{S}_i}^T]^T$, $\textbf{d}_{\mathcal{S}_i} = [\textbf{d}_{0, \mathcal{S}_i}^T, \textbf{d}_{1, \mathcal{S}_i}^T, \dots, \textbf{d}_{\xi_{\mathcal{S}_i}, \mathcal{S}_i}^T]^T$ and $\Omega_{\mathcal{S}_i} = \mathbb{R}_{-}^{\xi_{0,\mathcal{S}_i}} \times \bigtimes_{m = 1}^{\xi_{\mathcal{S}_i}} \mathcal{C}_{m,\mathcal{S}_i}$. Here $\times$ denotes the Cartesian product of two sets and $\bigtimes_{m = 1}^{\xi_{\mathcal{S}_i}} \mathcal{C}_{m,\mathcal{S}_i} = \mathcal{C}_{1,\mathcal{S}_i} \times \cdots \times \mathcal{C}_{\xi_{\mathcal{S}_i},\mathcal{S}_i}$
%From inequality constraints (\eqref{eq_soc_constraint01}-\eqref{eq_downstream_constraint}), it is easy to verify that the matrix $\textbf{D}_{\mathcal{S}_i}$ has full column rank.
The equality constraints \eqref{eq_SMPC_TVPredicted} corresponding to all road links $z$, where $\sigma(z) \in \mathcal{J}_{\mathcal{S}_i} \cup \mathcal{O}_{\mathcal{S}_i}$ for all $k = 0, \dots, K - 1$, can be stacked into the following linear equation
\begin{equation}\label{eq_localequalities_subproblem}
\textbf{M}_{\mathcal{S}_i}\textbf{x}_{\mathcal{S}_i} = \textbf{m}_{\mathcal{S}_i}
\end{equation}
where $\textbf{M}_{\mathcal{S}_i} = \left[[\textbf{M}_{\mathcal{S}_i}]_{mn}\right] \in \mathbb{R}^{K(\zeta_{1,\mathcal{S}_i}+\zeta_{3,\mathcal{S}_i}) \times \zeta_{\mathcal{S}_i}}$ and vector $\textbf{m}_{\mathcal{S}_i} = \left[[\textbf{m}_{\mathcal{S}_i}]_{m}\right] \in \mathbb{R}^{K(\zeta_{1,\mathcal{S}_i}+\zeta_{3,\mathcal{S}_i})}$ by\\
\scalebox{0.8}{$[\textbf{M}_{\mathcal{S}_i}]_{mn} = \left\{ \begin{matrix*}[l]
1, & \textrm{if } m = l+k(\zeta_{1,\mathcal{S}_i} + \zeta_{3,\mathcal{S}_i}), [\textbf{x}_{\mathcal{S}_i}]_n \equiv [\hat{\textbf{n}}_{\mathcal{S}_i,k}]_l \equiv \hat{n}_{z,k}, 0 < l \le \zeta_{1,\mathcal{S}_i} + \zeta_{3,\mathcal{S}_i}, 0 \le k \le K-1\\
-1, & \textrm{if } m = l+k(\zeta_{1,\mathcal{S}_i} + \zeta_{3,\mathcal{S}_i}), [\textbf{x}_{\mathcal{S}_i}]_n \equiv [\hat{\textbf{n}}_{\mathcal{S}_i,k-1}]_l \equiv \hat{n}_{z,k-1}, 0 < l \le \zeta_{1,\mathcal{S}_i} + \zeta_{3,\mathcal{S}_i}, 0 < k \le K-1\\
1, & \textrm{if } m = l+k(\zeta_{1,\mathcal{S}_i} + \zeta_{3,\mathcal{S}_i}), [\textbf{x}_{\mathcal{S}_i}]_n \equiv [\hat{\textbf{q}}_{\mathcal{S}_i,k}]_l \equiv \hat{q}_{z,k}, 0 < l \le \zeta_{1,\mathcal{S}_i}, 0 \le k \le K-1\\
1, & \textrm{if } m = l+\zeta_{1,\mathcal{S}_i}+k(\zeta_{1,\mathcal{S}_i} + \zeta_{3,\mathcal{S}_i}), [\textbf{x}_{\mathcal{S}_i}]_n \equiv \hat{q}_{z,k}^{copy} \textrm{ where } [\hat{\textbf{n}}_{\mathcal{S}_i,k}]_{l+\zeta_{1,\mathcal{S}_i}} \equiv \hat{n}_{z,k}, 0 < l \le \zeta_{3,\mathcal{S}_i}, 0 \le k \le K-1\\
-r_{wz,k}, & \textrm{if } m = l+k(\zeta_{1,\mathcal{S}_i} + \zeta_{3,\mathcal{S}_i}), [\textbf{x}_{\mathcal{S}_i}]_n \equiv \hat{q}_{w,k} \textrm{ where } [\hat{\textbf{n}}_{\mathcal{S}_i,k}]_{l} \equiv \hat{n}_{z,k}, 0 < l \le \zeta_{1,\mathcal{S}_i} + \zeta_{3,\mathcal{S}_i}, 0 \le k \le K-1\\
0, & otherwise.
\end{matrix*} \right.$}
\[[\textbf{m}_{\mathcal{S}_i}]_{m} = \left\{ \begin{matrix*}[l]
n_z(t) + e_{z,0}, & \textrm{if } m \le \zeta_{1,\mathcal{S}_i} + \zeta_{3,\mathcal{S}_i}, [\hat{\textbf{n}}_{\mathcal{S}_i,0}]_{m} \equiv \hat{n}_{z,0},\\
e_{z,k}, & \textrm{if } m = l+k(\zeta_{1,\mathcal{S}_i} + \zeta_{3,\mathcal{S}_i}), [\hat{\textbf{n}}_{\mathcal{S}_i,k}]_{l} \equiv \hat{n}_{z,k}, 1 \le l \le \zeta_{1,\mathcal{S}_i} + \zeta_{3,\mathcal{S}_i}, 1 \le k \le K-1.
\end{matrix*} \right.\]
Since the rows in $\textbf{M}_{\mathcal{S}_i}$ correspond to different road links and different time control steps, they are independent or $\textbf{M}_{\mathcal{S}_i}$ is a full row rank matrix.
%%%%%%%%%%%%%%%%%%%%%%%%%%%%%%%%%%%%%%%%%%%%%%%%%%%%%%%%%%%%%%%%%%%%%%%%%%%%%%%%%%%%%%%%%%
\subsubsection{Distributed stochastic MPC traffic signal problem}
%%%%%%%%%%%%%%%%%%%%%%%%%%%%%%%%%%%%%%%%%%%%%%%%%%%%%%%%%%%%%%%%%%%%%%%%%%%%%%%%%%%%%%%%%%
Summarizing the formulation in this subsection, we have the constrained optimization problem \eqref{eq_distributedSMPC} as the distributed version of the stochastic MPC traffic signal control problem \eqref{eq_SMPCProblem_detailedform}.
\begin{equation}\label{eq_distributedSMPC}
\min\limits_{\textbf{x}_{\mathcal{S}_1},\dots,\textbf{x}_{\mathcal{S}_N}} \textrm{ } \sum\limits_{i=1}^{N} \Phi_{\mathcal{S}_i}(\textbf{x}_{\mathcal{S}_i}) \textrm{ s.t. } \left\{ \begin{matrix*}[l]
\eqref{eq_localcost_subproblem}, \eqref{eq_localinequalities_subproblem}, \eqref{eq_localequalities_subproblem},  \forall \mathcal{S}_i, \\
\eqref{eq_coupledequalities_subproblem}, \forall \mathcal{S}_j \in \mathcal{N}_{\mathcal{S}_i}, \forall \mathcal{S}_i.
\end{matrix*}\right.
\end{equation}
To apply the proximal ADMM scheme \eqref{eq_ADMM} for solving the optimization problem \eqref{eq_distributedSMPC}, we rewrite this problem in the form of \eqref{eq_ADMMproblem}.
Let $\textbf{y}_{\mathcal{S}_i} = \textbf{D}_{\mathcal{S}_i}\textbf{x}_{\mathcal{S}_i} - \textbf{d}_{\mathcal{S}_i}$ and $\textbf{y}_{\mathcal{S}_i\mathcal{S}_j}^{(1)} = \textbf{P}_{\mathcal{S}_i\mathcal{S}_j}\textbf{x}_{\mathcal{S}_i}, \textbf{y}_{\mathcal{S}_i\mathcal{S}_j}^{(2)} = \textbf{Q}_{\mathcal{S}_i\mathcal{S}_j}\textbf{x}_{\mathcal{S}_i}, \forall \mathcal{S}_j \in \mathcal{N}_{\mathcal{S}_i}$ and define the set $\mathcal{X}_{\mathcal{S}_i}$ of all local equality constraints for each subproblem $\mathcal{S}_i$, the set $\Omega_{\mathcal{S}_i\mathcal{S}_j}$ of all coupled constraints between two agent $\mathcal{S}_i$ and $\mathcal{S}_j$ as follows.
\[\mathcal{X}_{\mathcal{S}_i} = \left\{\textbf{x}_{\mathcal{S}_i}: \textbf{M}_{\mathcal{S}_i}\textbf{x}_{\mathcal{S}_i} = \textbf{m}_{\mathcal{S}_i}\right\};\]
\[\Omega_{\mathcal{S}_i\mathcal{S}_j} = \left\{\textbf{y}_{\mathcal{S}_i\mathcal{S}_j}^{(1)}, \textbf{y}_{\mathcal{S}_j\mathcal{S}_i}^{(2)}: \textbf{y}_{\mathcal{S}_i\mathcal{S}_j}^{(1)} = \textbf{y}_{\mathcal{S}_j\mathcal{S}_i}^{(2)}\right\}\]
Then we have the optimization problem \eqref{eq_SMPC_ADMMconsensus} which is equivalent to the distributed stochastic MPC traffic signal control problem \eqref{eq_distributedSMPC}.
\begin{subequations}\label{eq_SMPC_ADMMconsensus}
\begin{align}
\min &\textrm{ } \sum\limits_{i=1}^{N} \left( \frac{1}{2}\textbf{x}_{\mathcal{S}_i}^T\textbf{H}_{\mathcal{S}_i}\textbf{x}_{\mathcal{S}_i} + \textbf{h}_{\mathcal{S}_i}^T\textbf{x}_{\mathcal{S}_i} \right)\\
\textrm{s.t. }& \textbf{x}_{\mathcal{S}_i} \in \mathcal{X}_{\mathcal{S}_i}, \textbf{y}_{\mathcal{S}_i} \in \Omega_{\mathcal{S}_i}, \forall i =  1, \dots, N,\\
&(\textbf{y}_{\mathcal{S}_i\mathcal{S}_j}^{(1)}, \textbf{y}_{\mathcal{S}_j\mathcal{S}_i}^{(2)}) \in \Omega_{\mathcal{S}_i\mathcal{S}_j}, \forall \mathcal{S}_j \in \mathcal{N}_{\mathcal{S}_i}, \forall i =  1, \dots, N,\\
&\textbf{D}_{\mathcal{S}_i}\textbf{x}_{\mathcal{S}_i} - \textbf{d}_{\mathcal{S}_i} = \textbf{y}_{\mathcal{S}_i}, \forall i =  1, \dots, N,\\
&\textbf{y}_{\mathcal{S}_i\mathcal{S}_j}^{(1)} = \textbf{P}_{\mathcal{S}_i\mathcal{S}_j}\textbf{x}_{\mathcal{S}_i}, \forall \mathcal{S}_j \in \mathcal{N}_{\mathcal{S}_i}, \forall i =  1, \dots, N,\\
&\textbf{y}_{\mathcal{S}_i\mathcal{S}_j}^{(2)} = \textbf{Q}_{\mathcal{S}_i\mathcal{S}_j}\textbf{x}_{\mathcal{S}_i}, \forall \mathcal{S}_j \in \mathcal{N}_{\mathcal{S}_i}, \forall i =  1, \dots, N.
\end{align}
\end{subequations}
In this reformulation, the local equality constraint \eqref{eq_localequalities_subproblem} is converted into $\textbf{x}_{\mathcal{S}_i} \in \mathcal{X}_{\mathcal{S}_i}$;
the local inequality constraint \eqref{eq_localinequalities_subproblem} is replaced by the equality constraint (\ref{eq_SMPC_ADMMconsensus}d) combined with the constraint $\textbf{y}_{\mathcal{S}_i} \in \Omega_{\mathcal{S}_i}$; similarly, the equality constraints (\ref{eq_SMPC_ADMMconsensus}e), (\ref{eq_SMPC_ADMMconsensus}f) and the constraint (\ref{eq_SMPC_ADMMconsensus}c) are used to replace the coupled constraints among neighboring agents.
For the coordination among agents in solving the optimization problem \eqref{eq_SMPC_ADMMconsensus}, we make the following assumption:
\begin{Assumption}
Every agent $\mathcal{S}_i$ can exchange information with its neighbors in $\mathcal{N}_{\mathcal{S}_i}$
\end{Assumption}
%%%%%%%%%%%%%%%%%%%%%%%%%%%%%%%%%%%%%%%%%%%%%%%%%%%%%%%%%%%%%%%%%%%%%%%%%%%%%%%%%%%%%%%%%%
\subsection{Distributed solution based ADMM}
%%%%%%%%%%%%%%%%%%%%%%%%%%%%%%%%%%%%%%%%%%%%%%%%%%%%%%%%%%%%%%%%%%%%%%%%%%%%%%%%%%%%%%%%%%
Define the stacked vectors $\textbf{x} = \left[ \begin{matrix}\textbf{x}_{\mathcal{S}_1}^T & \textbf{x}_{\mathcal{S}_2}^T & \cdots & \textbf{x}_{\mathcal{S}_N}^T\end{matrix} \right]^T$ and $\textbf{y} = \left[ \begin{matrix}\tilde{\textbf{y}}_{\mathcal{S}_1}^T & \tilde{\textbf{y}}_{\mathcal{S}_2}^T & \cdots & \tilde{\textbf{y}}_{\mathcal{S}_N}^T\end{matrix} \right]^T$ where $\tilde{\textbf{y}}_{\mathcal{S}_i} = \left[ \begin{matrix}\textbf{y}_{\mathcal{S}_1}^T & row\left\{[\textbf{y}_{\mathcal{S}_i\mathcal{S}_j}^{(1)T}, \textbf{y}_{\mathcal{S}_i\mathcal{S}_j}^{(2)T}]: \mathcal{S}_j \in \mathcal{N}_{\mathcal{S}_i}\right\}\end{matrix} \right]^T$.
It is clear that the problem \eqref{eq_SMPC_ADMMconsensus} has the form of \eqref{eq_ADMMproblem} with
\[\Psi_x(\textbf{x}) = \sum\limits_{i=1}^{N}\left\{ \frac{1}{2} \textbf{x}_{\mathcal{S}_i}^T\textbf{H}_{\mathcal{S}_i}\textbf{x}_{\mathcal{S}_i} + \textbf{h}_{\mathcal{S}_i}^T\textbf{x}_{\mathcal{S}_i}\right\}, \Psi_y(\textbf{y}) = 0,\]
$\textbf{A} = blkdiag\left\{\textbf{A}_{\mathcal{S}_1}, \textbf{A}_{\mathcal{S}_2}, \dots, \textbf{A}_{\mathcal{S}_N}\right\}$  where
$\textbf{A}_{\mathcal{S}_i} = blkdiag\left\{\textbf{D}_{\mathcal{S}_i}, blkdiag\left\{\textbf{P}_{\mathcal{S}_i\mathcal{S}_j}, \textbf{Q}_{\mathcal{S}_i\mathcal{S}_j}: \mathcal{S}_j \in \mathcal{N}_{\mathcal{S}_i}\right\}\right\}$, $\textbf{B} = -\textbf{I}$ where $\textbf{I}$ is an unity matrix, $\textbf{c} = col\left\{\textbf{c}_{\mathcal{S}_1}, \textbf{c}_{\mathcal{S}_2}, \dots, \textbf{c}_{\mathcal{S}_N}\right\}$ where $\textbf{c}_{\mathcal{S}_i} = col\left\{\textbf{d}_{\mathcal{S}_i}, \textbf{0}\right\}$, $\textbf{0}$ is the vector of all elements zero, and $\mathcal{X} = \bigtimes_{i = 1}^{N} \mathcal{X}_{\mathcal{S}_i}$, $\mathcal{Y} = \bigtimes_{i = 1}^{N} \left( \Omega_{\mathcal{S}_i} \times \bigtimes_{\mathcal{S}_j \in \mathcal{N}_{\mathcal{S}_i}} \Omega_{\mathcal{S}_i\mathcal{S}_j}\right)$.
Let $\rho$ be an arbitrarily positive constant and $\eta_{\mathcal{S}_i}, \forall i = 1, \dots, N,$ be a positive constant bigger than the largest eigenvalue of the matrix $\textbf{A}_{\mathcal{S}_i}^T\textbf{A}_{\mathcal{S}_i}, \forall i = 1, \dots, N$.
Define $\textbf{G}_{\mathcal{S}_i} = \eta_{\mathcal{S}_i}\textbf{I} - \rho\textbf{A}_{\mathcal{S}_i}^T\textbf{A}_{\mathcal{S}_i}$. It is easy to see $\textbf{G}_{\mathcal{S}_i}$ is a positive definite matrix and it can be determined by local information of the agent $\mathcal{S}_i$.

We use $\boldsymbol{\lambda}_{\mathcal{S}_i}$, $\boldsymbol{\lambda}_{\mathcal{S}_i\mathcal{S}_j}^{(1)}$ and $\boldsymbol{\lambda}_{\mathcal{S}_i\mathcal{S}_j}^{(2)}$ to denote the dual variables corresponding to the equality constraint (\ref{eq_SMPC_ADMMconsensus}d), (\ref{eq_SMPC_ADMMconsensus}e) and (\ref{eq_SMPC_ADMMconsensus}f), respectively.
\begin{subequations}\label{eq_updatelaw}
\begin{align}
\textbf{x}_{\mathcal{S}_i}(l + 1) &= \hat{\textbf{H}}_{\mathcal{S}_i}^{-1} \hat{\textbf{h}}_{\mathcal{S}_i}(l) - \hat{\textbf{H}}_{\mathcal{S}_i}^{-1}\textbf{M}_{\mathcal{S}_i}^T \left(\textbf{M}_{\mathcal{S}_i}\hat{\textbf{H}}_{\mathcal{S}_i}^{-1}\textbf{M}_{\mathcal{S}_i}^T\right)^{-1} \left(\textbf{M}_{\mathcal{S}_i}\hat{\textbf{H}}_{\mathcal{S}_i}^{-1}\hat{\textbf{h}}_{\mathcal{S}_i} - \textbf{m}_{\mathcal{S}_i}\right),\\
\textrm{where }\hat{\textbf{H}}_{\mathcal{S}_i} &= \textbf{Q}_{\mathcal{S}_i} + \eta_{\mathcal{S}_i}\textbf{I} \nonumber\\
\hat{\textbf{h}}_{\mathcal{S}_i}(l) &= \left(\eta_i\textbf{I} - \rho\textbf{D}_{\mathcal{S}_i}^T\textbf{D}_{\mathcal{S}_i} - \sum\limits_{\mathcal{S}_j \in \mathcal{N}_{\mathcal{S}_i}} \rho\left(\textbf{P}_{\mathcal{S}_i\mathcal{S}_j}^T\textbf{P}_{\mathcal{S}_i\mathcal{S}_j} + \textbf{Q}_{\mathcal{S}_i\mathcal{S}_j}^T\textbf{Q}_{\mathcal{S}_i\mathcal{S}_j}\right)\right) \textbf{x}_{\mathcal{S}_i}(l) + \textbf{D}_{\mathcal{S}_i}^T \left( \boldsymbol{\lambda}_{\mathcal{S}_i}(l) + \rho\left(\textbf{y}_{\mathcal{S}_i}(l) + \textbf{d}_{\mathcal{S}_i}\right) \right)\nonumber\\
&+ \sum_{\mathcal{S}_j \in \mathcal{N}_{\mathcal{S}_i}} \left( \textbf{P}_{\mathcal{S}_i\mathcal{S}_j}^T \left( \boldsymbol{\lambda}_{\mathcal{S}_i\mathcal{S}_j}^{(1)}(l) + \rho\textbf{y}_{\mathcal{S}_i\mathcal{S}_j}^{(1)}(l) \right) + \textbf{Q}_{\mathcal{S}_i\mathcal{S}_j}^T \left( \boldsymbol{\lambda}_{\mathcal{S}_i\mathcal{S}_j}^{(2)}(l) + \rho\textbf{y}_{\mathcal{S}_i\mathcal{S}_j}^{(2)}(l) \right) \right) - \textbf{h}_{\mathcal{S}_i}\nonumber\\
\textbf{y}_{\mathcal{S}_i}(l + 1) &= Prj_{\Omega_{\mathcal{S}_i}}\left(\textbf{D}_{\mathcal{S}_i}\textbf{x}_{\mathcal{S}_i}(l + 1) - \textbf{d}_{\mathcal{S}_i} - \frac{1}{\rho}(\boldsymbol{\lambda}_{\mathcal{S}_i}(l)\right),\\
\boldsymbol{\lambda}_{\mathcal{S}_i}(l + 1) &= \boldsymbol{\lambda}_{\mathcal{S}_i}(l) - \rho\left(\textbf{D}_{\mathcal{S}_i}\textbf{x}_{\mathcal{S}_i}(l + 1) - \textbf{d}_{\mathcal{S}_i} - \textbf{y}_{\mathcal{S}_i}(l + 1)\right),\\
\textbf{y}_{\mathcal{S}_i\mathcal{S}_j}^{(1)}(l + 1) &= \frac{1}{2}\Bigl( \textbf{P}_{\mathcal{S}_i\mathcal{S}_j}\textbf{x}_{\mathcal{S}_i}(l + 1) - \frac{1}{\rho}\boldsymbol{\lambda}_{\mathcal{S}_i\mathcal{S}_j}^{(1)}(l) + \textbf{Q}_{\mathcal{S}_j\mathcal{S}_i}\textbf{x}_{\mathcal{S}_j}(l+1) - \frac{1}{\rho}\boldsymbol{\lambda}_{\mathcal{S}_j\mathcal{S}_i}^{(2)}(l) \Bigr), \forall \mathcal{S}_j \in \mathcal{N}_{\mathcal{S}_i},\\
\boldsymbol{\lambda}_{\mathcal{S}_i\mathcal{S}_j}^{(1)}(l + 1) &= \boldsymbol{\lambda}_{\mathcal{S}_i\mathcal{S}_j}^{(1)}(l) - \rho\left(\textbf{P}_{\mathcal{S}_i\mathcal{S}_j}\textbf{x}_{\mathcal{S}_i}(l + 1) - \textbf{y}_{\mathcal{S}_i\mathcal{S}_j}^{(1)}(l + 1)\right), \forall \mathcal{S}_j \in \mathcal{N}_{\mathcal{S}_i},\\
\textbf{y}_{\mathcal{S}_i\mathcal{S}_j}^{(2)}(l + 1) &= \frac{1}{2}\Bigl( \textbf{P}_{\mathcal{S}_j\mathcal{S}_i}\textbf{x}_{\mathcal{S}_j}(l + 1) - \frac{1}{\rho}\boldsymbol{\lambda}_{\mathcal{S}_j\mathcal{S}_i}^{(1)}(l) + \textbf{Q}_{\mathcal{S}_i\mathcal{S}_j}\textbf{x}_{\mathcal{S}_i}(l+1) - \frac{1}{\rho}\boldsymbol{\lambda}_{\mathcal{S}_i\mathcal{S}_j}^{(2)}(l) \Bigr), \forall \mathcal{S}_j \in \mathcal{N}_{\mathcal{S}_i},\\
\boldsymbol{\lambda}_{\mathcal{S}_i\mathcal{S}_j}^{(2)}(l + 1) &= \boldsymbol{\lambda}_{\mathcal{S}_i\mathcal{S}_j}^{(2)}(l) - \rho\left(\textbf{Q}_{\mathcal{S}_i\mathcal{S}_j}\textbf{x}_{\mathcal{S}_i}(l + 1) - \textbf{y}_{\mathcal{S}_i\mathcal{S}_j}^{(2)}(l + 1)\right), \forall \mathcal{S}_j \in \mathcal{N}_{\mathcal{S}_i}.
\end{align}
\end{subequations}
By applying ADMM scheme \eqref{eq_ADMM}, which is summarized in Appendix \ref{subADMM}, to solve the problem \eqref{eq_SMPC_ADMMconsensus}, with $\textbf{G} = blkdiag\left\{\textbf{G}_{\mathcal{S}_i}: i = 1, \dots, N\right\}$, we derive the update law \eqref{eq_updatelaw} for all $\mathcal{S}_i, i = 1, \dots, N$, to find the optimal solution of the distributed stochastic MPC traffic signal control problem \eqref{eq_distributedSMPC}.
In the update (\ref{eq_updatelaw}b) for the variable $\textbf{y}_{\mathcal{S}_i}$, $Prj_{\Omega_{\mathcal{S}_i}}(\cdot)$ is the projection of a vector in $(\cdot)$ onto the set $\Omega_{\mathcal{S}_i}$.
In this paper, the set $\Omega_{\mathcal{S}_i}$ is a Cartersian product of bounded sets and second-order cone.
The projections on these sets are given in Appendix \ref{subProjBC}.
\begin{Theorem}\label{th_updatelaw}
Under the assumption that $\Omega_t \neq \emptyset$, the optimal solution of the distributed stochastic MPC traffic signal control problem \eqref{eq_distributedSMPC} is achieved asymptotically by the update law \eqref{eq_updatelaw} and the convergence rate is $o\left(\frac{1}{l}\right)$.
\end{Theorem}
\begin{proof}
Matching the distributed stochastic MPC traffic signal control problem \eqref{eq_SMPC_ADMMconsensus} with the form of \eqref{eq_ADMMproblem}, we have the function to be minimized in the update (\ref{eq_ADMM}a) is $\psi_{x}(\textbf{x},l) = \sum_{i = 1}^{N} \psi_{x,\mathcal{S}_i}(\textbf{x}_{\mathcal{S}_i},l)$ where $\psi_{x,\mathcal{S}_i}(\textbf{x}_{\mathcal{S}_i},l) = \frac{1}{2} \textbf{x}_{\mathcal{S}_i}^T \hat{\textbf{H}}_{\mathcal{S}_i} \textbf{x}_{\mathcal{S}_i} - \left(\hat{\textbf{h}}_{\mathcal{S}_i}(l)\right)^T\textbf{x}_{\mathcal{S}_i} + const$ with $\hat{\textbf{H}}_{\mathcal{S}_i}$ and $\hat{\textbf{h}}_{\mathcal{S}_i}$ are defined in (\ref{eq_updatelaw}a).
Then the update (\ref{eq_ADMM}a) is equivalent to the optimal solutions of $N$ constrained optimization problems \eqref{eq_temp1OptimizationProblem} as follows.
\begin{equation}\label{eq_temp1OptimizationProblem}
\min\limits_{\textbf{M}_{\mathcal{S}_i}\textbf{x}_{\mathcal{S}_i} = \textbf{m}_{\mathcal{S}_i}} \left\{\frac{1}{2} \textbf{x}_{\mathcal{S}_i}^T \hat{\textbf{H}}_{\mathcal{S}_i} \textbf{x}_{\mathcal{S}_i} - \left(\hat{\textbf{h}}_{\mathcal{S}_i}(l)\right)^T\textbf{x}_{\mathcal{S}_i}\right\}
\end{equation}
for all $i = 1, \dots, N$.
It is easy to verify that the problem \eqref{eq_temp1OptimizationProblem} is a quadratic program with equality constraints.
Its optimal solution can be found by solving the linear equation \eqref{eq_temp1OptimizationConditions}, which corresponds to its KKT conditions:
\begin{subequations}\label{eq_temp1OptimizationConditions}
\begin{align}
\hat{\textbf{H}}_{\mathcal{S}_i} \textbf{x}_{\mathcal{S}_i} - \hat{\textbf{h}}_{\mathcal{S}_i}(l) + \textbf{M}_{\mathcal{S}_i}^T \boldsymbol{\mu} &= \textbf{0}\\
\textbf{M}_{\mathcal{S}_i}\textbf{x}_{\mathcal{S}_i} &= \textbf{m}_{\mathcal{S}_i}
\end{align}
\end{subequations}
where $\boldsymbol{\mu}$ is the dual variable vector corresponding to the equality constraint in \eqref{eq_temp1OptimizationProblem}.
Since $\textbf{H}_{\mathcal{S}_i}$ is a positive semidefinite matrix, the matrix $\hat{\textbf{H}}_{\mathcal{S}_i}$ is strictly positive definite.
We also note that the matrix $\textbf{M}_{\mathcal{S}_i}$ is full row rank.
By solving the linear equation \eqref{eq_temp1OptimizationConditions}, we obtain the optimal solution
\begin{align*}
\textbf{x}_{\mathcal{S}_i}^* &= \hat{\textbf{H}}_{\mathcal{S}_i}^{-1} \left(\hat{\textbf{h}}_{\mathcal{S}_i} - \textbf{M}_{\mathcal{S}_i}^T \boldsymbol{\mu}^*\right),\\
\boldsymbol{\mu}^* &= \left(\textbf{M}_{\mathcal{S}_i}\hat{\textbf{H}}_{\mathcal{S}_i}^{-1}\textbf{M}_{\mathcal{S}_i}^T\right)^{-1} \left(\textbf{M}_{\mathcal{S}_i}\hat{\textbf{H}}_{\mathcal{S}_i}^{-1}\hat{\textbf{h}}_{\mathcal{S}_i} - \textbf{m}_{\mathcal{S}_i}\right).
\end{align*}
Thus the ADMM iteration step (\ref{eq_ADMM}a) for the update of $\textbf{x}$ is equivalent to the equation (\ref{eq_updatelaw}a) for the update of all $\textbf{x}_{\mathcal{S}_i}, i = 1, \dots, N$.
Consider the update (\ref{eq_ADMM}b), we have the cost function to be minimized is $\sum\limits_{i = 1}^{N} \left\{\frac{\rho}{2}\left|\left|\textbf{D}_{\mathcal{S}_i}\textbf{x}_{\mathcal{S}_i}(l+1) - \textbf{d}_{\mathcal{S}_i} - \textbf{y}_{\mathcal{S}_i} - \frac{1}{\rho}\boldsymbol{\lambda}_{\mathcal{S}_i}(l)\right|\right|^2\right\} + \sum\limits_{i = 1}^{N}\sum\limits_{\mathcal{S}_j \in \mathcal{N}_{\mathcal{S}_i}} \Bigl\{ \frac{\rho}{2}\left|\left|\textbf{P}_{\mathcal{S}_i\mathcal{S}_j}\textbf{x}_{\mathcal{S}_i}(l + 1) - \textbf{y}_{\mathcal{S}_i\mathcal{S}_j}^{(1)} - \frac{1}{\rho}\boldsymbol{\lambda}_{\mathcal{S}_i\mathcal{S}_j}^{(1)}(l)\right|\right|^2 + \frac{\rho}{2}\left|\left|\textbf{Q}_{\mathcal{S}_i\mathcal{S}_j}\textbf{x}_{\mathcal{S}_i}(l + 1) - \textbf{y}_{\mathcal{S}_i\mathcal{S}_j}^{(2)} - \frac{1}{\rho}\boldsymbol{\lambda}_{\mathcal{S}_i\mathcal{S}_j}^{(2)}(l)\right|\right|^2 \Bigr\}$.
Then the update (\ref{eq_ADMM}b) is equivalent to the optimization problem \eqref{eq_temp2OptimizationProblem} for all pair of neighboring agents $\mathcal{S}_i, \mathcal{S}_j$ and the optimization problem \eqref{eq_temp3OptimizationProblem} for all $\mathcal{S}_i, i = 1, \dots, N$, as follows.
\begin{equation}\label{eq_temp2OptimizationProblem}
\min\limits_{\textbf{y}_{\mathcal{S}_i\mathcal{S}_j}^{(1)} = \textbf{y}_{\mathcal{S}_j\mathcal{S}_i}^{(2)}} \left\{\frac{\rho}{2}\left|\left|\textbf{P}_{\mathcal{S}_i\mathcal{S}_j}\textbf{x}_{\mathcal{S}_i}(l + 1) - \textbf{y}_{\mathcal{S}_i\mathcal{S}_j}^{(1)} - \frac{1}{\rho}\boldsymbol{\lambda}_{\mathcal{S}_i\mathcal{S}_j}^{(1)}(l)\right|\right|^2 + \frac{\rho}{2}\left|\left|\textbf{Q}_{\mathcal{S}_j\mathcal{S}_i}\textbf{x}_{\mathcal{S}_j}(l + 1) - \textbf{y}_{\mathcal{S}_j\mathcal{S}_i}^{(2)} - \frac{1}{\rho}\boldsymbol{\lambda}_{\mathcal{S}_j\mathcal{S}_i}^{(2)}(l)\right|\right|^2 \right\}
\end{equation}
Since the problem \eqref{eq_temp2OptimizationProblem} has the same form as the problem \eqref{eq_temp1OptimizationProblem}, we use the same method to derive the optimal solution of the problem \eqref{eq_temp2OptimizationProblem} as $\textbf{y}_{\mathcal{S}_i\mathcal{S}_j}^{(1)}(l + 1)$ in (\ref{eq_updatelaw}d) and $\textbf{y}_{\mathcal{S}_i\mathcal{S}_j}^{(2)}(l + 1)$ in (\ref{eq_updatelaw}f).
\begin{equation}\label{eq_temp3OptimizationProblem}
\min\limits_{\textbf{y}_{\mathcal{S}_i} \in \Omega_{\mathcal{S}_i}} \left\{\frac{\rho}{2}\left|\left|\textbf{D}_{\mathcal{S}_i}\textbf{x}_{\mathcal{S}_i}(l+1) - \textbf{d}_{\mathcal{S}_i} - \textbf{y}_{\mathcal{S}_i} - \frac{1}{\rho}\boldsymbol{\lambda}_{\mathcal{S}_i}(l)\right|\right|_2^2\right\}
\end{equation}
Let $\textbf{y}_{\mathcal{S}_i}^*$ be the optimal solution and $\hat{\psi}(\textbf{y}_{\mathcal{S}_i})$ be the cost function of the problem \eqref{eq_temp3OptimizationProblem}.
From the optimality condition, we have $\textbf{y}_{\mathcal{S}_i}^* \in \Omega_{\mathcal{S}_i}$ and $\left(\nabla \hat{\psi}(\textbf{y}_{\mathcal{S}_i}^*)\right)^T\left( \textbf{y}_{\mathcal{S}_i} - \textbf{y}_{\mathcal{S}_i}^* \right) \ge 0, $ for all $\textbf{y}_{\mathcal{S}_i} \in \Omega_{\mathcal{S}_i}$. 
Since $\nabla \hat{\psi}(\textbf{y}_{\mathcal{S}_i}^*) = \rho\left(\textbf{y}_{\mathcal{S}_i}^* - \textbf{D}_{\mathcal{S}_i}\textbf{x}_{\mathcal{S}_i}(l+1) + \textbf{d}_{\mathcal{S}_i} \right) - \boldsymbol{\lambda}_{\mathcal{S}_i}(l)$, the optimality condition of the problem \eqref{eq_temp3OptimizationProblem} becomes
\[\left( \textbf{y}_{\mathcal{S}_i}^* - \left(\textbf{D}_{\mathcal{S}_i}\textbf{x}_{\mathcal{S}_i}(l+1) - \textbf{d}_{\mathcal{S}_i} - \frac{1}{\rho} \boldsymbol{\lambda}_{\mathcal{S}_i}(l)\right) \right)^T \left( \textbf{y}_{\mathcal{S}_i} - \textbf{y}_{\mathcal{S}_i}^* \right) \ge 0, \]
for all $\textbf{y}_{\mathcal{S}_i} \in \Omega_{\mathcal{S}_i}$.
In addition, $\Omega_{\mathcal{S}_i}$ is a convex set because it is a Cartersian product of a non-positive orthant and some unit second-order cones.
According to Theorem \ref{th_proj}, we have
\[\textbf{y}_{\mathcal{S}_i}^* = Prj_{\Omega_{\mathcal{S}_i}}\left(\textbf{D}_{\mathcal{S}_i}\textbf{x}_{\mathcal{S}_i}(l+1) - \textbf{d}_{\mathcal{S}_i} - \frac{1}{\rho} \boldsymbol{\lambda}_{\mathcal{S}_i}(l)\right).\]
This is the update (\ref{eq_updatelaw}b) for the variable $\textbf{y}_{\mathcal{S}_i}$.
Matching the equality constraints (\ref{eq_SMPC_ADMMconsensus}d), (\ref{eq_SMPC_ADMMconsensus}e) and (\ref{eq_SMPC_ADMMconsensus}f) with the form of (\ref{eq_ADMMproblem}b), we have the ADMM iteration (\ref{eq_ADMM}c) for the update of dual variables is equivalent to the equations (\ref{eq_updatelaw}c), (\ref{eq_updatelaw}e) and (\ref{eq_updatelaw}g) for all $\mathcal{S}_i, i = 1, \dots, N$.
Thus the update law \eqref{eq_updatelaw} coincides with ADMM iterations \eqref{eq_ADMM} for the problem \eqref{eq_SMPC_ADMMconsensus}.
The asymptotic convergence of ADMM iterations to the optimal solution of the problem \eqref{eq_SMPC_ADMMconsensus} is guaranteed by Theorem \ref{th_ADMM}.
Since the problem \eqref{eq_SMPC_ADMMconsensus} is equivalent to the distributed stochastic MPC traffic signal control problem \eqref{eq_distributedSMPC}, their optimal solutions are the same.
\end{proof}

\begin{algorithm}[htb]
\begin{algorithmic}[1]
\BState \emph{Collect traffic data and update its information in \eqref{eq_SMPC_ADMMconsensus}:} $\textbf{H}_{\mathcal{S}_i}, \textbf{h}_{\mathcal{S}_i}, \textbf{M}_{\mathcal{S}_i}, \textbf{m}_{\mathcal{S}_i}, \textbf{D}_{\mathcal{S}_i}, \textbf{d}_{\mathcal{S}_i}, \textbf{P}_{\mathcal{S}_i\mathcal{S}_j}, \textbf{Q}_{\mathcal{S}_i\mathcal{S}_j}$
\BState \emph{Coordinate with neighbors to solve \eqref{eq_SMPC_ADMMconsensus}:}
\State $\quad$\textbf{Initialization} $l \gets 0$.
\State $\quad$ Choose arbitrarily $\textbf{y}_{\mathcal{S}_i}(0) \in \Omega_{\mathcal{S}_i}$, $\boldsymbol{\lambda}_{\mathcal{S}_i}(0)$ and $\textbf{y}_{\mathcal{S}_i\mathcal{S}_j}^{(1)}(0)$, $\boldsymbol{\lambda}_{\mathcal{S}_i\mathcal{S}_j}^{(1)}(0)$, $\textbf{y}_{\mathcal{S}_i\mathcal{S}_j}^{(2)}(0)$, $\boldsymbol{\lambda}_{\mathcal{S}_i\mathcal{S}_j}^{(2)}(0)$, $\forall \mathcal{S}_j \in \mathcal{N}_{\mathcal{S}_i}$.
\State $\quad$\textbf{repeat:}
\State $\qquad$ Run the update $\left(\textbf{x}_{\mathcal{S}_i}, \textbf{y}_{\mathcal{S}_i}, \boldsymbol{\lambda}_{\mathcal{S}_i}, \textbf{y}_{\mathcal{S}_i\mathcal{S}_j}^{(1)}, \boldsymbol{\lambda}_{\mathcal{S}_i\mathcal{S}_j}^{(1)}, \textbf{y}_{\mathcal{S}_i\mathcal{S}_j}^{(2)}, \boldsymbol{\lambda}_{\mathcal{S}_i\mathcal{S}_j}^{(2)} \right)(l + 1) \gets \eqref{eq_updatelaw}$
\State $\qquad$ Check the termination condition \eqref{eq_terminated} and set $fl_{\mathcal{S}_i}$.
\State $\qquad$ Run $N$ steps of the min-consensus \eqref{eq_minconsensus}.
\State $\qquad$ If $fl_{\mathcal{S}_i}(N) = 1$ then set $l^{final} \gets l$ and \textbf{Stop} else set $l \gets l+1$ and continue \textbf{repeat}
\BState \emph{Finish algorithm } $\textbf{x}_{\mathcal{S}_i}^{opt} \gets \textbf{x}_{\mathcal{S}_i}(l^{final})$
\end{algorithmic}
\caption{Distributed method for an agent $\mathcal{S}_i$ to find $\textbf{x}_{\mathcal{S}_i}^{opt}$.}\label{alg_proposed_control}
\end{algorithm}
Since the convergence of ADMM iterations is asymptotic, we need a stopping criteria in real-time application.
In \cite{StephenBoyd2011}, the authors suggested reasonable termination criterions as $||\textbf{A}\textbf{x}(l) + \textbf{B}\textbf{y}(l) - \textbf{c}||_2 \le tol_p$ and/or $\rho ||\textbf{A}^T\textbf{B}\left(\textbf{y}(l+1) - \textbf{y}(l)\right)||_2 \le tol_d$ where $tol_p$ and $tol_d$ are given small positive tolerances.
In this paper, we use the maximum norm instead of the Euclidean norm for an implementation in distributed manner.
It should be noted that $||\textbf{u}||_{\infty} \le ||\textbf{u}||_2 \le \sqrt{n}||\textbf{u}||_{\infty}$ for any vector $\textbf{u} \in \mathbb{R}^n$.
The update \eqref{eq_updatelaw} is terminated when the following condition \eqref{eq_terminated} is satisfied for every agent $\mathcal{S}_i$.
\begin{subequations}\label{eq_terminated}
\begin{gather}
||\textbf{D}_{\mathcal{S}_i}\textbf{x}_{\mathcal{S}_i}(l) - \textbf{y}_{\mathcal{S}_i}(l) - \textbf{d}_{\mathcal{S}_i}||_{\infty} \le tol,\\
||\textbf{P}_{\mathcal{S}_i\mathcal{S}_j}\textbf{x}_{\mathcal{S}_i}(l) - \textbf{y}_{\mathcal{S}_i\mathcal{S}_j}^{(1)}(l)||_{\infty} \le tol, \forall \mathcal{S}_j \in \mathcal{N}_{\mathcal{S}_i},\\
||\textbf{Q}_{\mathcal{S}_i\mathcal{S}_j}\textbf{x}_{\mathcal{S}_i}(l) - \textbf{y}_{\mathcal{S}_i\mathcal{S}_j}^{(2)}(l)||_{\infty} \le tol, \forall \mathcal{S}_j \in \mathcal{N}_{\mathcal{S}_i}.
\end{gather}
\end{subequations}
To check this condition, each agent $\mathcal{S}_i$ can use a flag $fl_{\mathcal{S}_i}$ and run the following min-consensus law:
\begin{equation}\label{eq_minconsensus}
fl_{\mathcal{S}_i}(\varsigma+1) = \min\left\{fl_{\mathcal{S}_j}(\varsigma): \mathcal{S}_j \in \mathcal{N}_{\mathcal{S}_i} \cup \{\mathcal{S}_i\}\right\}
\end{equation}
where $fl_{\mathcal{S}_i}(0) = 1$ if \eqref{eq_terminated} is satisfied and $fl_{\mathcal{S}_i}(0) = 0$ otherwise.
It is well-known that the min-consensus \eqref{eq_minconsensus} will converge to the minimum of initial values after a finite steps, which is smaller than $N$.
In other words, $fl_{\mathcal{S}_i}(N) = \min\{fl_{\mathcal{S}_1}(0), fl_{\mathcal{S}_2}(0), \cdots, fl_{\mathcal{S}_N}(0) \}$ for all $i = 1, \dots, N$.
To conclude this subsection, we provide Algorithm \ref{alg_proposed_control} as distributed method for each agent $\mathcal{S}_i$ to find the optimal solution of the stochastic MPC traffic signal control problem \eqref{eq_SMPCProblem_detailedform}.
\begin{Remark}
In the iteration \eqref{eq_updatelaw}, agent $\mathcal{S}_i$ uses only its own information for the update of control variables (i.e., the update of $\textbf{x}_{\mathcal{S}_i}$ in (\ref{eq_updatelaw}a)), the update of variables corresponding to local inequality constraints (i.e., the update of $\textbf{y}_{\mathcal{S}_i}$ in (\ref{eq_updatelaw}b)) and the update of dual variables corresponding to equality constraints (i.e., the update of $\boldsymbol{\lambda}_{\mathcal{S}_i}, \boldsymbol{\lambda}_{\mathcal{S}_i\mathcal{S}_j}^{(1)}$ and $\boldsymbol{\lambda}_{\mathcal{S}_i\mathcal{S}_j}^{(2)}$ in (\ref{eq_updatelaw}c), (\ref{eq_updatelaw}e) and (\ref{eq_updatelaw}g)).
For the variables corresponding to the coupled constraints (i.e., $\textbf{y}_{\mathcal{S}_i\mathcal{S}_j}^{(1)}(l)$ and $\textbf{y}_{\mathcal{S}_i\mathcal{S}_j}^{(2)}(l)$ in (\ref{eq_updatelaw}d) and (\ref{eq_updatelaw}f)), where $\mathcal{S}_j \in \mathcal{N}_{\mathcal{S}_i}$), the updates depend on the information of the agent $\mathcal{S}_i$ and its neighboring agents.
So, the communication between the agent $\mathcal{S}_i$ and its neighboring agent $\mathcal{S}_j \in \mathcal{N}_{\mathcal{S}_i}$ corresponding to the iteration $l+1$ is expressed as follows: it sends the vectors $\textbf{P}_{\mathcal{S}_i\mathcal{S}_j}\textbf{x}_{\mathcal{S}_i}(l+1) - \frac{1}{\rho}\boldsymbol{\lambda}_{\mathcal{S}_i\mathcal{S}_j}^{(1)}(l)$ and $\textbf{Q}_{\mathcal{S}_i\mathcal{S}_j}\textbf{x}_{\mathcal{S}_i}(l+1) - \frac{1}{\rho}\boldsymbol{\lambda}_{\mathcal{S}_i\mathcal{S}_j}^{(2)}(l)$ to $\mathcal{S}_j$ and receives the vectors $\textbf{P}_{\mathcal{S}_j\mathcal{S}_i}\textbf{x}_{\mathcal{S}_j}(l+1) - \frac{1}{\rho}\boldsymbol{\lambda}_{\mathcal{S}_j\mathcal{S}_i}^{(1)}(l)$ and $\textbf{Q}_{\mathcal{S}_j\mathcal{S}_i}\textbf{x}_{\mathcal{S}_j}(l+1) - \frac{1}{\rho}\boldsymbol{\lambda}_{\mathcal{S}_j\mathcal{S}_i}^{(2)}(l)$ from $\mathcal{S}_j$.
In addition, the exchange of the flags to check the termination condition \eqref{eq_terminated}, i.e., $fl_{\mathcal{S}_i}, i = 1, \dots, N,$ is also required among neighboring agents.
Thus, Algorithm \ref{alg_proposed_control} is fully distributed.
\end{Remark}
%%%%%%%%%%%%%%%%%%%%%%%%%%%%%%%%%%%%%%%%%%%%%%%%%%%%%%%%%%%%%%%%%%%%%%%%%%%%%%%%%%%%%%%%%%
\subsection{Strategy for optimal traffic signal control}
%%%%%%%%%%%%%%%%%%%%%%%%%%%%%%%%%%%%%%%%%%%%%%%%%%%%%%%%%%%%%%%%%%%%%%%%%%%%%%%%%%%%%%%%%%
The main purpose of this paper is to provide a distributed stochastic MPC traffic signal control and coordination method for urban networks.
%%%%%%%%%%%%%%%%%%%%%%%%%%%%%%%%%%%%%%%%%%%%%%%%%%%%%%%%%%%%%%%%%%%%%%%%%%%%%%%%%%%%%%%%%%
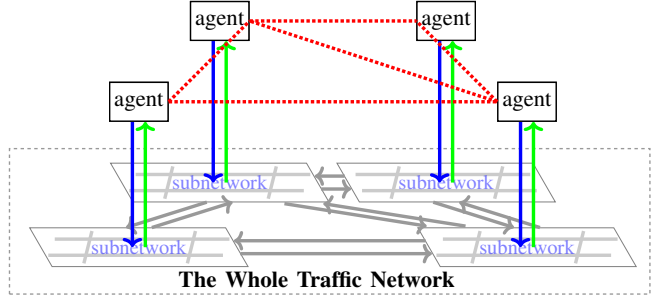
\begin{figure}[htb]
\begin{center}
\centering
\scalebox{0.43}{\begin{tikzpicture}[
textnode/.style={rectangle, font=\fontsize{20}{20}\selectfont},
squarenode/.style={rectangle, draw=black, fill=white,  very thick, minimum size=12mm, font=\fontsize{20}{20}\selectfont},
squarenode1/.style={rectangle, draw=black!40, dashed, fill=white,  very thick, minimum width=198mm, minimum height=45mm},
squarenode2/.style={rectangle, draw=black, fill=white,  very thick, dotted, minimum width=100mm, minimum height=20mm},
posnode/.style={rectangle},
]
%Nodes
\node at (3.9,0.8) [squarenode1] (nNetwork) {};
\node at (-2,4.5) [squarenode] (nLocal1) {agent};
\node at (0.5,7) [squarenode] (nLocal3) {agent};
\node at (10,4.5) [squarenode] (nLocal2) {agent};
\node at (7.5,7) [squarenode] (nLocal4) {agent};

\draw[-,{line width=3pt},black!40, transform canvas={xshift=-20mm}] (-1.7,-0.5)--(-1.45,0.5);
\draw[-,{line width=3pt},black!40, transform canvas={xshift=-20mm}] (1.5,-0.5)--(1.75,0.5);
\draw[-,{line width=3pt},black!40, transform canvas={xshift=-20mm}] (-2.8,-0.25)--(2.6,-0.25);
\draw[-,{line width=3pt},black!40, transform canvas={xshift=-20mm}] (-2.6,0.25)--(2.8,0.25);
\draw[-,{line width=3pt},black!40, transform canvas={xshift=100mm}] (-1.7,-0.5)--(-1.45,0.5);
\draw[-,{line width=3pt},black!40, transform canvas={xshift=100mm}] (1.5,-0.5)--(1.75,0.5);
\draw[-,{line width=3pt},black!40, transform canvas={xshift=100mm}] (-2.8,-0.25)--(2.6,-0.25);
\draw[-,{line width=3pt},black!40, transform canvas={xshift=100mm}] (-2.6,0.25)--(2.8,0.25);
\draw[-,{line width=3pt},black!40, transform canvas={xshift=5mm,yshift=20mm}] (-1.7,-0.5)--(-1.45,0.5);
\draw[-,{line width=3pt},black!40, transform canvas={xshift=5mm,yshift=20mm}] (1.5,-0.5)--(1.75,0.5);
\draw[-,{line width=3pt},black!40, transform canvas={xshift=5mm,yshift=20mm}] (-2.8,-0.25)--(2.6,-0.25);
\draw[-,{line width=3pt},black!40, transform canvas={xshift=5mm,yshift=20mm}] (-2.6,0.25)--(2.8,0.25);
\draw[-,{line width=3pt},black!40, transform canvas={xshift=75mm,yshift=20mm}] (-1.7,-0.5)--(-1.45,0.5);
\draw[-,{line width=3pt},black!40, transform canvas={xshift=75mm,yshift=20mm}] (1.5,-0.5)--(1.75,0.5);
\draw[-,{line width=3pt},black!40, transform canvas={xshift=75mm,yshift=20mm}] (-2.6,-0.25)--(2.8,-0.25);
\draw[-,{line width=3pt},black!40, transform canvas={xshift=75mm,yshift=20mm}] (-2.8,0.25)--(2.6,0.25);

\node[trapezium, draw,trapezium left angle=120, trapezium right angle=60, fill=white, opacity=0.5, minimum size=12mm, text=blue,text opacity=0.5,font=\fontsize{18}{18}\selectfont] at (-2,0) (nSub1) {subnetwork};
\node[trapezium, draw,trapezium left angle=120, trapezium right angle=60, fill=white, opacity=0.5, minimum size=12mm, text=blue,text opacity=0.5,font=\fontsize{18}{18}\selectfont] at (10,0) (nSub2) {subnetwork};
\node[trapezium, draw,trapezium left angle=120, trapezium right angle=60, fill=white, opacity=0.5, minimum size=12mm, text=blue,text opacity=0.5,font=\fontsize{18}{18}\selectfont] at (0.5,2) (nSub3) {subnetwork};
\node[trapezium, draw,trapezium left angle=120, trapezium right angle=60, fill=white, opacity=0.5, minimum size=12mm, text=blue,text opacity=0.5,font=\fontsize{18}{18}\selectfont] at (7.5,2) (nSub4) {subnetwork};
\node at (3.5,-1) [textnode] (nNetwork) {\textbf{The Whole Traffic Network}};

\node at (3.1,1.2) [posnode] (ps1) {};
\node at (7.5,0.8) [posnode] (ps2) {};

%Arrows
\draw[->,{line width=3pt},black!40, transform canvas={yshift=-2mm,xshift=1mm}] (nSub1.east)->(nSub2.west);
\draw[<-,{line width=3pt},black!40, transform canvas={yshift=2mm,xshift=-1mm}] (nSub1.east)->(nSub2.west);

\draw[->,{line width=3pt},black!40, transform canvas={yshift=-2mm,xshift=1mm}] (nSub3.east)->(nSub4.west);
\draw[<-,{line width=3pt},black!40, transform canvas={yshift=2mm,xshift=-1mm}] (nSub3.east)->(nSub4.west);

\draw[<-,{line width=3pt},black!40, transform canvas={xshift=6mm}] (ps1.north)->(ps2.south);
\draw[->,{line width=3pt},black!40, transform canvas={xshift=-6mm}] (ps1.north)->(ps2.south);

\draw[->,{line width=3pt},black!40, transform canvas={xshift=4mm}] (nSub1.north)->(nSub3.south);
\draw[<-,{line width=3pt},black!40, transform canvas={xshift=-4mm}] (nSub1.north)->(nSub3.south);

\draw[->,{line width=3pt},black!40, transform canvas={xshift=4mm}] (nSub2.north)->(nSub4.south);
\draw[<-,{line width=3pt},black!40, transform canvas={xshift=-4mm}] (nSub2.north)->(nSub4.south);

\draw[->,{line width=3pt},blue, transform canvas={xshift=-2mm}] (nLocal1.south)--(nSub1.center);
\draw[<-,{line width=3pt},green, transform canvas={xshift=2mm}] (nLocal1.south)--(nSub1.center);

\draw[->,{line width=3pt},blue, transform canvas={xshift=-2mm}] (nLocal2.south)--(nSub2.center);
\draw[<-,{line width=3pt},green, transform canvas={xshift=2mm}] (nLocal2.south)--(nSub2.center);

\draw[->,{line width=3pt},blue, transform canvas={xshift=-2mm}] (nLocal3.south)--(nSub3.center);
\draw[<-,{line width=3pt},green, transform canvas={xshift=2mm}] (nLocal3.south)--(nSub3.center);

\draw[->,{line width=3pt},blue, transform canvas={xshift=-2mm}] (nLocal4.south)--(nSub4.center);
\draw[<-,{line width=3pt},green, transform canvas={xshift=2mm}] (nLocal4.south)--(nSub4.center);

\node at (-2,4.5) [squarenode] (nAgent1) {agent};
\node at (0.5,7) [squarenode] (nAgent3) {agent};
\node at (10,4.5) [squarenode] (nAgent2) {agent};
\node at (7.5,7) [squarenode] (nAgent4) {agent};
\draw[-,{line width=3pt},red, dotted] (nAgent1.east)--(nAgent2.west);
\draw[-,{line width=3pt},red, dotted] (nAgent3.east)--(nAgent4.west);
\draw[-,{line width=3pt},red, dotted] (nAgent1.east)--(nAgent3.east);
\draw[-,{line width=3pt},red, dotted] (nAgent2.west)--(nAgent4.west);
\draw[-,{line width=3pt},red, dotted] (nAgent2.west)--(nAgent3.east);

\end{tikzpicture}}
\end{center}
\caption{The schematic diagram of distributed stochastic MPC traffic signal control method.}
\label{fig_coordination_graph}
\end{figure}
%%%%%%%%%%%%%%%%%%%%%%%%%%%%%%%%%%%%%%%%%%%%%%%%%%%%%%%%%%%%%%%%%%%%%%%%%%%%%%%%%%%%%%%%%%
Every agent $\mathcal{S}_i$ is required to repeat the following steps in each control time step $t$.
\begin{enumerate}
\item Using sensors and appropriate estimation method, agent $\mathcal{S}_i$ measures and estimates the current traffic states $n_z(t)$ for all road links $z: \sigma(z) \in \mathcal{J}_{\mathcal{S}_i}$, the uncertain traffic model parameters in the set $\mathcal{RP}_{\mathcal{S}_i}(t)$.
\item Formulate and solve the distributed stochastic MPC traffic signal control problem \eqref{eq_SMPC_ADMMconsensus} by applying Algorithm \ref{alg_proposed_control}.
\item Apply the optimal control decisions $g_z^{opt}(t) = \frac{1}{S_z}\hat{q}_{z,0}^{opt}, \forall z \in \mathcal{L}_{J_v}^{in}$ and $g_p^{opt}(t) = \hat{g}_{p,0}^{opt}, \forall p \in \mathcal{P}_{J_v}$ for every internal junction $J_v \in \mathcal{J}_{\mathcal{S}_i}$.
Then wait until the next cycle. 
\end{enumerate}
Fig. \ref{fig_coordination_graph} illustrates the scheme of our proposed traffic signal control method for an urban network.
The red lines indicate communication links among agents, the blue arrows indicate signal control transferring from local controllers to traffic signal lights, the green arrows indicate traffic states and model parameters transferring from sensors to local controllers, and the gray arrows indicate the traffic flows among subnetworks.
%%%%%%%%%%%%%%%%%%%%%%%%%%%%%%%%%%%%%%%%%%%%%%%%%%%%%%%%%%%%%%%%%%%%%%%%%%%%%%%%%%%%%%%%%%

%%%%%%%%%%%%%%%%%%%%%%%%%%%%%%%%%%%%%%%%%%%%%%%%%%%%%%%%%%%%%%%%%%%%%%%%%%%%%%%%%%%%%%%%%%
\section{Simulation}
%%%%%%%%%%%%%%%%%%%%%%%%%%%%%%%%%%%%%%%%%%%%%%%%%%%%%%%%%%%%%%%%%%%%%%%%%%%%%%%%%%%%%%%%%%
In this section, we make some simulations using VISSIM\footnote{available at www.vissim.de.} and MATLAB to test the effectiveness of our proposed method in controlling an urban network and the computational capacities of the update iteration \eqref{eq_updatelaw} in solving the stochastic program \eqref{eq_distributedSMPC}.
VISSIM is one of the most widely-used microscopic traffic simulation software since it can describe the operation of traffic signals and the movement of vehicles similar as in a realistic environment.
MATLAB is used to execute some traffic signal control methods in order to find the planned traffic signal splits.
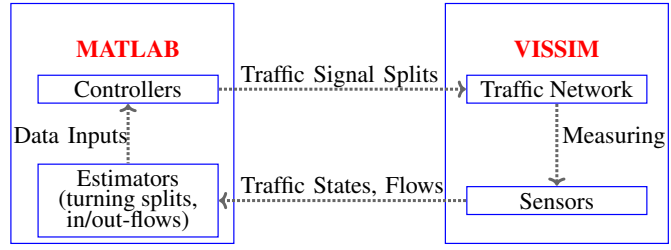
\begin{figure}[htb]
\centering
\scalebox{0.38}{\begin{tikzpicture}[
textnode1/.style={rectangle},
textnode2/.style={rectangle,red, font=\bf},
textlabel/.style={rectangle, scale=2.0},
squarenode1/.style={rectangle, draw=blue, fill=white,  very thick, text width = 5em, text centered, minimum width = 6mm, minimum height = 6mm, node distance = 5em},
squarenode2/.style={rectangle, scale=1.4, draw=blue, very thick, text width = 10em, text centered, minimum height = 6mm, node distance = 5em},
squarenode3/.style={rectangle, scale=1.2, draw=blue, very thick, text width = 10em, text centered, minimum width = 65mm, minimum height = 70mm, node distance = 5em},
]
%Nodes
\node at (1.5,6.2) [squarenode2] (nMPC) {\Large Controllers};
\node at (1.5,2.3) [squarenode2] (nEstimator) {\Large Estimators 

 (turning splits, in/out-flows)};
\node at (1.3,5.0) [squarenode3] (nRec1) {};
\node at (16.5,6.2) [squarenode2] (nModel) {\Large Traffic Network};
\node at (16.5,2.3) [squarenode2] (nSensor) {\Large Sensors};
\node at (16.5,5.0) [squarenode3] (nRec2) {};
%Arrows
\draw[<-,{line width=3pt},black!60, dotted, transform canvas={yshift=0mm}] (nMPC.south)->(nEstimator.north);
\draw[->,{line width=3pt},black!60, dotted, transform canvas={yshift=0mm}] (nModel.south)->(nSensor.north);
\draw[->,{line width=3pt},black!60, dotted, transform canvas={yshift=0mm}] (nMPC.east)->(nModel.west);
\draw[<-,{line width=3pt},black!60, dotted, transform canvas={yshift=0mm}] (nEstimator.east)->(nSensor.west);
\node at (8.9,6.65) [textlabel] (t1) {Traffic Signal Splits};
\node at (9,2.8) [textlabel] (t2) {Traffic States, Flows};
\node at (-0.5,4.5) [textlabel] (t3) {Data Inputs};
\node at (18.5,4.5) [textlabel] (t4) {Measuring};
\node at (1.5,7.6) [textnode2] (t5) {\Huge MATLAB};
\node at (16.5,7.6) [textnode2] (t6) {\Huge VISSIM};
\end{tikzpicture}}
\caption{Information exchange between VISSIM and MATALB.}\label{fig_diagramVM}
\end{figure}
Fig. \ref{fig_diagramVM} illustrates the simulation environment.
The measurements are transferred from VISSIM to MATLAB since they are required in some control methods.
The computed traffic signal splits are transferred from MATLAB to VISSIM.
%%%%%%%%%%%%%%%%%%%%%%%%%%%%%%%%%%%%%%%%%%%%%%%%%%%%%%%%%%%%%%%%%%%%%%%%%%%%%%%%%%%%%%%%%%
\subsection{Simulation setup}
%%%%%%%%%%%%%%%%%%%%%%%%%%%%%%%%%%%%%%%%%%%%%%%%%%%%%%%%%%%%%%%%%%%%%%%%%%%%%%%%%%%%%%%%%%
The urban network tested in this paper is presented in Fig. \ref{fig_testednetwork}.
\begin{figure}[htb]
\begin{minipage}{.4\linewidth}
\centering
\includegraphics[width=0.68\textwidth]{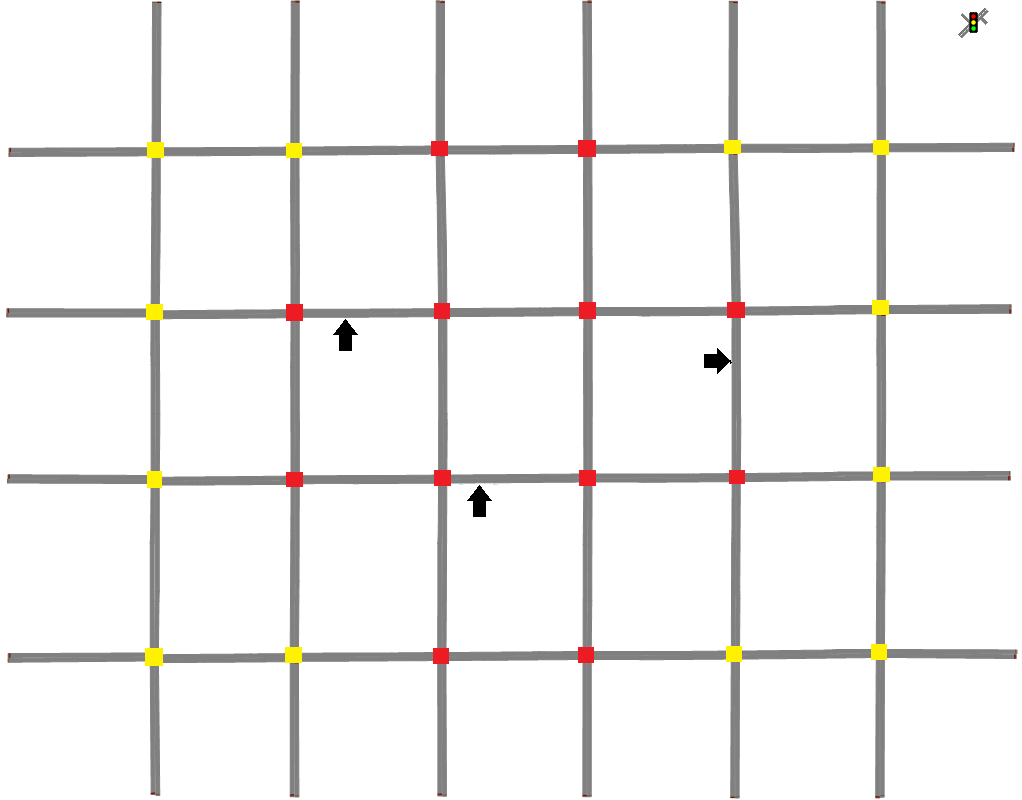}
\caption{The tested traffic network consists of $24$ junctions and $116$ roads (Image is captured from VISSIM).}\label{fig_testednetwork}
\end{minipage} 
\begin{minipage}{.55\linewidth}
\centering
\captionof{table}{Nominal traffic demand in Scenario 1.}\label{tb_demand}
\scalebox{0.6}{
\begin{tabular}{ccccccc}
 \hline
 $\Delta t$ (min.)  & 1 - 20 & 21 - 40 & 41 - 60 & 61 - 80 & 81 - 100 & 101 - 120 \\
 \hline
Traffic demand & \multirow{2}{*}{$1000$} & \multirow{2}{*}{$1250$} & \multirow{2}{*}{$1150$} & \multirow{2}{*}{$1000$} & \multirow{2}{*}{$900$} & \multirow{2}{*}{$800$} \\
$d_B^*(\Delta t)$ (veh./hour) &&&&&&\\
 \hline
\end{tabular}}
\captionof{table}{Nominal turning ratios.}\label{tb_ratios}
\scalebox{0.8}{
\begin{tabular}{c|c|cccc}
 \hline
Road & $(u + v)$ mod $4$  & 0 & 1 & 2 & 3 \\
\hline
\multirow{3}{*}{$u \rightarrow v$} & turning left & 0.2 & 0.15 & 0.15 & 0.2 \\
& turning right & 0.2 & 0.15 & 0.2 & 0.15 \\
& going straight & 0.6 & 0.7 & 0.65 & 0.65 \\
 \hline
\end{tabular}}
\end{minipage} 
\end{figure}
It consists of $24$ internal junctions, which are marked by red/yellow squares, and $116$ roads. The length of roads are in the range of $[280, 400]$ meters. For the junctions corresponding to red (resp. yellow) squares, the sequences of their traffic signal phases are set as Type 3 (resp. Type 2) in TABLE. \ref{tbl_sequencePhases} and their incoming roads have $5$ (resp. 3) lanes. We assume that only one lane is reserved for turning left in every incoming roads of junctions with Type 3. 
The saturation flows of roads are set by $0.55$ (veh/s/lane).
We set the control time length as $T = 60$ seconds and the lost time as $L = 4$ seconds for all junctions.
For all destination roads, we assume that they have three lanes and their downstream traffic flows are restricted by an upper bound of $\overline{q_{out}} = 20$ (veh.) for all control time steps.
%The maximum admissible number of vehicles in one road link $z$ is chosen by a random number $\overline{n_z} \in [165, 225]$ (then the length of the road link $z$, in the simulation model, is chosen suitable with the number $\overline{n_z}$).

Let the indexes of internal junctions be in the increasing-order from left to right and from top to bottom. In addition, we set all source nodes to have the same index $0$.
The turning ratios and the difference of exogenous in/out-flows corresponding to the road from the junction $u$ to the junction $v$, where $u, v \in \{0,1, \dots, 24\}$, in $\Delta t$-th minute are described by the following functions \eqref{eq_UncertainTurningRatios} and \eqref{eq_UncertainExogenousFlows}.
\begin{subequations}\label{eq_UncertainTurningRatios}
\begin{gather}
r_{uv}^{left}(\Delta t) = r_{uv,B}^{left} + \delta_{uv,r}^{left}(\Delta t)\\
r_{uv}^{right}(\Delta t) = r_{uv,B}^{right} + \delta_{uv,r}^{right}(\Delta t)\\
r_{uv}^{straight}(\Delta t) = 1 - r_{uv}^{left}(\Delta t) - r_{uv}^{right}(\Delta t)
\end{gather}
\end{subequations}
\begin{equation}\label{eq_UncertainExogenousFlows}
e_{uv}(t) =  (d_{uv,B}(t) + \delta_{uv}^d(t)) - (s_{uv,B}(t) + \delta_{uv}^s(t))
\end{equation}
where $r_{uv,B}^{left}$ and $r_{uv,B}^{right}$ are the nominal values of the ratios for turning left and turning right movements; $d_{uv,B}(t)$ is the nominal value of the traffic demand; $s_{uv,B}(t)$ is the nominal value of the exit flow; $\delta_{uv,r}^{left}(t)$ and $\delta_{uv,r}^{right}(t)$ are random numbers in the range $[-0.1,0.1]$; $\delta_{uv,d}(t)$ is random number in the range $[-350, 350]$ if $u = 0$; $\delta_{uv,d}(t)$ for $u \neq 0$ and $\delta_{uv,s}(t)$ are random numbers in $[-50, 50]$.
Note that the real values of traffic model parameters are given in \eqref{eq_UncertainTurningRatios} and \eqref{eq_UncertainExogenousFlows}, while traffic signal controllers estimate only their nominal values (also expected values) and variances.
The total simulation time is $2$ hours corresponding to $N_C = 120$ control time steps.

In this paper, three scenarios of traffic demands (i.e., the exogenous inflows) are considered and they are referred as low-level, medium-level and high-level ones, respectively.
In Scenario 1, we assume that the nominal exogenous inflows of all the source roads that have $3$ lanes are equal to $d_B^*(\Delta t)$ given in TABLE \ref{tb_demand}. For the source roads having $5$ lanes, we set their nominal exogenous inflows as $1.25 d_B^*(\Delta t)$.
The traffic demands in Scenario 2 and Scenario 3 are $25$ and $40$ percent larger than the one in Scenario 1. 
The black arrows in Fig. \ref{fig_testednetwork} represent the places from which the vehicles entering into the network with the flow rate $d_{in}$. We set $d_{in} = 0$ in Scenario 1 and $d_{in}$ as a random number in the ranges $[500,600]$ and $[900,1100]$ in Scenario 2 and 3, respectively.
The nominal turning ratios of the road from junction $u$ to junction $v$ are determined by one row in TABLE \ref{tb_ratios} where the index of the corresponding row is the remainder of $u + v$ dividing $4$.
%%%%%%%%%%%%%%%%%%%%%%%%%%%%%%%%%%%%%%%%%%%%%%%%%%%%%%%%%%%%%%%%%%%%%%%%%%%%%%%%%%%%%%%%%%
\subsection{Results of microscopic simulation}
%%%%%%%%%%%%%%%%%%%%%%%%%%%%%%%%%%%%%%%%%%%%%%%%%%%%%%%%%%%%%%%%%%%%%%%%%%%%%%%%%%%%%%%%%%
In this part, we compare the simulation results of four control methods.
\begin{enumerate}
\item Pretimed control method: the splits of traffic signal phases are set proportionally to the average traffic flows of their corresponding road links. Four cycle times ($40, 60, 80, 100$ seconds) are tested and the one with the best performance is used to compare with other control methods in each scenario.
\item The back-pressure (BP)-based control: the method is proposed in \cite{JeanGregoire2015} when the control time interval is chosen as $10$ seconds.
\item Nominal MPC traffic signal control method: the applied traffic signal splits are corresponding to the optimal solution of the nominal MPC traffic signal control problem \eqref{eq_problemTP}.
They can be obtained by our proposed method when assuming all variances of traffic model parameters are zero.
In this case, the set $\Omega_{\mathcal{S}_i}$ of local constraints in \eqref{eq_localinequalities_subproblem} consists of only linear inequalities for every subproblem $\mathcal{S}_i$, $i = 1, \dots, N$.
\item Stochastic MPC traffic signal control method: the traffic signal splits are obtained by applying our proposed method to solve the stochastic MPC traffic signal control problem \eqref{eq_SMPCProblem} with $\epsilon_t = 0.2$.
This is our proposed control method.
\end{enumerate}
The following cost function is used in MPC controllers.
\begin{equation}
\Phi(t) = \sum\limits_{k = 1}^{3}\sum\limits_{z \in \mathcal{L}} \scalebox{0.85}{$\left( \frac{\left(n_z(t+k)\right)^2}{\overline{n_z}} + \beta_{z,k} n_z(t+k) - \gamma_{z,k}q_z(t+k) \right)$}
\end{equation}
where $\beta_{z,k} = \gamma_{z,k} = 0.3$ for all $k = 0, \dots, K-1, z \in \mathcal{L}$. We choose $K = 3$ for the horizontal time.
%%%%%%%%%%%%%%%%%%%%%%%%%%%%%%%%%%%%%%%%%%%%%%%%%%%%%%%%%%%%%%%%%%%%%%%%%%%%%%%%%%%%%%%%%%
\begin{figure*}[htb]
\begin{center}
\captionof{table}{Simulation results of four different control strategies.}
\label{tb_vb}
\scalebox{0.68}{
\begin{tabular}{c|c|c|c|c|c|c|c|c|c|c|c|c}
\hline
\multirow{2}{*}{} & \multicolumn{4}{c|}{Scenario 1} & \multicolumn{4}{c|}{Scenario 2} & \multicolumn{4}{c}{Scenario 3}\\
\cline{2-13}
 & \multirow{2}{*}{Pretimed} & \multirow{2}{*}{BP-based} & Nominal & Stochastic & \multirow{2}{*}{Pretimed} & \multirow{2}{*}{BP-based} & Nominal & Stochastic & \multirow{2}{*}{Pretimed} & \multirow{2}{*}{BP-based} & Nominal & Stochastic\\
&  &  & MPC & MPC &  &  & MPC & MPC &  &  & MPC & MPC\\
\hline
Criterion 1 & $56055$ & $56177$ & $56249$ & $56256$ & $64216$ & $62442$ & $66530$ & $66577$ & $66028$ & $57100$ & $70136$ & $70553$\\
(veh.) & $\cdot$ & $+0.12\%$ & $+0.35\%$ & $+0.36\%$ & $\cdot$ & $-2.76\%$ & $+3.6\%$ & $+3.67\%$ & $\cdot$ & $-13.5\%$ & $+6.22\%$ & $+6.85\%$\\
 \hline
Criterion 2 & $53975$ & $54897$ & $54832$ & $54763$ & $58535$ & $55081$ & $61463$ & $61595$ & $58633$ & $51409$ & $60408$ & $61659$\\
(veh.) & $\cdot$ & $+1.70\%$ & $+1.59\%$ & $+1.46\%$ & $\cdot$ & $-5.9\%$ & $+5.00\%$ & $+5.22\%$ & $\cdot$ & $-12.32\%$ & $+3.03\%$ & $+5.16\%$\\
 \hline
Criterion 3 & $2.528$ & $2.549$ & $2.532$ & $2.535$ & $2.801$ & $2.448$ & $2.920$ & $2.932$ & $2.672$ & $2.087$ & $2.936$ & $2.995$\\
$\times 10^5$ (veh.) & $\cdot$ & $+0.83\%$ & $+0.16\%$ & $+0.28\%$ & $\cdot$ & $-12.60\%$ & $+4.25\%$ & $+4.68\%$ & $\cdot$ & $-21.89\%$ & $+9.88\%$ & $+12.08\%$\\
 \hline
Criterion 4 & $352.4$ & $318.3$ & $337.1$ & $323.6$ & $804.4$ & $991.5$ & $629.7$ & $617.9$ & $1085.6$ & $1245.6$ & $922.5$ & $888.5$\\
(seconds) & $\cdot$ & $-9.68\%$ & $-4.34\%$ & $-8.17\%$ & $\cdot$ & $+23.26\%$ & $-21.72\%$ & $-23.19\%$ & $\cdot$ & $+14.74\%$ & $-15.02\%$ & $-18.16\%$\\
 \hline
\end{tabular}}
\end{center}
\end{figure*}
\begin{figure}[htb]
\centering
\includegraphics[width=0.24\textwidth]{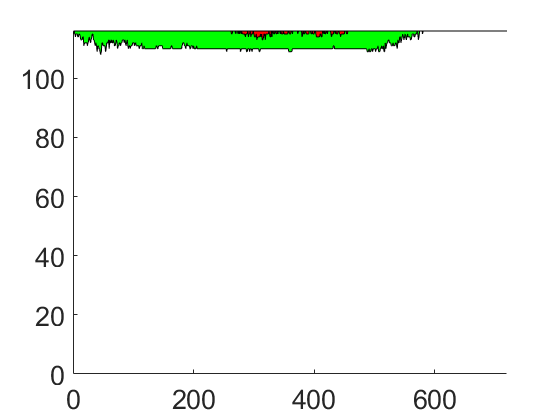}
\includegraphics[width=0.24\textwidth]{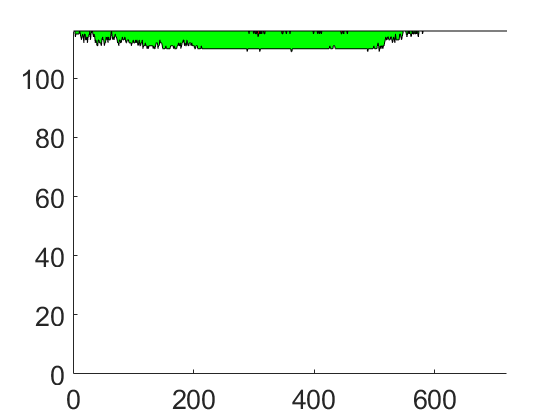}
\includegraphics[width=0.24\textwidth]{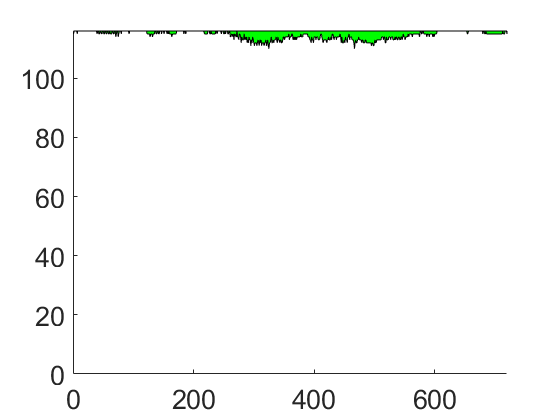}
\includegraphics[width=0.24\textwidth]{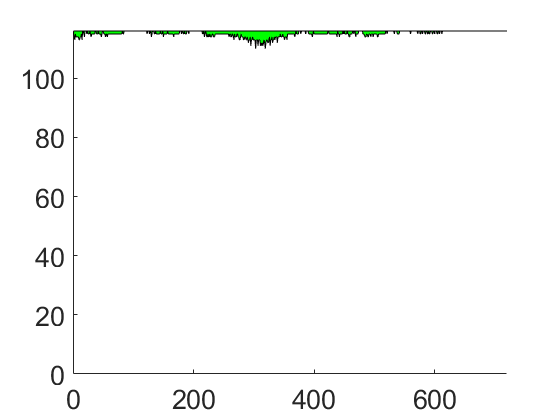}

\includegraphics[width=0.24\textwidth]{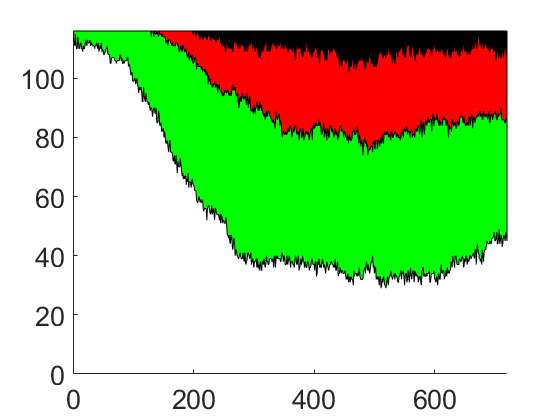}
\includegraphics[width=0.24\textwidth]{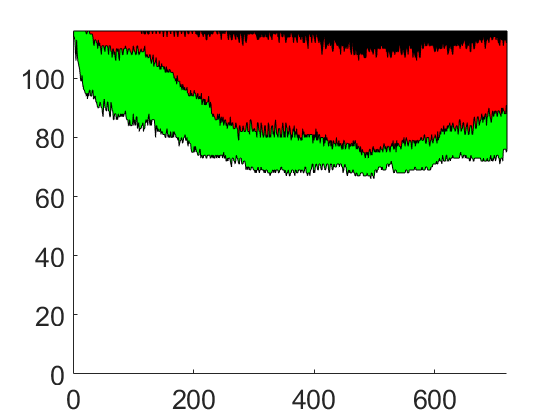}
\includegraphics[width=0.24\textwidth]{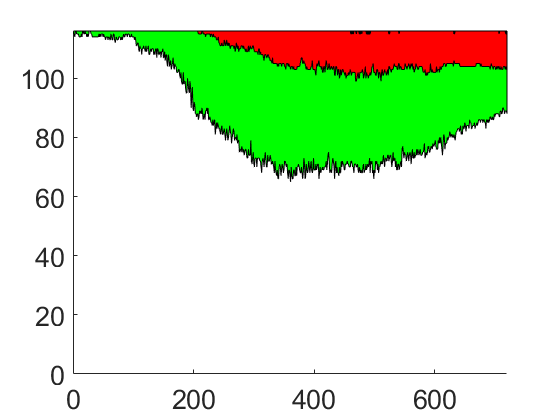}
\includegraphics[width=0.24\textwidth]{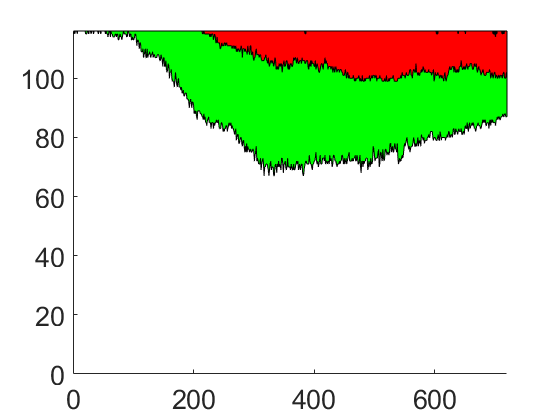}

\includegraphics[width=0.24\textwidth]{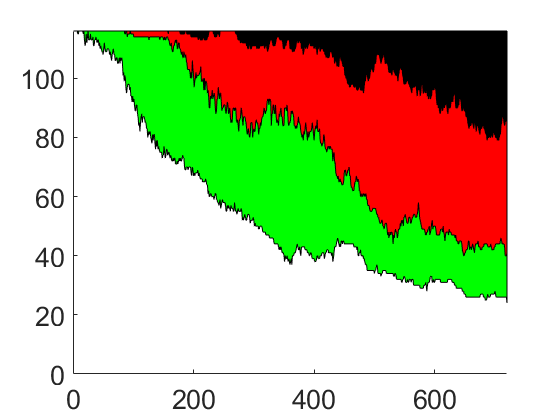}
\includegraphics[width=0.24\textwidth]{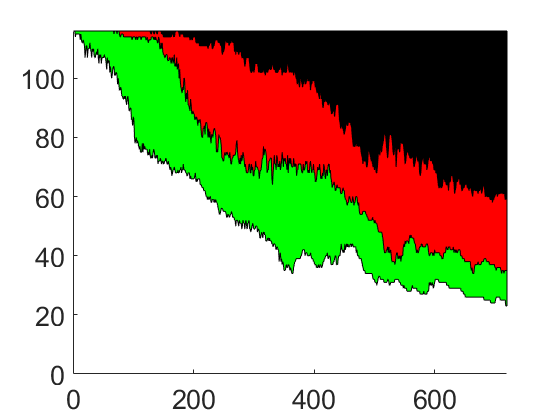}
\includegraphics[width=0.24\textwidth]{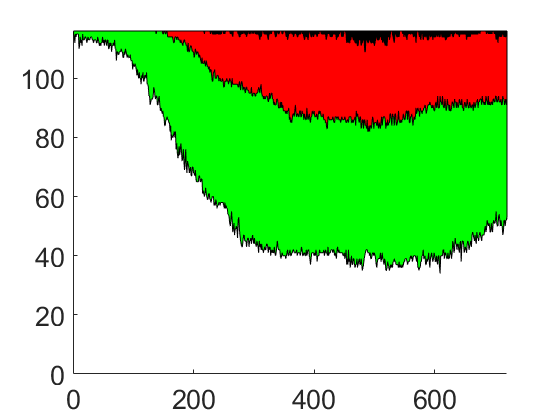}
\includegraphics[width=0.24\textwidth]{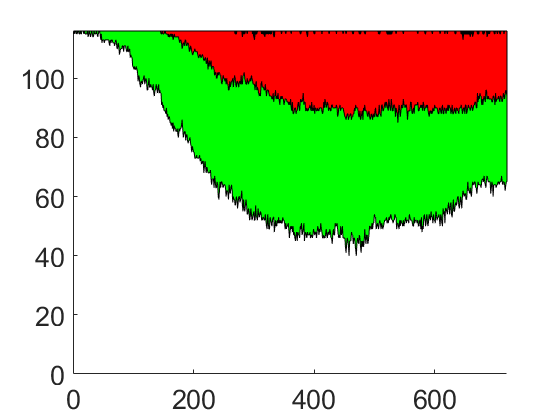}
\caption{The numbers of roads in four different congestion modes over simulation time. Each step in time horizon axis corresponds to $10$ seconds.}\label{fig_compareDensities}
\end{figure}
%%%%%%%%%%%%%%%%%%%%%%%%%%%%%%%%%%%%%%%%%%%%%%%%%%%%%%%%%%%%%%%%%%%%%%%%%%%%%%%%%%%%%%%%%%

In order to test the effectiveness of the above control traffic signal control methods in avoiding traffic congestion and reducing travel time of vehicles in the urban network, we compare the simulation results according to five evaluation criteria as follows.
\begin{enumerate}
\item Criterion 1: the total number of vehicles entering the urban network over simulation time.
\item Criterion 2: the total number of vehicles exiting the urban network over simulation time.
\item Criterion 3: the total number of active vehicles crossing internal junctions over simulation time.
\item Criterion 4: the mean of wait time of vehicles in the urban network.
\end{enumerate}
The above four criteria of tested control methods are reported in TABLE. \ref{tb_vb}.
In each cell of this table, we provide both the result of the performance index and the improvement of the control methods compared with the pretimed control method.
\begin{enumerate}\setcounter{enumi}{3}
\item The level of urban congestion in the traffic network is measured by the relative occupation of vehicles in all road links.
This perfomance index for one road link $z$ is computed by $n_z^{rel}(\Delta \tilde{t}) = \frac{n_z(\Delta \tilde{t})}{\overline{n_z}}$ where $\Delta \tilde{t}$ in second.
We consider one road in a low-crowded mode if $n_z^{rel}(\Delta \tilde{t}) \le 0.45$; a medium-crowded mode if $0.45 < n_z^{rel}(\Delta \tilde{t}) \le 0.75$; a high-crowded mode if $0.75 < n_z^{rel}(\Delta t) \le 0.95$; a congested mode if $n_z^{rel}(\Delta t) > 0.95$.
\end{enumerate}
Fig. \ref{fig_compareDensities} represents the numbers of roads with four different modes over simulation time corresponding to four different control methods and three different scenarios.
In this figure, each step in time axis corresponds to $10$ seconds.
The first row corresponds to Scenario 1, the second row corresponds to Scenario 2 and the last row corresponds to Scenario 3.
First, second, third and fourth column presents respectively the results of the pretimed control method, the BP method, the nominal MPC traffic signal control method and the stochastic MPC traffic signal control method.
At each time step, the numbers of roads in low-crowded mode, medium-crowded mode, high-crowded mode and congested mode are described respectively by the region in white, green, red and black colors.

As shown in TABLE \ref{tb_vb} and Fig. \ref{fig_compareDensities}, it is observable that MPC control approaches are successful in reducing the number of roads with high congestion risk and wait time of vehicles in the urban network while enhancing the total throughput over this network.
In Scenario 1, the total throughput corresponding to MPC-based methods is comparable to BP-based control method, which is well-known to provide maximum throughput in the case of low traffic demands. When the traffic demands increase in Scenario 2 and Scenario 3, the performance of the pretimed control method and BP-based method decrease significantly. Specially, the traffic congestion occurs in Scenario 3 under the control of these two methods.
Meanwhile, MPC-based methods not only exploit the capacity of roads better (see Criterion 1, 2) but also enhance the smooth movement of vehicles (see Criterion 3, 4).
Although the nominal MPC traffic signal control method provides a good performance, the stochastic one shows an improvement in control performance as the traffic demands of the urban network increase.
%%%%%%%%%%%%%%%%%%%%%%%%%%%%%%%%%%%%%%%%%%%%%%%%%%%%%%%%%%%%%%%%%%%%%%%%%%%%%%%%%%%%%%%%%%
\subsection{Computational load of the proposed control method}
%%%%%%%%%%%%%%%%%%%%%%%%%%%%%%%%%%%%%%%%%%%%%%%%%%%%%%%%%%%%%%%%%%%%%%%%%%%%%%%%%%%%%%%%%%
To test the computational effectiveness of the update rule \eqref{eq_updatelaw} when solving \eqref{eq_distributedSMPC}, we use MATLAB to implement the method and measure the execution time.
Our experiments are conducted in a computer with chip Intel Core I5 8500 and $16$ GB RAM.
Fig.\ref{fig_converged} illustrates the evolution of the distance from the estimated solution to the true optimal solution, i.e., $||\textbf{x}(l) - \textbf{x}^{opt}||$, and the termination condition in \eqref{eq_terminated} when $\rho = 0.01$ and $K = 3$.
We can see that after about $1400$ iteration steps, the estimated solution goes very close to the optimal one.
%Matching with the evolution of the values for checking
In addition, from Theorem \ref{th_ADMM}, it is reasonable to choose $tol = 10^{-6}$. % in the implementation of the control method proposed in the subsection IV-C.
\begin{figure}[htb]
\begin{center}
\includegraphics[width=0.41\textwidth]{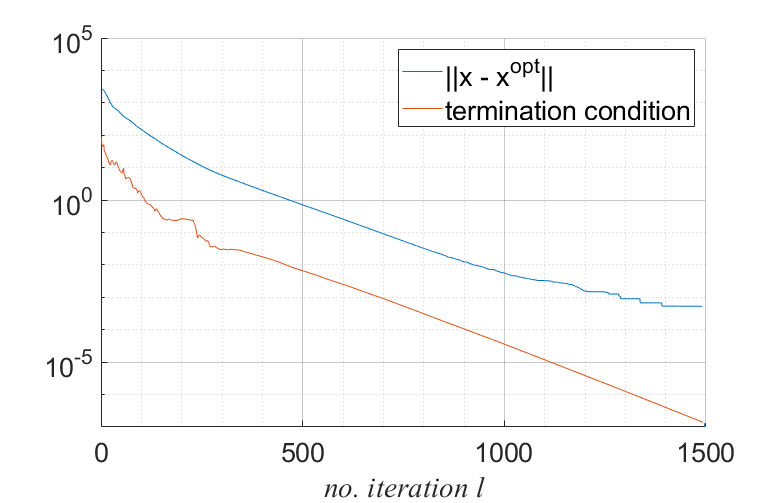}
\caption{The convergence of the update rule \eqref{eq_updatelaw} when $\rho = 0.01$.} 
\label{fig_converged}
\end{center}
\end{figure}
In the previous subsection, we use Algorithm \ref{alg_proposed_control} to solve the distributed stochastic MPC traffic signal control problem \eqref{eq_distributedSMPC} when assuming that the traffic network is divided into multiple subnetworks and each of them corresponds to one junction. The computation time required by all agents corresponding to one control time step is in the range of $14.84 \sim 25.34$ seconds.
We expect significant reduction in the case of there are more computers working in parallel.
In order to test the effectiveness of the proposed control method in distributed setup, we use functions \textit{tic} and \textit{toc} in MATLAB to measure the time taken by one local controller to update its local variables in every iteration, i.e., one loop of Algorithm 1.
For every parallel process which can start at the same time, we add the maximum value of the time length taken by local controllers to the total execution time for running the proposed control method.
\begin{table}[htb]
\centering
\caption{Computational load of the proposed control method for solving stochastic MPC traffic signal control problem (of the urban network in Fig. \ref{fig_testednetwork}) with different $K$.}
\label{tb_convergence1}
\scalebox{0.75}{
\begin{tabular}{c|c|c|c|c|c|c}
\hline
\multirow{2}{*}{K} & \multicolumn{2}{c|}{\# of iteration steps} & \multicolumn{2}{c|}{Centralized time (s)} & \multicolumn{2}{c}{Distributed time (s)}\\
\cline{2-7}
 & average & maximum & average & maximum & average & maximum\\
\hline
1 & $560$ & $896$ & $3.10$ & $4.88$ & $0.15$ & $0.28$\\
\hline
2 & $820$ & $1280$ & $9.89$ & $15.01$ & $0.45$ & $0.67$\\
\hline
3 & $996$ & $1370$ & $18.02$ & $25.34$ & $0.93$ & $1.25$\\
\hline
4 & $1283$ & $1791$ & $39.62$ & $52.51$ & $1.92$ & $2.55$\\
 \hline
5 & $1620$ & $2187$ & $61.42$ & $84.68$ & $3.89$ & $4.25$\\
 \hline
6 & $2005$ & $2975$ & $102.36$ & $152.42$ & $6.21$ & $7.62$\\
 \hline
\end{tabular}}
\end{table}
TABLE~\ref{tb_convergence1} shows the average and maximum computation time of the update \eqref{eq_updatelaw} until the termination condition \eqref{eq_terminated} is satisfied for some different horizontal time $K$.
In this table, we also provide the average and maximum numbers of iteration steps in each case. Note that, these numbers are the same in distributed and centralized setups.
It is easy to see that the execution time is smaller than the time interval between two consecutive control steps in both centralized and distributed setups.
Specially, the execution time required to obtain the optimal stochastic traffic signal splits is only one or two seconds in the distributed setup.
So, our proposed control method can be considered promising when a good communication among local controllers is available.
Furthermore, we apply our proposed control method to solve the nominal MPC traffic signal control problem.
As shown in TABLE \ref{tb_convergence2}, the computational load is significantly reduced when compared to the stochastic case.
It is reasonable to state that our proposed control method is also applicable to the nominal MPC traffic signal control approach when the estimation of traffic model parameters is highly confident.
\begin{table}[htb]
\centering
\caption{Computational load of the proposed control method for solving nominal MPC traffic signal control problem (of the urban network in Fig. \ref{fig_testednetwork}) with different $K$.}
\label{tb_convergence2}
\scalebox{0.75}{
\begin{tabular}{c|c|c|c|c|c|c}
\hline
\multirow{2}{*}{K} & \multicolumn{2}{c|}{\# of iteration steps} & \multicolumn{2}{c|}{Centralized time (s)} & \multicolumn{2}{c}{Distributed time (s)}\\
\cline{2-7}
 & average & maximum & average & maximum & average & maximum\\
\hline
1 & $840$ & $948$ & $1.67$ & $1.88$ & $0.1$ & $0.15$\\
\hline
2 & $1505$ & $1831$ & $6.85$ & $8.11$ & $0.34$ & $0.41$\\
\hline
3 & $1895$ & $2378$ & $11.532$ & $14.17$ & $0.52$ & $0.63$\\
\hline
4 & $2128$ & $2593$ & $22.34$ & $26.04$ & $1.11$ & $1.35$\\
 \hline
5 & $2237$ & $2731$ & $29.02$ & $34.25$ & $1.28$ & $1.82$\\
 \hline
6 & $2875$ & $2987$ & $45$ & $45.81$ & $3.15$ & $3.2$\\
 \hline
\end{tabular}}
\end{table}
%%%%%%%%%%%%%%%%%%%%%%%%%%%%%%%%%%%%%%%%%%%%%%%%%%%%%%%%%%%%%%%%%%%%%%%%%%%%%%%%%%%%%%%%%%

%%%%%%%%%%%%%%%%%%%%%%%%%%%%%%%%%%%%%%%%%%%%%%%%%%%%%%%%%%%%%%%%%%%%%%%%%%%%%%%%%%%%%%%%%%
\section{Conclusion}
%%%%%%%%%%%%%%%%%%%%%%%%%%%%%%%%%%%%%%%%%%%%%%%%%%%%%%%%%%%%%%%%%%%%%%%%%%%%%%%%%%%%%%%%%%
In this paper, we proposed a distributed stochastic MPC traffic signal control method for an urban network when the traffic model parameters cannot be estimated precisely.  
Assuming that the exogenous in/out-flows and the turning ratios of the downstream traffic flows are random parameters with known expected values and variances, the optimal traffic signal splits of the road links are found: 1) to minimize the expectation of a cost function, which corresponds to some performance indexes of the network; 2) to satisfy all hard constraints on the limitation of traffic signal splits and downstream traffic flows; and 3) to guarantee that the probability of traffic congestion is less than a small number $\epsilon_t$ for all road links.
Microscopic simulation results illustrate that our proposed method deals well with uncertainties in traffic model parameters' estimation.

Taking the advantage of the spatial separability of the urban network, we designed a multiagent framework, in which the considered urban network is divided into multiple subnetworks. The stochastic MPC traffic signal control problem of the overall network is transformed into an union of multiple subproblems with some coupled constraints.
Each subproblem corresponds to one subnetwork and is controlled by one agent.
Based on the proximal ADMM scheme, our proposed method requires every agent to use only its own information and the information of its neighboring agents. In addition, each agent uses only simple arithmetic operators in its local computation work. By sharing the computation load among agents, the execution time of our control methods could be reduced significantly in a distributed manner.
In our future works, we would like to consider variable cycle lengths and the weights of performance indexes in the cost function for further optimizing the traffic signal control problems.
%%%%%%%%%%%%%%%%%%%%%%%%%%%%%%%%%%%%%%%%%%%%%%%%%%%%%%%%%%%%%%%%%%%%%%%%%%%%%%%%%%%%%%%%%%
\section*{acknowledgement}
%%%%%%%%%%%%%%%%%%%%%%%%%%%%%%%%%%%%%%%%%%%%%%%%%%%%%%%%%%%%%%%%%%%%%%%%%%%%%%%%%%%%%%%%%%
This work was supported by the National Research Foundation of Korea (NRF) grant funded by the Korea Government (MSIT) under Grant 2022R1A2B5B03001459.
%%%%%%%%%%%%%%%%%%%%%%%%%%%%%%%%%%%%%%%%%%%%%%%%%%%%%%%%%%%%%%%%%%%%%%%%%%%%%%%%%%%%%%%%%%

%%%%%%%%%%%%%%%%%%%%%%%%%%%%%%%%%%%%%%%%%%%%%%%%%%%%%%%%%%%%%%%%%%%%%%%%%%%%%%%%%%%%%%%%%%
\bibliographystyle{IEEEtran}
\bibliography{mylib}

% Generated by IEEEtran.bst, version: 1.13 (2008/09/30)
\begin{thebibliography}{10}
\providecommand{\url}[1]{#1}
\csname url@samestyle\endcsname
\providecommand{\newblock}{\relax}
\providecommand{\bibinfo}[2]{#2}
\providecommand{\BIBentrySTDinterwordspacing}{\spaceskip=0pt\relax}
\providecommand{\BIBentryALTinterwordstretchfactor}{4}
\providecommand{\BIBentryALTinterwordspacing}{\spaceskip=\fontdimen2\font plus
\BIBentryALTinterwordstretchfactor\fontdimen3\font minus
  \fontdimen4\font\relax}
\providecommand{\BIBforeignlanguage}[2]{{%
\expandafter\ifx\csname l@#1\endcsname\relax
\typeout{** WARNING: IEEEtran.bst: No hyphenation pattern has been}%
\typeout{** loaded for the language `#1'. Using the pattern for}%
\typeout{** the default language instead.}%
\else
\language=\csname l@#1\endcsname
\fi
#2}}
\providecommand{\BIBdecl}{\relax}
\BIBdecl

\bibitem{MarkosPapageorgiou2003}
M.~Papageorgiou, C.~Diakaki, V.~Dinopoulou, A.~Kotsialos, and Y.~Wang, ``Review
  of road traffic control strategies,'' \emph{Proceedings of the IEEE},
  vol.~91, no.~12, pp. 4416--4426, Dec 2003.

\bibitem{LeiChen2016}
L.~Chen and C.~Englund, ``Cooperative intersection management: A survey,''
  \emph{IEEE Trans. on Intelligent Transportation Systems}, vol.~17, no.~2, pp.
  570 -- 586, Feb 2016.

\bibitem{PBHunt1982}
P.~B. Hunt, D.~I. Robertson, R.~D. Bretherton, and M.~C. Royle, ``The scoot
  on-line traffic signal optimization technique,'' \emph{Traffic Engineering \&
  Control}, vol.~23, p. 190–192, 1982.

\bibitem{PRLowrie1982}
P.~R. Lowrie, ``Scats: The sydney co-ordinated adaptive traffic
  system—principles, methodology, algorithms,'' in \emph{Proc. IEE Int. Conf.
  Road Traffic Signalling}, 1982, pp. 67--70.

\bibitem{JLFarges1983}
J.~L. Farges, J.~J. Henry, and J.~Tufal, ``The prodyn real-time traffic
  algorithm,'' in \emph{Proc. 4th IFAC Symp. Trasport. Systems, Baden-Baden,
  Germany}, 1983, p. 307–312.

\bibitem{NHGartner1983}
N.~H. Gartner, ``Opac: A demand-responsive strategy for traffic signal
  control,'' \emph{Transport. Res. Rec.}, no. 906, p. 75–84, 1983.

\bibitem{PMirchandani1998}
P.~Mirchandani and L.~Head, ``Rhodes—a real–time traffic signal control
  system: Architecture, algorithms, and analysis,'' in \emph{TRISTAN III
  (Triennial Symp. Transport. Analysis), vol. 2, San Juan, Puerto Rico, Jun.},
  1998, p. 17–23.

\bibitem{ChristinaDiakaki2002}
C.~Diakaki, M.~Papageorgiou, and K.~Aboudolas, ``A multivariable regulator
  approach to traffic-responsive networkwide signal control,'' \emph{Control
  Engineering Practice}, vol.~10, pp. 183--195, 2002.

\bibitem{KonstantinosAmpountolas2009}
K.~Aboudolas, M.~Papageorgiou, and E.~Kosmatopoulos, ``Store-and-forward based
  methods for the signal control problem in large-scale congested urban road
  networks,'' \emph{Transportation Research Part C: Emerging Technologies},
  vol.~18, no.~5, p. 680–694, 2010.

\bibitem{KonstantinosAmpountolas2010}
------, ``A rolling-horizon quadratic programming approach to the signal
  control problem in large scale congested urban road networks,''
  \emph{Transportation Research Part C: Emerging Technologies}, vol.~17, no.~2,
  p. 163–174, 2009.

\bibitem{ShuLin2011}
S.~Lin, B.~D. Schutter, Y.~Xi, and H.~Hellendoorn, ``Fast model predictive
  control for urban road networks via milp,'' \emph{IEEE Trans. on Intelligent
  Transportation Systems}, vol.~12, no.~3, pp. 846 -- 856, 2011.

\bibitem{LucasBarcelosdeOliveira2010}
L.~B. de~Oliveira and E.~Camponogara, ``Multi-agent model predictive control of
  signaling split in urban traffic networks,'' \emph{Transportation Research
  Part C}, vol.~18, no.~1, pp. 120 -- 139, Feb. 2010.

\bibitem{EduardoCamponogara2011}
E.~Camponogara and H.~F. Scherer, ``Distributed optimization for model
  predictive control of linear dynamic networks with control-input and output
  constraints,'' \emph{IEEE Transactions on Automation Science and
  Engineering}, vol.~8, no.~1, pp. 233 -- 242, 2011.

\bibitem{AndyHFChow2020}
A.~H.~F. Chow, R.~Sha, and Y.~Li, ``Adaptive control strategies for urban
  network traffic via a decentralized approach with user-optimal routing,''
  \emph{IEEE Transactions on Intelligent Transportation Systems}, vol.~21,
  no.~4, pp. 1697 -- 1704, Apr. 2020.

\bibitem{RRNegenborn2008}
R.~R. Negenborn, B.~D. Schutter, and J.~Hellendoorn, ``Multi-agent model
  predictive control for transportation networks: Serial versus parallel
  schemes,'' \emph{J. Optim. Theory Appl}, vol.~21, no.~3, pp. 353 -- 366,
  2008.

\bibitem{SteliosTimotheou2015}
S.~Timotheou, C.~G. Panayiotou, and M.~M. Polycarpou, ``Distributed traffic
  signal control using the cell transmission model via the alternating
  direction method of multipliers,'' \emph{IEEE Trans. on Intelligent
  Transportation Systems}, vol.~16, no.~2, pp. 919 -- 933, Apr 2015.

\bibitem{ZhaoZhou2017}
Z.~Zhou, B.~D. Schutter, S.~Lin, and Y.~Xi, ``Two-level hierarchical
  model-based predictive control for large-scale urban traffic networks,''
  \emph{IEEE Trans. on Intelligent Transportation Systems}, vol.~25, no.~2, pp.
  496 -- 508, 2017.

\bibitem{BaoLinYe2016}
B.~L. Ye, W.~Wu, K.~Ruan, L.~Li, and W.~Mao, ``A hierarchical model predictive
  control approach for signal splits optimization in large-scale urban road
  networks,'' \emph{IEEE Trans. on Intelligent Transportation Systems},
  vol.~17, no.~8, pp. 2182--2192, 2016.

\bibitem{PietroGrandinetti2018}
P.~Grandinetti, C.~C. de~Wit, and F.~Garin, ``Distributed optimal traffic
  lights design for large-scale urban networks,'' \emph{IEEE Trans. on Control
  Systems Technology}, pp. 1 -- 14, 2018.

\bibitem{Heydecker2005}
B.~G. Heydecker and J.~D. Addison, ``Analysis of dynamic traffic equilibrium
  with departure time choice,'' \emph{Transp. Sci}, vol.~39, no.~1, pp. 39--57,
  2005.

\bibitem{BaoLinYe2019}
B.~L. Ye, W.~Wu, K.~Ruan, L.~Li, T.~Chen, H.~Gao, and Y.~Chen, ``A survey of
  model predictive control methods for traffic signal control,'' \emph{IEEE/CAA
  J. of Automatica Sinica}, vol.~3, no.~6, pp. 623--640, 2019.

\bibitem{StephenBoyd2011}
S.~Boyd, N.~Parikh, E.~Chu, B.~Peleato, and J.~Eckstein, ``Distributed
  optimization and statistical learning via the alternating direction method of
  multipliers,'' \emph{Foundations and Trends in Machine Learning}, vol.~3,
  no.~1, pp. 1 -- 122, 2011.

\bibitem{BingshengHe2015}
B.~He and X.~Yuan, ``On non-ergodic convergence rate of douglas–rachford
  alternating direction method of multipliers,'' \emph{Numer. Math.}, vol. 130,
  pp. 567 -- 577, 2015.

\bibitem{XinxinLi2015}
X.~Li and X.~Yuan, ``A proximal strictly contractive peaceman–rachford
  splitting method for convex programming with applications to imaging,''
  \emph{SIAM. J. Imaging Sciences}, vol.~8, no.~2, pp. 1332 -- 1365, 2015.

\bibitem{WeiDeng2017}
W.~Deng, M.-J. Lai, Z.~Peng, and W.~Yin, ``Parallel multi-block admm with
  o(1/k) convergence,'' \emph{J. Sci. Comput.}, vol.~71, pp. 712 -- 736, 2017.

\bibitem{JackReilly2015}
J.~Reilly and A.~M. Bayen, ``Distributed optimization for shared state systems:
  Applications to decentralized freeway control via subnetwork splitting,''
  \emph{IEEE Transactions on Intelligent Transportation Systems}, vol.~16,
  no.~6, pp. 3465 -- 3472, Dec. 2015.

\bibitem{ErminWei2013}
E.~Wei and A.~Ozdaglar, ``On the o(1/k) convergence of asynchronous distributed
  alternating direction method of multipliers,'' \emph{arXiv preprint
  arXiv:1307.8254v1}, 2013.

\bibitem{TamasTettamanti2014}
T.~Tettamanti, T.~Luspay, B.~Kulcsár, T.~Péni, and I.~Varga, ``Robust control
  for urban road traffic networks,'' \emph{IEEE Trans. on Intelligent
  Transportation Systems}, vol.~15, no.~1, pp. 385--398, 2014.

\bibitem{VietHoangPham2020}
V.~Pham, K.~Sakurama, S.~Mou, and H.-S. Ahn, ``Distributed traffic control for
  a large-scale urban network,'' \emph{https://arxiv.org/abs/2005.02007}.

\bibitem{VietHoangPham2020TITS}
------, ``Distributed control for an urban traffic network,'' \emph{Submitted
  for a publication}.

\bibitem{YafengYin2008}
Y.~Yin, ``Robust optimal traffic signal timing,'' \emph{Transportation Research
  Part B: Methodology}, vol.~42, pp. 911--924, 2008.

\bibitem{LihuiZhang2010}
L.~Zhang, Y.~Yin, and Y.~Lou, ``Robust signal timing for arterials under
  day-to-day demand variations,'' \emph{Transportation Research Record: Journal
  of the Transportation Research Board}, vol. 2192, pp. 156--166, 2010.

\bibitem{HaoLiu2022}
H.~Liu, C.~Claudel, R.~Machemehl, and K.~A. Perrine, ``A robust traffic control
  model considering uncertainties in turning ratios,'' \emph{IEEE Trans. on
  Intelligent Transportation Systems}, vol.~23, no.~7, pp. 6539--6555, 2022.

\bibitem{Calafiore2006}
G.~C. Calafiore and L.~E. Ghaoui, ``On distributionally robust
  chanceconstrained linear programs,'' \emph{J. of Optimization Theory and
  Application}, vol. 130, no.~1, pp. 1--22, 2006.

\bibitem{ZCSu2021}
Z.~C. Su, A.~H.~F. Chow, and R.~X. Zhong, ``Adaptive network traffic control
  with an integrated model-based and data-driven approach and a decentralised
  solution method,'' \emph{Transportation Research Part C}, vol. 128, 2021.

\bibitem{RoyDYates2004}
R.~D. Yates and D.~J. Goodman, \emph{Probability and Stochastic Processes: A
  Friendly Introduction for Electrical and Computer Engineers}, 2nd~ed.\hskip
  1em plus 0.5em minus 0.4em\relax Wiley, 2004.

\bibitem{DavidKinderlehrer1980}
D.~Kinderlehrer and G.~Stampacchia, \emph{An Introduction to Variational
  Inequalities and Their Applications}.\hskip 1em plus 0.5em minus 0.4em\relax
  Academic Press, New York, NY, USA, 1980.

\bibitem{LingchenKong2009}
L.~C. Kong, L.~Tuncel, and N.~H. Xiu, ``Clarke generalized jacobian of the
  projection onto symmetric cones,'' \emph{Set-Valued and Variational
  Analysis}, vol.~17, pp. 135 -- 151, 2009.

\bibitem{DenosCGazis1963}
D.~C. Gazis, R.~B. Potts, and T.~J. Watson, ``The oversaturated intersection,''
  in \emph{Proc. 2nd Int. Symp. Traffic Theory, London, U.K.}, 1963, pp. 221 --
  237.

\bibitem{JeanGregoire2015}
J.~Gregoire, X.~Qian, E.~Frazzoli, A.~Fortelle, and T.~Wongpiromsarn,
  ``Capacity-aware backpressure traffic signal control,'' \emph{IEEE
  Transactions on Control of Network Systems}, vol.~2, no.~2, p. 164–173,
  2015.

\end{thebibliography}
%%%%%%%%%%%%%%%%%%%%%%%%%%%%%%%%%%%%%%%%%%%%%%%%%%%%%%%%%%%%%%%%%%%%%%%%%%%%%%%%%%%%%%%%%%
\end{document}